\documentclass[acmsmall,10pt,screen]{acmart}

\acmJournal{PACMPL}
\acmVolume{1}
\acmNumber{CONF} 
\acmArticle{1}
\acmYear{2018}
\acmMonth{1}
\acmDOI{} 
\startPage{1}

\setcopyright{none}

\usepackage[utf8]{inputenc} 
\usepackage{amsmath}
\usepackage{amssymb}
\usepackage[compact]{titlesec}
\usepackage{thmtools}
\usepackage[english]{babel}
\usepackage{pdfpages}
\usepackage{paralist}
\usepackage{graphicx}
\usepackage{subcaption}
\usepackage[labelfont=bf]{caption}
\usepackage{longtable}
\usepackage{booktabs}
\usepackage{tabu}
\usepackage[referable]{threeparttablex}
\usepackage{varwidth}
\usepackage{multirow}
\usepackage{wrapfig}

\usepackage{longtable}

\usepackage{pifont}

\usepackage[ruled, linesnumbered]{algorithm2e}
\usepackage{parskip}

\usepackage{hyperref}
\usepackage{thm-restate}
\usepackage[capitalise]{cleveref}

\crefname{figure}{Figure}{Figure}
\crefformat{footnote}{#2\footnotemark[#1]#3}
\usepackage{tikz}
\usepackage{comment}

\usepackage{todonotes}
\usepackage{tikz}
\usetikzlibrary{arrows,automata,shapes,decorations,decorations.markings,calc, matrix,decorations.pathmorphing, patterns,backgrounds}

\usepackage{bm}
\usepackage{pgfplots}
\usepackage{enumerate,paralist}
\usepackage{pbox}

\usepackage{appendix}
\usepackage{calc}
\usepackage{fp}
\usepackage{multirow}
\usepackage{pbox}

\newtheorem{remark}{Remark}

\DeclareMathAlphabet{\mathpzc}{OT1}{pzc}{m}{it}

\usepackage{graphicx}

\newcommand{\Maz}{\mathcal{HB}}
\newcommand{\MazE}{\sim_{\Maz}}
\newcommand{\VHB}{\mathcal{\mathcal{VHB}}}
\newcommand{\VHBE}{\sim_{\VHB}}
\newcommand{\SysLeafEvents}{\SysEvents_{\neq\RootProcess}}
\newcommand{\LeafEvents}[1]{\SysEvents_{\neq\RootProcess}(#1)}
\newcommand{\Leaves}{\System \setminus\{ \RootProcess \}}
\mathchardef\mhyphen="2D 

\newcommand{\RootProcess}{p_1}
\newcommand{\Poly}{\mathsf{poly}}
\newcommand{\WriteExtend}{\mathsf{WExtend}}

\newcommand{\ov}{\overline}

\newcommand{\Source}{\mathsf{Source}}
\newcommand{\Optimal}{\mathsf{Optimal}}
\newcommand{\OptimalObs}{\mathsf{Optimal}^*}
\newcommand{\ExtendPO}{\mathsf{Extend}}
\newcommand{\Width}[1]{\mathsf{width}(#1)}
\newcommand{\MWidth}[1]{\mathsf{Mwidth}(#1)}
\newcommand{\VCDPOR}{\operatorname{VC-DPOR}}
\newcommand{\DCDPOR}{\operatorname{DC-DPOR}}

\newcommand{\NegativeAnnotation}{\mathcal{C}}

\newcommand{\Trace}{t}

\newcommand{\SysReads}{\mathcal{R}}
\newcommand{\SysWrites}{\mathcal{W}}
\newcommand{\Read}{r}
\newcommand{\TO}{\mathsf{TO}}
\newcommand{\Write}{w}
\newcommand{\Event}{e}

\newcommand{\HB}[3]{#1\mathsf{\to}_{#2}#3}
\newcommand{\CHB}[3]{#1\mathsf{\mapsto}_{#2}#3}
\newcommand{\NCHB}[3]{#1\mathsf{\not \mapsto}_{#2}#3}

\newcommand{\Value}{\mathsf{val}}
\newcommand{\CandidateSet}{M}
\newcommand{\GuardingRead}{\mathsf{Guard}}
\newcommand{\LocalStates}{\mathcal{L}}
\newcommand{\MutateRoot}{\mathsf{ExtendRoot}}
\newcommand{\MutateLeaf}{\mathsf{ExtendLeaf}}

\newcommand{\Unordered}[3]{#1\parallel_{#2} #3}
\newcommand{\Ordered}[3]{#1 \not \parallel_{#2} #3}
\newcommand{\Refines}{\sqsubseteq}
\newcommand{\StrictRefines}{\sqsubset}
\newcommand{\LeafRefines}{\preccurlyeq}
\newcommand{\Confl}[2]{#1 \Join #2}

\newcommand{\Enabled}{\mathsf{enabled}}

\newcommand{\Events}[1]{\SysEvents(#1)}
\newcommand{\Reads}[1]{\SysReads(#1)}
\newcommand{\Writes}[1]{\SysWrites(#1)}

\newcommand{\Project}{|}
\newcommand{\Domain}{\mathsf{dom}}
\newcommand{\ValueDomain}{\mathcal{D}}
\newcommand{\SideDomain}{[2]}
\newcommand{\SideAnnotation}{S}
\newcommand{\AnnotatedPO}{\mathcal{P}}
\newcommand{\AnnotatedPOQ}{\mathcal{Q}}
\newcommand{\AnnotatedPOK}{\mathcal{K}}
\newcommand{\AnnotatedPOF}{\mathcal{F}}

\newcommand{\Image}{\mathsf{img}}
\newcommand{\SeqTrace}{\tau}
\newcommand{\Realize}{\mathsf{Realize}}
\newcommand{\True}{\mathsf{True}}
\newcommand{\False}{\mathsf{False}}

\newcommand{\System}{\mathcal{H}}

\newcommand{\PartialOrders}{\mathcal{A}}
\newcommand{\Process}{p}
\newcommand{\Proc}[1]{\Process(#1)}

\newcommand{\StateSpace}{\mathcal{S}_{\System}}

\newcommand{\Globals}{\mathcal{G}}

\newcommand{\Class}[2]{[#1]_{#2}}

\newcommand{\Location}{\mathsf{loc}}

\newcommand{\SysEvents}{\mathcal{E}}

\newcommand{\State}{s}

\newcommand{\TraceSpace}{\mathcal{T}_{\System}}
\newcommand{\TraceSpaceMax}{\mathcal{T}_{\System}^{\max}}

\newcommand{\Observation}{\mathcal{O}}

\newcommand{\ClosureAlgo}{\mathsf{Closure}}

\newcommand{\SideIndicator}{\mathcal{I}}
\newcommand{\RuleOneAlgo}{\mathsf{Rule1}}
\newcommand{\RuleTwoAlgo}{\mathsf{Rule2}}
\newcommand{\RuleThreeAlgo}{\mathsf{Rule3}}
\newcommand{\Flag}{\mathsf{Flag}}

\newcommand{\GoodWrites}{\mathsf{GoodW}}
\newcommand{\BadWrites}{\mathsf{BadW}}

\newcommand{\VisibleWrites}{\mathsf{VisibleW}}
\newcommand{\HeadWrites}{\mathsf{MinW}}
\newcommand{\TailWrites}{\mathsf{MaxW}}

\newcommand{\bbot}{\rotatebox[origin=c]{90}{$\models$}}

\usetikzlibrary{arrows.meta,calc,decorations.markings,math,arrows.meta}

\preto\tabular{\setcounter{magicrownumbers}{0}}
\newcounter{magicrownumbers}

\pgfdeclarelayer{bg}    
\pgfsetlayers{bg,main}  

\def \darkred {black!20!red}

\SetKw{Continue}{continue}

\SetCommentSty{mycommfont}

\usepackage{booktabs}   
\usepackage{subcaption} 

\begin{document}


\title{Value-centric Dynamic Partial Order Reduction}

\author{Krishnendu Chatterjee}
\affiliation{
  \institution{IST Austria}            
  \streetaddress{Am Campus 1}
  \city{Klosterneuburg}
  \postcode{3400}
  \country{Austria}                    
}
\email{krishnendu.chatterjee@ist.ac.at}          

\author{Andreas Pavlogiannis}
\affiliation{
  \institution{EPFL}            
  \streetaddress{Route Cantonale}
  \city{Lausanne}
  \postcode{1015}
  \country{Switzerland}                    
}
\email{pavlogiannis@cs.au.dk}          

\author{Viktor Toman}
\affiliation{
  \institution{IST Austria}            
  \streetaddress{Am Campus 1}
  \city{Klosterneuburg}
  \postcode{3400}
  \country{Austria}                    
}
\email{viktor.toman@ist.ac.at}          

\begin{abstract}
The verification of concurrent programs remains an open challenge,
as thread interaction has to be accounted for, which leads to state-space explosion.
Stateless model checking battles this problem by exploring traces rather than states of the program.
As there are exponentially many traces, dynamic partial-order reduction (DPOR) techniques
are used to partition the trace space into equivalence classes, and explore a few representatives from each class.
The standard equivalence that underlies most DPOR techniques is the \emph{happens-before} equivalence,
however recent works have spawned a vivid interest towards coarser equivalences.
The efficiency of such approaches is a product of two parameters: 
(i)~the size of the partitioning induced by the equivalence, and
(ii)~the time spent by the exploration algorithm in each class of the partitioning.

In this work, we present a new equivalence, called \emph{value-happens-before} and show that it has two appealing features.
First, value-happens-before is always at least \emph{as coarse as} the happens-before equivalence, and can be even exponentially coarser.
Second, the value-happens-before partitioning is efficiently explorable when the number of threads is bounded.
We present an algorithm called \emph{value-centric} DPOR ($\VCDPOR$), which explores the underlying partitioning using polynomial time per class.
Finally, we perform an experimental evaluation of  $\VCDPOR$ on various benchmarks, and compare it against other state-of-the-art approaches.
Our results show that value-happens-before typically induces a significant reduction in the size of the underlying partitioning,
which leads to a considerable reduction in the running time for exploring the whole partitioning.
\end{abstract}



\maketitle

\section{Introduction}\label{sec:intro}

\noindent{\em Model checking of concurrent programs.} 
The formal analysis of concurrent programs is a key problem in program analysis and verification. 
Concurrency incurs a combinatorial explosion in the behavior of the program, which makes errors hard to reproduce by testing (often identified as {\em Heisenbugs}~\cite{Musuvathi08}).
Thus, the formal analysis of concurrent program requires a systematic exploration of the 
state space, which is addressed by {\em model checking}~\cite{Clarke00}.
However there are two key issues related to model-checking of concurrent programs:
first, is related to the state-space explosion, and second, is related to 
the number of interleavings. 
Below we describe the main techniques to address these problems. 

\noindent{\em Stateless model checking.}
Model checkers typically store a large number of global states, and 
cannot handle realistic concurrent programs. 
The standard solution that is adopted to battle this problem on concurrent programs is 
{\em stateless model checking}~\cite{G96}.
Stateless model-checking methods typically explore traces rather than states of the analyzed program,
and only have to store a small number of traces.
In such techniques, model checking is achieved by a scheduler, which drives the 
program execution based on the current interaction between the threads.
The depth-first nature of the search enables it to be both systematic and memory-efficient. 
Stateless model-checking techniques have been employed successfully in several well-established
tools, e.g., VeriSoft~\cite{Godefroid97,Godefroid05} and {\sc CHESS}~\cite{Musuvathi07b}.

\noindent{\em Partial-order Reduction (POR).}
While stateless model checking deals with the state-space issue, 
one key challenge that remains is exploring efficiently the exponential number of interleavings, which results from non-deterministic interprocess communication.
There exist various techniques for reducing the number of explored interleavings, such as
depth bounding and context bounding~\cite{Lal09,Musuvathi07}.
One of the most well-studied techniques is {\em partial-order reduction (POR)}~\cite{Clarke99,G96,Peled93}.
The main principle of POR is that two interleavings can be regarded as equal
if they agree on the order of conflicting (dependent) events. 
In other words, POR considers certain pairs of traces to be equivalent, and the theoretical 
foundation of POR is the equivalence relation induced on the trace space, known as the {\em happens-before} (or {\em Mazurkiewicz)} equivalence $\Maz$~\cite{Mazurkiewicz87}.
POR algorithms explore at least one trace from each equivalence class and 
guarantee a complete coverage of all behaviors that can occur in any interleaving, 
while exploring only a subset of the trace space. 
For the most interesting properties that arise in formal 
verification, such as safety, race freedom, absence of global deadlocks, and absence of 
assertion violations, POR-based algorithms make sound reports of correctness~\cite{G96}.

\noindent{\em Dynamic Partial-order Reduction (DPOR).}
Dynamic partial-order reduction (DPOR) is an on-the-fly version of POR~\cite{Flanagan05}.
DPOR records conflicts that actually occur during the execution of traces, and thus is able to infer independence more frequently than static POR, which typically relies on over-approximations of conflicting events.
Similar to POR, DPOR-based algorithms guarantee the exploration of at least one trace in each class of the happens-before partitioning.
Recently, an optimal method for DPOR was developed~\cite{Abdulla14} that explores exactly one trace from 
each happens-before equivalence class.
 
\noindent{\em Efficiency of DPOR techniques.} 
The efficiency of DPOR algorithms typically depends on two parameters, namely 
(i)~the size of the trace-space partitioning and
(ii)~the time required to explore each class of the partitioning.
The overall efficiency of the algorithm is a product of the two above, and there is usually 
a trade-off between the two, as coarser partitionings typically make the problem of moving between different classes of the partitioning
computationally harder.

\noindent{\em Beyond the Mazurkiewicz equivalence.}
Lately, there has been a considerable effort into going beyond Mazurkiewicz equivalence, by developing algorithms
that explore partitionings of the trace space induced by equivalence relations that are coarser than Mazurkiewicz~\cite{HUANG15,Chalupa17,Elvira17,Aronis18}.
Such approaches can be broadly classified as {\em oracle-based} methods that rely on NP-hard 
oracles such as SMT-solvers to guide the exploration, and {\em explicit} methods which avoid such computationally expensive oracles.
Explicit methods include DC-DPOR~\cite{Chalupa17} and Optimal DPOR with Observers~\cite{Aronis18},
which rely on equivalences provably coarser than the happens-before equivalence,
as well as Context-sensitive DPOR~\cite{Elvira17} which sometimes might be coarser, but not always.
On the other hand, oracle-based methods include MCR~\cite{HUANG15} and SATCheck~\cite{Demsky15}.

\noindent{\em Value-centric DPOR.}
The happens-before equivalence and most coarser equivalences which admit an efficient exploration
are insensitive to the values that variables take during the execution of a trace.
On the other hand, it is well-understood that equivalences which are sensitive to such values can be very coarse, thereby reducing the size of the trace-space partitioning.
An interesting approach with a value-centric partitioning was recently explored in \cite{HUANG15}.
However, that approach is implicit and relies on expensive NP-oracles repeatedly for guiding the exploration of the partitioning.
This NP bottleneck was identified in that work, and was subsequently only partially improved with static-analysis-based heuristics~\cite{Huang017}.
Hence, the challenge of constructing value-centric equivalences that also admit efficient explorations has remained open.
In the next section, we illustrate the benefits of such equivalences on a small example.

\subsection{A Small Motivating Example}\label{subsec:example}

\begin{figure}
\small
\centering
\begin{subfigure}[t]{0.15\textwidth}
\begin{align*}
\text{Th}&\text{read}~\Process_{1}:\\
\hline\\[-1em]
1.~& \Write(x,1)\\
\end{align*}
\end{subfigure}
\qquad
\begin{subfigure}[t]{0.15\textwidth}
\begin{align*}
\text{Th}&\text{read}~\Process_{2}:\\
\hline\\[-1em]\
1.~& \Write(x,2)\\
2.~& \Write(x,1)\\
3.~& \Read(x)\\
\end{align*}
\end{subfigure}
\vspace{-2em}
\caption{A toy program with two threads.}
\label{fig:motivating_example}
\end{figure}

Consider the simple program given in \cref{fig:motivating_example}, which consists of two threads communicating over a global variable $x$.
We have two types of events: $\Process_1$ writes to $x$ the value~1, whereas $\Process_2$ first writes to $x$ the value~2, then writes to $x$ the value~1,
and finally it reads the value of $x$ to its local variable.
When we analyze this program, it becomes apparent that a model-checking algorithm can benefit if it takes into account the values written by the write events.
Indeed, denote by $\Write_i^j$ the $j$-th write event of thread $i$, and by $\Read$ the unique read event.
There exist $4$ Mazurkiewicz orderings.
\begin{align*}
&\Trace_1:~\Write_1^1 \Write_2^1 \Write_2^2 \Read \qquad
\Trace_2:~\Write_2^1 \Write_1^1 \Write_2^2 \Read \qquad
\Trace_3:~\Write_2^1 \Write_2^2 \Write_1^1 \Read \qquad
\Trace_4:~\Write_2^1 \Write_2^2 \Read \Write_1^1
\end{align*}
Hence, any algorithm that uses the Mazurkiewicz equivalence for exploring the trace space of the above program will have to explore at least~4 traces.
Moreover, any sound algorithm that is insensitive to values will, in general, explore at least two traces (e.g., $\Trace_1$, and $\Trace_3$),
since the value read by $\Read$ can, in principle, be different in both cases.
This is true, for example, for $\DCDPOR$~\cite{Chalupa17}, which is based on the recently introduced Observation equivalence,
as well as the Optimal DPOR with observers~\cite{Aronis18} (which explores $3$).
On the other hand, it is clear that examining a single trace suffices for visiting all the local states of all threads.
Although minimal, the above example illustrates the advantage that stateless model checking algorithms can gain by being sensitive to the values
used by the events during an execution.

\subsection{Challenges and Our Contributions}

\noindent{\em Challenges.}
The above example illustrates that a value-centric partitioning can be coarse.
To our knowledge, value-centric equivalences have been used systematically (i.e., with provable guarantees) only by implicit methods which rely on NP oracles (e.g., SMT-solvers~\cite{HUANG15,Huang017}),
which makes the exploration of the (reduced) partitionings computationally expensive.
The challenge that arises naturally is to produce a partitioning that 
(a)~is provably \emph{always} coarser than Mazurkiewicz trace equivalence;
and (b)~is \emph{efficiently} explorable 
(i.e., the time required in each step of the search is small/polynomial). 
In this work we address this challenge.

\noindent{\em Our contributions.}
The main contribution of this work is a new value-centric equivalence, called  \emph{value-happens-before} ($\VHB$).
Intuitively, $\VHB$ distinguishes (arbitrarily) a thread of the program, called the \emph{root}, from the other threads, called the \emph{leaves}.
The coarsening of the partitioning is achieved by $\VHB$ by relaxing the happens-before orderings between events that belong to the root and leaf threads.
Given two traces $\Trace_1$ and $\Trace_2$ which have the same happens-before ordering on the events of leaf threads,
$\VHB$ deems $\Trace_1$ and $\Trace_2$ equivalent by using a combination of (i)~the \emph{values} and (ii)~the \emph{causally-happens-before} orderings on pairs of events between the root and the leaves.

\noindent{\em Properties of $\VHB$.}
We discuss two key properties of $\VHB$.
\begin{compactenum}
\item \emph{Soundness.}
The $\VHB$ equivalence is sound for reporting correctness of local-state properties.
In particular, if $\Trace_1\VHBE\Trace_2$, then every trace is guaranteed to visit the same local states in both executions.
Thus, in order to report local-state-specific properties (e.g., absence of assertion violations), it is sound to explore a single representative from each class of the underlying partitioning.
Global-state properties can be encoded as local properties by using a thread to monitor the global state.
Due to this fact, many other recent works on DPOR focus on local-state properties only~\cite{HUANG15,Huang017,Aronis18,Chalupa17}.

\item \emph{Exponentially coarser than happens-before.}
The $\VHB$ is always at least as coarse as the happens-before (or Mazurkiewicz) equivalence,
i.e., if two traces are $\Maz$-equivalent, then they are also $\VHB$-equivalent.
This implies that the underlying $\VHB$ partitioning is never larger than the $\Maz$ partitioning.
In addition, we show that there exist programs for which the $\VHB$ partitioning is exponentially smaller,
thereby getting a significant reduction in one of the two factors that affect the efficiency of DPOR algorithms.
Interestingly, this reduction is achieved even if there are \emph{no} concurrent writes in the program.
\end{compactenum}

\noindent{\em Value-centric DPOR.}
We develop an efficient DPOR algorithm that explores the $\VHB$ partitioning, called $\VCDPOR$.
This algorithm is guaranteed to visit every class of the $\VHB$ partitioning, and for a constant number of threads,
the time spent in each class is polynomial.
Hence, $\VCDPOR$ explores efficiently a value-centric partitioning without relying on NP oracles.
For example, in the program of \cref{fig:motivating_example}, $\VCDPOR$ explores only one trace.

\noindent{\em Experimental results.}
Finally, we make a prototype implementation of $\VCDPOR$ and evaluate it on various classes of concurrency benchmarks.
We use our implementation to assess
(i)~the coarseness of the $\VHB$ partitioning in practice, and
(ii)~the efficiency of $\VCDPOR$ to explore such partitionings.
To this end, we compare these two metrics with existing state-of-the-art explicit DPOR algorithms,
namely, the Source-DPOR~\cite{Abdulla14}, Optimal-DPOR~\cite{Abdulla14}, Optimal-DPOR with observers~\cite{Aronis18}, as well as $\DCDPOR$~\cite{Chalupa17}.
Our results show a significant reduction in the size of the partitioning compared to the partitionings explored by existing techniques, which also typically leads to smaller running times.

\section{Preliminaries}\label{sec:prelim}

\subsection{Concurrent Computation Model}\label{subsec:model}

In this section we define the model of concurrent programs and 
introduce general notation.
We follow a standard exposition found in the 
literature~(e.g., \cite{Abdulla14,Chalupa17}).
For simplicity of presentation we do not consider locks in our model.
Later, we remark how locks can be handled naturally by our approach (see \cref{rem:locks}).

\smallskip\noindent{\bf General notation.}
Given a natural number $i\geq 1$, we denote by $[i]$ the set $\{ 1,2,\dots, i \}$.
Given a map $f\colon X\to Y$, we let $\Domain(f)=X$ and $\Image(f)=Y$ denote the domain and image sets of $f$, respectively.
We represent maps $f$ as sets of tuples $\{ (x, f(x))\}_x$.
Given two maps $f_1, f_2$, we write $f_1=f_2$ to denote that $\Domain(f_1)= \Domain(f_2)$ and for every $x\in \Domain(f_1)$ we have $f_1(x)=f_2(x)$, and we write $f_1\neq f_2$ otherwise.
A binary relation $\sim$ on a set $X$ is an {\em equivalence} relation iff $\sim$ is reflexive, symmetric and transitive. 
Given an equivalence $\sim_E$ and some $x\in X$, we denote by $\Class{x}{E}$ the equivalence class of $x$ under $\sim_E$, i.e.,
$
\Class{x}{E}=\{y\in X:~x\sim_E y\}
$.
The \emph{quotient set} $X/{E}= \{\Class{x}{E}\ |\ x\in X\}$ of $X$ under $\sim_E$ is the set of all equivalence classes of $X$ under $\sim_E$.

\smallskip\noindent{\bf Concurrent program.}
We consider a concurrent program $\System=\{ \Process_i \}_{i=1}^k$ of $k$ threads communicating over shared memory, where $k$ is some arbitrary constant.
For simplicity of presentation, we neglect dynamic thread creation.
We distinguish $\Process_1$ as the \emph{root thread} of $\System$, and refer to the remaining threads $\Process_2,\dots, \Process_k$ as \emph{leaf threads}.
The shared memory consists of a finite set $\Globals$ of global variables, where each variable receives values from a finite value domain $\ValueDomain$.
Every thread executes instructions, which we call \emph{events}, and are of the following types.
\begin{compactenum}
\item A \emph{write event} $\Write$ writes a value $v\in \ValueDomain$ to a global variable $x\in \Globals$.
\item A \emph{read event} $\Read$ reads the value $v\in \ValueDomain$ of a global variable $x\in \Globals$.
\item A \emph{local (invisible) event} is an event that does not access any global variable.
\end{compactenum}
Although typically threads contain local events to guide the control-flow, such events are not relevant in our setting, and will thus be ignored.
For simplicity of exposition, we consider that every thread is represented as an unrolled tree, which captures its unrolled control-flow, and every event is a node in this tree.
In practice, each event is sufficiently identified by its thread identifier and an integer that counts how many preceding events of the same thread have been executed already.
Given an event $\Event$, we denote by $\Proc{\Event}$ the thread of $\Event$ and by $\Location(\Event)$ the unique global variable that $\Event$ accesses.
We denote by $\SysEvents$ the set of all events, by $\SysWrites$ the set of write events, and by $\SysReads$ the set of read events of $\System$.
Given a thread $\Process$, we denote by $\SysEvents_{\Process}$, $\SysWrites_{\Process}$ and $\SysReads_{\Process}$ the set of events, read events and write events of $\Process$, respectively.
In addition, we let $\SysEvents_{\neq \Process} = \bigcup_{\Process'\neq \Process} \SysEvents_{\Process'}$ and similarly for $\SysWrites_{\neq\Process}$ and $\SysReads_{\neq\Process}$,
i.e., $\SysEvents_{\neq\Process}$ denote the set of events of threads other than thread $\Process$, and 
similarly, for $\SysWrites_{\neq\Process}$ and $\SysReads_{\neq\Process}$.
Finally, given a set $X\subseteq \SysEvents$, 
we let $\Writes{X}=X\cap \SysWrites$ and $\Reads{X}=X\cap \SysReads$ for the set of write and read events
of $X$, respectively.

\smallskip\noindent{\bf Concurrent program semantics.}
The semantics of $\System$ are defined by means of a transition system over a state space of global states $\State=(\Value, \LocalStates_1,\dots, \LocalStates_k)$,
where $\Value\colon\Globals \to \ValueDomain$ is a \emph{value function} that maps every global variable to a value, and $\LocalStates_i$ is a \emph{local state} of thread $\Process_i$, which contains the values of the local variables of each thread.
The memory model considered here is sequentially consistent.
Since the setting is standard, we omit here the formal setup and refer the reader to \cite{Godefroid05} for details.
As usual in stateless model checking, we focus our attention on state spaces $\StateSpace$ that are acyclic
(hence our focus is on bounded model checking).

\smallskip\noindent{\bf Traces.}
A (concurrent) \emph{trace} is a sequence of events $\Trace=\Event_1,\dots,\Event_j$ that corresponds to a valid execution of $\System$.
Given a trace $\Trace$, we denote by $\Events{\Trace}$ the set of events that appear in $\Trace$,
and by $\Reads{\Trace}=\Events{\Trace}\cap \SysReads$ (resp., $\Writes{\Trace}=\Events{\Trace}\cap \SysWrites$) the read (resp., write) events in $\Trace$.
We let $\Enabled(\Trace)$ denote the set of enabled events in the state reached after $\Trace$ is executed, and call $\Trace$ \emph{maximal} if $\Enabled(\Trace)=\emptyset$.
We write $\TraceSpace$ and $\TraceSpaceMax$ for the set of all traces and maximal traces, respectively, of $\System$.
Given a set of events $A$, we denote by $\Trace \Project A$ the \emph{projection} of $\Trace$ on $A$, which is the unique subsequence of $\Trace$ that contains all events of $A\cap \Events{\Trace}$, and only those.

\smallskip\noindent{\bf Observation, side and value functions.}
Given a trace $\Trace$ and a read event $\Read\in \Reads{\Trace}$, the \emph{observation} of $\Read$ in $\Trace$ is the last write event $\Write$ that appears before $\Read$ in $\Trace$ such that $\Location(\Read)=\Location(\Write)$.
The \emph{observation function} of $\Trace$ is a function $\Observation_{\Trace}\colon\Reads{\Trace}\to \Writes{\Trace}$ such that $\Observation_{\Trace}(\Read)$ is the observation of $\Read$ in $\Trace$.
The \emph{side function} of $\Trace$ is a function $\SideAnnotation_{\Trace}\colon \Reads{\Trace}\cap \SysReads_{\RootProcess}\to [2]$ such that
$\SideAnnotation_{\Trace}(\Read)=1$ if $\Proc{\Observation_{\Trace}(\Read)} = \RootProcess$ and $\SideAnnotation_{\Trace}(\Read)=2$ otherwise.
In other words, a side function is defined for the read events of the root thread, and 
assigns~1 (resp.,~2) to each read event if it observes a local (resp., remote) write event in the trace
\footnote{Although the definition of side functions might appear arbitrary, we rely on this definition later for computing the $\VHB$ abstraction.}.
The \emph{value function} of $\Trace$ is a function $\Value_{\Trace}\colon\Events{\Trace}\to \ValueDomain$ such that $\Value_{\Trace}(\Event)$ is the value of the global variable $\Location(\Event)$ after the prefix of $\Trace$ up to $\Event$ has been executed.
Note that since each thread is deterministic, this value is always unique and thus $\Value_{\Trace}$ is well-defined.

\subsection{Problem and Complexity Parameters}

\smallskip\noindent{\bf The local-state reachability problem.}
The problem we address in this work is detecting erroneous local states of threads, e.g., whether a thread ever encounters an assertion violation.
The underlying algorithmic problem is that of discovering every possible local state of every thread of $\System$, and checking whether a bug occurs.
In stateless model checking, the focal object for this task is the trace, and algorithms solve the problem by exploring different maximal traces of the trace space $\TraceSpaceMax$.
DPOR techniques use an equivalence $E$ to partition the trace space into equivalence classes, and explore the partitioning $\TraceSpaceMax/E$ instead of the whole space $\TraceSpaceMax$.

\smallskip\noindent{\bf Complexity parameters.}
Given an equivalence $E$ over $\TraceSpaceMax$, the efficiency of an algorithm that explores the partitioning $\TraceSpaceMax/E$ is typically a product of two factors $O(\alpha\cdot \beta)$.
The first factor $\alpha$ is the size of the partitioning itself, i.e., $\alpha=|\TraceSpaceMax/E|$, which is typically exponentially large. 
As we construct coarser equivalences $E$, $\alpha$ decreases.
The second factor $\beta$ captures the amortized time on each explored class, and can be either polynomial (i.e., efficient) or exponential.
There is a tradeoff between $\alpha$ and $\beta$:
typically, for coarser equivalences $E$ the algorithms spend more time to explore each class, and hence $\alpha$ is decreased at the cost of increasing $\beta$. 
Hence, the challenge is to make $\alpha$ as small as possible without increasing $\beta$ much.

\smallskip\noindent{\bf This work.}
In this work, we introduce the value-happens before equivalence $\VHB$ and show that the $\VHB$-partitioning is efficiently explorable.
For a constant number of threads, which is typically the case, $\beta=\Poly(n)$, i.e., $\beta$ is polynomial in the length of the longest trace in $\TraceSpaceMax$.
Since, on the other hand, $\alpha$ is usually exponentially large in $n$, we will not focus on establishing the exact dependency of $\beta$ on $n$.
This helps to keep the exposition of the main message clear and focused.

Due to space restrictions, proofs and some experimental details are relegated to the appendix.

\subsection{Partial Orders}\label{subsec:traces_partial_orders}
Here we introduce some useful notation around partial orders, which are the central objects of our algorithms in later sections.

\smallskip\noindent{\bf Partial orders.}
Given a trace $\Trace$ and a set $X\subseteq \Events{\Trace}$, a \emph{(strict) partial order} $P(X)$ over $X$ is an irreflexive, antisymmetric and transitive relation over $X$
(i.e., $<_{P(X)}\subseteq X\times X$).
When $X$ is clear from the context, we will simply write $P$ for the partial order $P(X)$.
Given two events $\Event_1,\Event_2\in X$, we write $\Event_1\leq_P \Event_2$ to denote that $\Event_1<\Event_2$ or $\Event_1=\Event_2$
Given two distinct events $\Event_1,\Event_2\in X$, we say that $\Event_1$ and $\Event_2$ are \emph{unordered} by $P$, denoted by $\Unordered{\Event_1}{P}{\Event_2}$, if neither $\Event_1<_{P}\Event_2$ nor $\Event_2<_{P} \Event_1$.
Given a set $Y\subseteq X$, we denote by $P\Project Y$ the \emph{projection} of $P$ on the set $Y$,
i.e., $<_{P\Project Y}\subseteq Y\times Y$, and for every pair of events $\Event_1, \Event_2\in Y$, we have that $\Event_1<_{P\Project Y} \Event_2$ iff $\Event_1<_{P} \Event_2$.
Given two partial orders $P$ and $Q$ over a common set $X$, we say that $Q$ \emph{refines} $P$, denoted by $Q\Refines P$, if
for every pair of events $\Event_1, \Event_2\in X$, if $\Event_1<_{P}\Event_2$ then $\Event_1<_{Q}\Event_2$.
We write $Q\StrictRefines P$ to denote that $Q\Refines P$ and $P\not\Refines Q$.
A \emph{linearization} of $P$ is a total order that refines $P$.
Note that a trace $\Trace$ is a partial (and, in fact, total) order over the set $\Events{\Trace}$.

\smallskip\noindent{\bf Conflicting events, width and Mazurkiewicz width.}
Two events $\Event_1, \Event_2$ are called \emph{conflicting}, written $\Confl{\Event_1}{\Event_2}$, if they access the same global variable and at least one writes to the variable.
Let $P$ be a partial order over a set $X$.
The \emph{width} $\Width{P}$ of $P$ is the length of its longest antichain,
i.e., it is the smallest integer $i$ such that for every set $Y\subseteq X$ of size $i+1$ $\Event_1, \Event_2\in Y$ such that $\Ordered{\Event_1}{P}{\Event_2}$.
A set $Y\subseteq X$ is called \emph{pairwise conflicting} if for every pair of distinct events $\Event_1,\Event_2\in Y$, we have that $\Confl{\Event_1}{\Event_2}$.
We define the \emph{Mazurkiewicz width} $\MWidth{P}$ of $P$ as the smallest integer $i$ such that 
for every pairwise conflicting set $Y\subseteq X$ of size $i+1$ there exists a pair $\Event_1, \Event_2\in Y$ such that $\Ordered{\Event_1}{P}{\Event_2}$.
Intuitively, $\MWidth{P}$ is similar to $\Width{P}$, with the difference that, in the first case, we focus on events that are conflicting as opposed to any events.

\smallskip\noindent{\bf The thread order $\TO$.}
The \emph{thread order} $\TO$ of $\System$ is a partial order $<_{\TO}\subseteq \SysEvents\times \SysEvents$ that defines a fixed order between pairs of events of the same thread. 
For every trace $\Trace\in \TraceSpace$, we have that $\Trace\Refines \TO\Project \Events{\Trace}$.
Every partial order $P$ used in this work respects the thread order.

\smallskip\noindent{\bf Visible, maximal and minimal writes.}
Consider a partial order $P$ over a set $X$.
Given a read event $\Read\in \Reads{X}$ we define the set of \emph{visible writes}  of $\Read$ as
\begin{align*}
\VisibleWrites_{P}(\Read)=&\{ \Write\in \Writes{X}\colon~\Confl{\Read}{\Write}   \text{ and }   \Read\not <_{P}\Write   \text{ and for each }   \Write'\in \Writes{X}\\
&   \text{ s.t. }  \Confl{\Read}{\Write'}\text{,}
\text{ if }  \Write<_{P} \Write'   \text{ then }  \Write'\not  <_{P} \Read
 \}
 \end{align*}
 In words, $\VisibleWrites_{P}(\Read)$ contains the write events $\Write$ that conflict with $\Read$ and are not ``hidden'' to $\Read$ by $P$, i.e., there exist linearizations $\Trace$ of $P$ such that $\Observation_{\Trace}(\Read)=\Write$ (note that here $\Trace$ is not necessarily an actual trace of $\System$).
 The set of \emph{minimal writes} $\HeadWrites_{P}(\Read)$ (resp., \emph{maximal writes} $\TailWrites_{P}(\Read)$) of $\Read$ contains the write events that are minimal (resp., maximal) elements in $P\Project\VisibleWrites_{P}(\Read)$.

\smallskip\noindent{\bf The happens-before partial order.}
A trace $\Trace$ induces a \emph{happens-before} partial order $\HB{}{\Trace}{}\subseteq \Events{\Trace}\times \Events{\Trace}$,
which is the smallest transitive relation on $\Events{\Trace}$ such that (i)~$\HB{}{\Trace}{}\Refines\TO\Project\Events{\Trace}$ and
(ii)~$\HB{\Event_1}{\Trace}{\Event_2}$ if $\Event_1<_{\Trace}\Event_2$ and $\Confl{\Event_1}{\Event_2}$.

\smallskip\noindent{\bf The causally-happens-before partial order.}
A trace $\Trace$ induces a \emph{causally-happens-before} partial order $\CHB{}{\Trace}{} \subseteq \Events{\Trace}\times \Events{\Trace}$,
which is the smallest transitive relation on $\Events{\Trace}$ such that (i)~$\CHB{}{\Trace}{}\Refines \TO\Project\Events{\Trace}$ and (ii)~for every read event $\Read\in \Reads{\Trace}$, we have $\CHB{\Observation_{\Trace}(\Read)}{\Trace}{\Read}$.
In words, $\CHB{}{}{}$ captures the flow of write events into read events, and is closed under composition with the thread order.
Intuitively, for an event $\Event$, the set of events $\Event'$ that causally-happen-before $\Event$ are the events that need to be present so that $\Event$ is enabled.
Note that $\HB{}{\Trace}{}\Refines \CHB{}{\Trace}{}$, i.e., the happens-before partial order refines the causally-happens-before partial order.

We refer to \cref{fig:trace_example} for an illustration of the $\HB{}{\Trace}{}$ and $\CHB{}{\Trace}{}$ partial orders.
\begin{figure}[!h]
\begin{subfigure}[b]{0.26\textwidth}
\centering
\small
\def\rownumber{}
\begin{tabular}[b]{@{\makebox[1.2em][r]{\rownumber\space}} | l | l | l |}
\normalsize{$\mathbf{\SeqTrace_1}$} & \normalsize{$\mathbf{\SeqTrace_2}$} & \normalsize{$\mathbf{\SeqTrace_3}$}
\gdef\rownumber{\stepcounter{magicrownumbers}\arabic{magicrownumbers}} \\
\hline
$\Write(x,1)$ & & \\
& & $\Write(x,1)$ \\
& $\Write(y,1)$ & \\
& $\Read(y,1)$ & \\
& $\Write(x,1)$ & \\
& & $\Write(y,2)$ \\
$\Write(y,1)$ & & \\
$\Read(x,1)$ &  &\\
\hline
\end{tabular}
\caption{A trace $\Trace$ of three threads.}
\label{subfig:trace_example1}
\end{subfigure}
\quad
\begin{subfigure}[b]{0.7\textwidth}
\centering
\small
\begin{tikzpicture}[thick,
pre/.style={<-,shorten >= 2pt, shorten <=2pt, very thick},
post/.style={->,shorten >= 2pt, shorten <=2pt,  very thick},
seqtrace/.style={->, line width=2},
und/.style={very thick, draw=gray},
event/.style={rectangle, minimum height=0.8mm, minimum width=3mm, fill=black!100,  line width=1pt, inner sep=0},
virt/.style={circle,draw=black!50,fill=black!20, opacity=0}]

\newcommand{\xdisposition}{0}
\newcommand{\ydisposition}{0}
\newcommand{\xstep}{1.6}
\newcommand{\ystep}{0.8}
\newcommand{\ybias}{0.4}
\newcommand{\xbias}{0.7}
\newcommand{\xbiassmall}{0.55}

\node[] at (\xdisposition + 0*\xstep, \ydisposition + 4*\ystep) {
$
\begin{aligned}
\CHB{}{\Trace}{}  = & \TO \Project \Events{\Trace}~\cup & \text{(thread order)}\\
& \{\CHB{\Event_5}{\Trace}{\Event_8} \}  & \text{( on } x \text{)}
\end{aligned}
$
};

\node[] at (\xdisposition + 0.5*\xstep, \ydisposition + 6.5*\ystep) {
$
\begin{aligned}
\HB{}{\Trace}{} =&  \TO \Project \Events{\Trace}~\cup & \text{(thread order)}\\
& \{ \HB{\Event_1}{\Trace}{\Event_2} \HB{}{\Trace}{\Event_5} \HB{}{\Trace}{\Event_8} \}~\cup & \text{( on } x \text{)}\\
& \{\HB{\Event_3}{\Trace}{\Event_4} \HB{}{\Trace}{\Event_6} \HB{}{\Trace}{\Event_7}\} & \text{( on } y \text{)}
\end{aligned}
$
};

\end{tikzpicture}
\caption{The happens-before $\HB{}{\Trace}{}$ and causally-happens-before $\CHB{}{\Trace}{}$ partial orders.}
\label{subfig:trace_example2}
\end{subfigure}
\caption{A trace (\protect\subref{subfig:trace_example1}) and the induced happens-before and causally-happens-before partial orders (\protect\subref{subfig:trace_example2}).
We use the notation $\Event_i$ to refer to the $i$-th event of $\Trace$.
}
\label{fig:trace_example}
\end{figure}

\section{The Value-happens-before Equivalence}\label{sec:ave}

In this section we introduce our new equivalence between traces, called \emph{value-happens-before}, and prove some of its properties.
We start with the happens-before equivalence, which has been used by DPOR algorithms in the literature.

\smallskip\noindent{\bf The happens-before equivalence.}
Two traces $\Trace_1, \Trace_2\in \TraceSpace$ are called \emph{happens-before-equivalent} (commonly referred to as \emph{Mazurkiewicz equivalent}), written $\Trace_1\MazE\Trace_2$, if the following hold.
\begin{compactenum}
\item $\Events{\Trace_1}=\Events{\Trace_2}$, i.e., they consist of the same set of events.
\item $\HB{}{\Trace_1}{} = \HB{}{\Trace_2}{}$, i.e., their happens-before partial orders are equal.
\end{compactenum}

\smallskip\noindent{\bf The value-happens-before equivalence.}
Two traces $\Trace_1,\Trace_2\in \TraceSpace$ are called \emph{value-happens-before-equivalent}, written $\Trace_1\VHBE \Trace_2$, if the following  hold.
\begin{compactenum}
\item $\Events{\Trace_1}=\Events{\Trace_2}$, $\Value_{\Trace_1}=\Value_{\Trace_2}$ and $\SideAnnotation_{\Trace_1}=\SideAnnotation_{\Trace_2}$,
i.e., they consist of the same set of events, and their value functions and side functions are equal.
\item $\CHB{}{\Trace_1}{} \Project \SysReads = \CHB{}{\Trace_2}{} \Project\SysReads$, i.e., $\CHB{}{\Trace_i}{}$ agree on the read events.
\item $\HB{}{\Trace_1}{}\Project{\SysLeafEvents} = \HB{}{\Trace_2}{}\Project{\SysLeafEvents}$, i.e., $\HB{}{\Trace_i}{}$ agree on the events of the leaf threads.
\end{compactenum}

\smallskip
\begin{remark}[Soundness]\label{rem:soundness}
Since every thread of $\System$ is deterministic, for any two traces $\Trace_1, \Trace_2\in \TraceSpace$ such that $\Events{\Trace_1}=\Events{\Trace_2}$ and $\Value_{\Trace_1}=\Value_{\Trace_2}$, the local states of each thread after executing $\Trace_1$ and $\Trace_2$ agree.
It follows that any algorithm that explores every class of $\TraceSpaceMax/\VHB$ discovers every local state of every thread, and thus $\VHB$ is a sound equivalence for local-state reachability.
\end{remark}

\smallskip\noindent{\bf Exponential coarseness.}
Here we  provide two toy examples which illustrate different cases where $\VHB$ can be exponentially coarser than $\Maz$,
i.e., $\TraceSpace/ \Maz$ can have exponentially more classes than $\TraceSpace/\VHB$.

\begin{figure}[]
\small
\centering
\begin{subfigure}[b]{0.45\textwidth}
\begin{subfigure}[t]{0.15\textwidth}
\begin{align*}
\text{Thread}&~\Process_{1}:\\
\hline\\[-1em]
1.~& \Write(x,0)\\
2.~& \Write(x,0)\\
\dots~& \dots\\
n.~& \Write(x,0)
\end{align*}
\end{subfigure}
\quad
\begin{subfigure}[t]{0.15\textwidth}
\begin{align*}
\text{Thread}&~\Process_{2}:\\
\hline\\[-1em]
1.~& \Read(x)\\
2.~& \Read(x)\\
\dots~& \dots\\
n.~& \Read(x)
\end{align*}
\end{subfigure}
\caption{Many operations on one variable.}
\label{subfig:m_comparisona}
\end{subfigure}
\begin{subfigure}[b]{0.45\textwidth}
\begin{subfigure}[t]{0.15\textwidth}
\begin{align*}
\text{Thread}&~\Process_{1}:\\
\hline\\[-1em]
1.~~& \Write(x_1,0)\\
2.~~& \Write(x_1,0)\\
\dots~& \dots\\
2\cdot n-1.~~& \Write(x_n, 0)\\
2\cdot n.~~& \Write(x_n, 0)
\end{align*}
\end{subfigure}
\qquad
\begin{subfigure}[t]{0.15\textwidth}
\begin{align*}
\text{Thread}&~\Process_{2}:\\
\hline\\[-1em]
1.~~& \Read(x_1)\\
2.~~& \Read(x_2)\\
\dots~~& \dots\\
n.~~& \Read(x_n)
\end{align*}
\end{subfigure}
\caption{Few operations on many variables.}
\label{subfig:m_comparisonb}
\end{subfigure}
\caption{Toy programs where $\VHB$ is exponentially coarser than $\Maz$.}
\label{fig:m_comparison}
\end{figure}

\smallskip\noindent{\em Many operations on one variable.}
First, consider the program shown in \cref{subfig:m_comparisona} which consists of two threads $\Process_1$ and $\Process_2$, with $\Process_1$ being the root thread.
This program has a single global variable $x$, and the threads perform operations on $x$ repeatedly.
We assume a salient write event $\Write(x,0)$ that writes the initial value of $x$.
Consider any two traces $\Trace_1, \Trace_2$ that consist of the $i\geq 0$ first $\Write(x)$ events of $\Process_1$ and $j\geq 0$ first $\Read(x)$ events of $\Process_2$ (hence $\Events{\Trace_1}=\Events{\Trace_2}$).
Since each $\Write(x)$ writes the same value, we have $\Value_{\Trace_1}(\Read)=\Value_{\Trace_2}(\Read)$ for every read event $\Read$ in $\Process_2$.
Moreover, since the root thread $\Process_1$ has no read events, we trivially have $\SideAnnotation_{\Trace_1}=\SideAnnotation_{\Trace_2}$. 
Since all read events are on thread $\Process_2$, we have $\CHB{}{\Trace_1}{}\Project \SysReads = \CHB{}{\Trace_2}{}\Project \SysReads = \TO\Project\Reads{\Trace_1}$.
Finally, since we only have one leaf thread, $\HB{}{\Trace_1}{}\Project\SysLeafEvents = \HB{}{\Trace_2}{}\Project\SysLeafEvents = \TO\Project\LeafEvents{\Trace_1}$.
We conclude that $\Trace_1\VHBE \Trace_2$, and thus given $i\geq 0$ and $j\geq 0$ there exists a single class of $\VHBE$ that contains the first $i$ and first $j$ events of $\Process_1$ and $\Process_2$, respectively. Thus
$
|\TraceSpace/\VHB| =  O(n^2).
$
On the other hand, given the first $i\geq 0$ and $j\geq 0$ events of threads $\Process_1$ and $\Process_2$, respectively, 
there exist $\frac{(i+j)!}{i!\cdot j!}=\binom{i+j}{i}$ different ways to order them without violating the thread order.
Observe that every such reordering induces a different happens-before relation.
Using Stirling's approximation, we obtain
\[
|\TraceSpace/\Maz| \geq  \frac{(2\cdot n)!}{(n!)^2} \simeq \frac{\sqrt{2\cdot \pi \cdot 2\cdot n}\cdot \left( 2\cdot n/e \right)^{2\cdot n} }{ \left(\sqrt{2\cdot \pi \cdot n} \cdot (n/e)^{n} \right)^2} = \Omega\left(\frac{4^{ n}}{\sqrt{n}}\right)
\]

\smallskip\noindent{\em Few operations on many variables.}
Now consider the example program shown in \cref{subfig:m_comparisonb} which consists of two threads $\Process_1$ and $\Process_2$, with $\Process_1$ being the root thread.
We assume a salient write event $\Write(x_i,0)$ that writes the initial value of $x_i$.
Consider any two traces $\Trace_1, \Trace_2$ that consist of the $i\geq 0$ first $\Write(x)$ events of $\Process_1$ and $j\geq 0$ first $\Read(x)$ events of $\Process_2$ (hence $\Events{\Trace_1}=\Events{\Trace_2}$).
Since each $\Write(x_i,0)$ writes the same value, we have $\Value_{\Trace_1}(\Read)=\Value_{\Trace_2}(\Read)$ for every read event $\Read$ in $\Process_2$.
Moreover, since the root thread $\Process_1$ has no read events, we trivially have $\SideAnnotation_{\Trace_1}=\SideAnnotation_{\Trace_2}$. 
Since all read events are on thread $\Process_2$, we have $\CHB{}{\Trace_1}{}\Project \SysReads = \CHB{}{\Trace_2}{}\Project \SysReads = \TO\Project\Reads{\Trace_1}$.
Finally, since we only have one leaf thread, $\HB{}{\Trace_1}{}\Project\SysLeafEvents = \HB{}{\Trace_2}{}\Project\SysLeafEvents = \TO\Project\LeafEvents{\Trace_1}$.
We conclude that $\Trace_1\VHBE \Trace_2$, and thus given $i\geq 0$ and $j\geq 0$ there exists a single class of $\VHBE$ that contains the first $i$ and first $j$ events of $\Process_1$ and $\Process_2$, respectively. Thus
$
|\TraceSpace/\VHB| = O(n^2).
$
On the other hand, given the first $i$ read events of $\Process_2$ and $2\cdot i$ write events of $\Process_1$,
there exist at least $2^i$ different observation functions that map each read event $\Read$ to one of the two write events that $\Read$ observes.
Hence $|\TraceSpace/\Maz| = \Omega(2^n)$.

\smallskip
\begin{restatable}{theorem}{themcomparison}\label{them:comparison}
$\VHB$ is sound for local-state reachability. 
Also, $\VHB$ is at least as coarse as $\Maz$, and there exist programs where $\VHB$ is exponentially coarser.
\end{restatable}

We also refer to \cref{sec:comparison} for a comparison of $\VHB$ and our algorithm $\VCDPOR$ which explores the $\VHB$ partitioning
with other related works, in particular with the works of~\cite{Abdulla14}, \cite{HUANG15}, \cite{Chalupa17}, \cite{Elvira17} and \cite{Aronis18}.

\section{Closed Annotated Partial Orders}\label{sec:closure}

In this section we develop the core algorithmic concepts that will be used in the enumerative exploration of the $\VHB$.
We introduce \emph{annotated partial orders}, which are traditional partial orders over events, with additional constraints.
We formulate the question of the realizability of an annotated partial order $\AnnotatedPO$, which asks for a witness trace $\Trace$ that linearizes $\AnnotatedPO$ and satisfies the constraints.
We develop the notion of \emph{closure} of annotated partial orders, and show that
(i)~an annotated partial order is realizable if and only if its closure exists, and
(ii)~deciding whether the closure exists can be done efficiently.
This leads to an efficient procedure for deciding realizability.

\subsection{Annotated Partial Orders}\label{subsec:apo}

Here we introduce the notion of annotated partial orders, which is a central concept of our work.
We build some definitions and notation, and provide some intuition around them.

\smallskip\noindent{\bf Annotated Partial Orders.}
An \emph{annotated partial order} is a tuple $\AnnotatedPO=(X_1, X_2, P, \Value, \SideAnnotation, \GoodWrites)$ where the following hold.
\begin{compactenum}
\item $X_1, X_2$ are sets of events such that $X_1\cap X_2=\emptyset$.
\item $P$ is a partial order over the set $X=X_1 \cup X_2$.
\item $\Value\colon X\to \ValueDomain$ is a value function.
\item $\SideAnnotation\colon \Reads{X_1}\to \SideDomain$ is a side function.
\item $\GoodWrites\colon \Reads{X} \to  2^{\Writes{X}}$ is a good-writes function such that
$\Write\in \GoodWrites(\Read)$ only if $\Confl{\Read}{\Write} \quad \text{and}\quad  \Value(\Read)=\Value(\Write)$ and, if $\Read\in X_1$ then $\Write\in X_{\SideAnnotation(\Read)}$.
\item $\Width{P\Project X_1}= \MWidth{P\Project X_2}=1$.
\end{compactenum}
We let the bad-writes function be $\BadWrites(\Read) = \{\Write\in \Writes{X}\setminus\GoodWrites(\Read):~\Confl{\Read}{\Write} \}$.
We call $\AnnotatedPO$ \emph{consistent} if for every thread $\Process$, we have that $\SeqTrace_{\Process}=\TO\Project (X\cap \SysEvents_{\Process})$ is a
local trace of thread $\Process$ that occurs if every event $\Event$ of $\SeqTrace_{\Process}$ reads/writes the value $\Value(\Event)$.
Hereinafter we only consider consistent annotated partial orders.

\smallskip\noindent{\bf The realizability problem for annotated partial orders.}
Consider an annotated partial order $\AnnotatedPO=(X_1, X_2, P, \Value, \SideAnnotation, \GoodWrites)$.
A trace $\Trace$ is a \emph{linearization} of $\AnnotatedPO$ if 
(i)~$\Trace\Refines P$ and
(ii)~for every read event $\Read\in \Reads{X_1\cup X_2}$ we have that $\Observation_{\Trace}(\Read) \in \GoodWrites(\Read)$.
In words, $\Trace$  must be a linearization of the partial order $P$ with the additional constraint that 
the observation function of $\Trace$ must agree with the good-writes function $\GoodWrites$ of $\AnnotatedPO$. 
We call $\AnnotatedPO$ \emph{realizable} if it has a linearization.
The associated realizability problem takes as input an annotated partial order $\AnnotatedPO$ and asks whether $\AnnotatedPO$ is realizable.

\smallskip
\begin{remark}[Realizability to valid traces.]\label{rem:realizable_valid}
If $\Trace$ is a linearization of some consistent annotated partial order $\AnnotatedPO$ then $\Trace$ is a valid (i.e., actual) trace of $\System$.
This holds because of the following observations.
\begin{compactenum}
\item Since $\Trace$ is a linearization of $\AnnotatedPO$, we have $\Observation_{\Trace}(\Read)\in \GoodWrites(\Read)$ for every read event $\Read\in \Reads{\Trace}$.
\item Due to the previous item and the consistency of $\AnnotatedPO$, for every thread $\Process$ we have that $\SeqTrace_{\Process}=\TO\Project (X\cap \SysEvents_{\Process})$ is a valid local trace of $\Process$.
\end{compactenum}
\end{remark}

\smallskip\noindent{\em Intuition.}
An annotated partial order $\AnnotatedPO$ contains a partial order $P$ over a set $X=X_1\cup X_2$ of events and the value of each event of $X$.
Intuitively, the consistency of $\AnnotatedPO$ states that we obtain the set of events $X$ if we execute each thread and force every read event in this execution to observe the value of a write event according to the good-writes function. 
In the next section, our $\VCDPOR$ algorithm uses annotated partial orders to represent different classes of the $\VHB$ equivalence in order to guide the trace-space exploration.
The set $X_1$ (resp., $X_2$) will contain the events of the root thread (resp., leaf threads).
We will see that if $\VCDPOR$ constructs two annotated partial orders $\AnnotatedPO$ and $\AnnotatedPOQ$ during the exploration, then any two linearizations $\Trace_1$ and $\Trace_2$ of $\AnnotatedPO$ and $\AnnotatedPOQ$, respectively, will satisfy that $\Trace_1\not \VHBE\Trace_2$,
and hence $\AnnotatedPO$ and $\AnnotatedPOQ$ represent different classes of the $\VHB$ partitioning.

\smallskip\noindent{\bf Closed annotated partial orders.}
Consider an annotated partial order $\AnnotatedPO=(X_1, X_2, P, \Value, \SideAnnotation, \GoodWrites)$ and let $X=X_1 \cup X_2$.
We say that $\AnnotatedPO$ is \emph{closed} if the following conditions hold for every read event $\Read\in \Reads{X}$.
\begin{compactenum}
\item\label{item:closure1} There exists a write event $\Write\in \GoodWrites(\Read)\cap \HeadWrites_{P}(\Read)$ such that $\Write<_{P} \Read$.
\item\label{item:closure2} $\TailWrites_{P}(\Read)\cap \GoodWrites(\Read)\neq \emptyset$.
\item\label{item:closure3} For every write event $\Write'\in \BadWrites(\Read)\cap \HeadWrites_{P}(\Read)$ such that $\Write'<_{P}\Read$ there exists a write event $\Write\in \GoodWrites(\Read)\cap \VisibleWrites_{P}(\Read)$ such that $\Write'<_{P} \Write$.
\end{compactenum}
Our motivation behind this definition becomes clear from the following lemma,
which states that closed annotated partial orders are realizable.

\smallskip
\begin{restatable}{lemma}{lemclosedlinearizable}\label{lem:closed_linearizable}
If $\AnnotatedPO$ is closed then it is realizable and a witness can be constructed in $O(\Poly(n))$ time.
\end{restatable}

In particular, the witness trace of $\AnnotatedPO$ is constructed by the following process.
\begin{compactenum}
\item Create a partial order $Q$ as follows. 
\begin{compactenum}
\item For every pair of events $\Event_1, \Event_2$ with $\Event_1<_{P}\Event_2$, we have $\Event_1<_{Q}\Event_2$.
\item For every pair of events $\Event_1, \Event_2$ with $\Event_i\in X_i$ for each $i\in [2]$, if $\Event_2\not <_{P} \Event_1$ then $\Event_1<_{Q}\Event_2$.
\end{compactenum}
\item Create $\Trace$ by linearizing $Q$ arbitrarily.
\end{compactenum}
The above construction is guaranteed to produce a valid witness trace for $\AnnotatedPO$.
The consistency of annotated partial orders guarantees that $\Trace$ is a valid trace of the concurrent program $\System$ (see \cref{rem:realizable_valid}).
We provide an illustration of this construction later in \cref{fig:closure_example}.

We now introduce the notion of \emph{closures}.
Intuitively, the closure of an annotated partial order $\AnnotatedPO$ strengthens $\AnnotatedPO$ by introducing the smallest set of event orderings such that the resulting annotated partial order $\AnnotatedPOQ$ is closed.
The intuition behind the closure is the following:
whenever a rule forces some ordering, any trace that witnesses the realizability of $\AnnotatedPO$ also linearizes $\AnnotatedPOQ$.
In some cases this operation results to cyclic orderings, and thus the closure does not exist.
We also show that obtaining the closure or deciding that it does not exist can be done in polynomial time.
Thus, in combination with \cref{lem:closed_linearizable}, we obtain an efficient algorithm for deciding whether $\AnnotatedPO$ is realizable, by deciding whether it has a closure.

\smallskip\noindent{\bf Closure of annotated partial orders.}
Consider an annotated partial order $\AnnotatedPO=(X_1, X_2, P, \Value, \SideAnnotation, \GoodWrites)$.
We say that an annotated partial order $\AnnotatedPOQ=(X_1, X_2, Q, \Value, \SideAnnotation, \GoodWrites)$ is a \emph{closure} of $\AnnotatedPO$ if
(i)~$Q\Refines P$, 
(ii)~$\AnnotatedPOQ$ is closed, and
(iii)~for any partial order $K$ with $Q \StrictRefines K\Refines P$, we have that the annotated partial order $(X_1, X_2, K, \Value, \SideAnnotation, \GoodWrites)$ is not closed.
As the following lemma states, $\AnnotatedPO$ can have at most one closure. 

\smallskip
\begin{restatable}{lemma}{lemclosuredefined}\label{lem:closure_defined}
There exists at most one weakest partial order $Q$ such that $Q\Refines P$ and $(X_1, X_2, Q, \Value, \SideAnnotation, \GoodWrites)$ is closed.
\end{restatable}

\smallskip\noindent{\bf Feasible annotated partial orders.}
In light of \cref{lem:closure_defined}, we define the \emph{closure} of $\AnnotatedPO$ as the unique annotated partial order $\AnnotatedPOQ$ that is a closure of $\AnnotatedPO$, if such $\AnnotatedPOQ$ exists, and $\bot$ otherwise.
We call $\AnnotatedPO$ \emph{feasible} if its closure is not $\bot$.
We have the following lemma.

\smallskip
\begin{restatable}{lemma}{lemrealizableifffeasible}\label{lem:realizable_iff_feasible}
$\AnnotatedPO$ is realizable if and only if it is feasible.
\end{restatable}

Intuitively, \cref{lem:realizable_iff_feasible} states that the closure rules give the weakest strengthening of $\AnnotatedPO$ that is met by any linearization of $\AnnotatedPO$.
If that strengthening can be made (i.e., $\AnnotatedPO$ is feasible), then $\AnnotatedPO$ has a linearization.
Hence, to decide whether $\AnnotatedPO$ is realizable, it suffices to decide whether it is feasible, by computing its closure.
In the next section we show that this computation can be done efficiently.

\subsection{Computing the Closure}\label{subsec:closure}

We now turn our attention to computing the closure of annotated partial orders, which will provide us with a way of solving the realizability problem.

\smallskip\noindent{\bf Algorithm $\ClosureAlgo$.}
Consider an annotated partial order $\AnnotatedPO=(X_1, X_2, P, \Value, \SideAnnotation, \GoodWrites)$ and let $X=X_1\cup X_2$.
The algorithm $\ClosureAlgo$ either computes the closure of $\AnnotatedPO$, or concludes that $\AnnotatedPO$ is not feasible, and returns $\bot$.
Intuitively, the algorithm maintains a partial order $Q$, initially identical to $P$.
The algorithm iterates over every read event $\Read$ and tests whether $\Read$ violates \cref{item:closure1}, \cref{item:closure2} or \cref{item:closure3} of the definition of closed annotated partial orders.
When it discovers that $\Read$ violates one such closure rule, $\ClosureAlgo$ calls one of the \emph{closure methods} $\RuleOneAlgo(\Read)$, $\RuleTwoAlgo(\Read)$, $\RuleThreeAlgo(\Read)$, for violation of \cref{item:closure1}, \cref{item:closure2} and \cref{item:closure3} of the definition, respectively.
In turn, each of these methods inserts a new ordering $\Event_1\to \Event_2$ in $Q$, with the guarantee that if $\AnnotatedPO$ has a closure $\AnnotatedPOF=(X_1, X_2, F, \Value, \SideAnnotation, \GoodWrites)$, then $\Event_1<_{F}\Event_2$.
Hence, $\Event_1\to \Event_2$ is a \emph{necessary ordering} in the closure of $\AnnotatedPO$.
Finally, when the algorithm discovers that all closure rules are satisfied by every read event in $Q$, it returns the annotated partial order $(X_1, X_2, Q, \Value, \SideAnnotation, \GoodWrites)$, which, due to \cref{lem:closure_defined}, is guaranteed to be the closure of $\AnnotatedPO$.
We refer to \cref{algo:closure} for a formal description.

We now provide some intuition behind each of the closure methods.
Given an event $\Event\in X$, we let $\SideIndicator_{\AnnotatedPO}(\Event)=i$ such that $\Event\in X_i$.
Given two events $\Event_1, \Event_2\in X$, we say that $\Event_2$ is \emph{local} to $\Event_1$ if $\SideIndicator_{\AnnotatedPO}(\Event_1)=\SideIndicator_{\AnnotatedPO}(\Event_2)$, i.e., $\Event_1$ and $\Event_2$ belong to the same set $X_i$.
If $\Event_2$ is not local to $\Event_1$, then it is \emph{remote} to $\Event_1$.
\begin{compactenum}
\item $\RuleOneAlgo(\Read)$. 
This rule is called when \cref{item:closure1} of closure is violated, i.e., there exists no write event $\Write\in \GoodWrites(\Read)\cap\HeadWrites_{Q}(\Read)$ such that $\Write<_{Q} \Read$.
Observe that in this case there is no write event that is (i)~local to $\Read$, (ii)~good for $\Read$ and (iii)~visible to $\Read$.
To make $\Read$ respect this rule, the algorithm finds the first write event $\Write$ that is (i)~good for $\Read$ and (ii)~visible to $\Read$, and orders $\Write\to \Read$ in $Q$.
See \cref{subfig:algorule1} provides an illustration.
\item $\RuleTwoAlgo(\Read)$.
This rule is violated when $\TailWrites_{Q}(\Read)\cap \GoodWrites(\Read)=\emptyset$, i.e., every maximal write event is bad for $\Read$.
To make $\Read$ respect this rule, the algorithm finds the unique maximal write event $\Write$ that is remote to $\Read$ and orders $\Read\to \Write$ in $Q$.
$\RuleTwoAlgo(\Read)$ is called only if $\Read$ does not violate \cref{item:closure1} of closure, which guarantees that $\Write$ exists.
\cref{subfig:algorule2} provides an illustration.
\item $\RuleThreeAlgo(\Read)$.
This rule is violated when there exists a write event $\ov{\Write}\in \BadWrites(\Read)\cap \HeadWrites_{Q}(\Read)$ such that
(i)~$\ov{\Write}<_{Q} \Read$, and 
(ii)~there exists no write event $\Write'\in \GoodWrites(\Read)\cap \VisibleWrites_{Q}(\Read)$ such that $\ov{\Write}<_{Q}\Write'$.
To make $\Read$ respect this rule, the algorithm determines a maximal write event $\Write$ that is (i)~remote to $\ov{\Write}$ and (ii)~a good write for $\Read$, and orders $\ov{\Write}\to \Write$ in $Q$.
$\RuleThreeAlgo(\Read)$ is called only if $\Read$ does not violate either \cref{item:closure1} or \cref{item:closure2} of closure, which guarantees that $\Write$ exists.
\cref{subfig:algorule3} provides an illustration, depending on whether $\ov{\Write}$ is local or remote to $\Read$.
\end{compactenum}

\begin{figure}[!h]
\begin{subfigure}[t]{0.2\textwidth}
\small
\centering
\begin{tikzpicture}[thick,
pre/.style={<-,shorten >= 1pt, shorten <=1pt, thick},
post/.style={->,shorten >= 2pt, shorten <=2pt,  very thick},
seqtrace/.style={->, line width=2},
und/.style={very thick, draw=gray},
event/.style={rectangle, minimum height=0.8mm, minimum width=3mm, fill=black!100,  line width=1pt, inner sep=0},
virt/.style={circle,draw=black!50,fill=black!20, opacity=0}]

\newcommand{\xdisposition}{0}
\newcommand{\ydisposition}{0}
\newcommand{\xstep}{1.3}
\newcommand{\ystep}{0.6}
\newcommand{\ybias}{0.4}
\newcommand{\xbias}{0.4}

\node	[]		(t1a)	at	(\xdisposition + 0*\xstep, \ydisposition + 0*\ystep)	{\normalsize$P\Project X_i$};
\node	[]		(t1b)	at	(\xdisposition + 0*\xstep, \ydisposition + -3*\ystep)	{};
\node	[]		(t2a)	at	(\xdisposition + 1*\xstep, \ydisposition + 0*\ystep)	{\normalsize$P\Project X_{3-i}$};
\node	[]		(t2b)	at	(\xdisposition + 1*\xstep, \ydisposition + -3*\ystep)	{};

\node[event] (w1) at (\xdisposition + 0*\xstep, \ydisposition + -1*\ystep) {};
\node[] (wt1) at (\xdisposition + 0*\xstep - \xbias, \ydisposition + -1*\ystep) {$\Write$};
\node[event] (w2) at (\xdisposition + 0*\xstep, \ydisposition + -2*\ystep) {};
\node[] (wt2) at (\xdisposition + 0*\xstep - \xbias, \ydisposition + -2*\ystep) {$\Write$};

\node[event] (r) at (\xdisposition + 1*\xstep, \ydisposition + -2*\ystep) {};
\node[] (rt) at (\xdisposition + 1*\xstep + \xbias, \ydisposition + -2*\ystep) {$\Read$};

\draw[seqtrace] (t1a) to (t1b);
\draw[seqtrace] (t2a) to (t2b);

\draw[post, \darkred, dashed] (w1) to (r);

\end{tikzpicture}
\caption{$\RuleOneAlgo(\Read)$}
\label{subfig:algorule1}
\end{subfigure}
\qquad
\begin{subfigure}[t]{0.2\textwidth}
\small
\centering
\begin{tikzpicture}[thick,
pre/.style={<-,shorten >= 1pt, shorten <=1pt, thick},
post/.style={->,shorten >= 2pt, shorten <=2pt,  very thick},
seqtrace/.style={->, line width=2},
und/.style={very thick, draw=gray},
event/.style={rectangle, minimum height=0.8mm, minimum width=3mm, fill=black!100,  line width=1pt, inner sep=0},
virt/.style={circle,draw=black!50,fill=black!20, opacity=0}]

\newcommand{\xdisposition}{0}
\newcommand{\ydisposition}{0}
\newcommand{\xstep}{1.3}
\newcommand{\ystep}{0.6}
\newcommand{\ybias}{0.4}
\newcommand{\xbias}{0.4}

\node	[]		(t1a)	at	(\xdisposition + 0*\xstep, \ydisposition + 0*\ystep)	{\normalsize$P\Project X_i$};
\node	[]		(t1b)	at	(\xdisposition + 0*\xstep, \ydisposition + -3*\ystep)	{};
\node	[]		(t2a)	at	(\xdisposition + 1*\xstep, \ydisposition + 0*\ystep)	{\normalsize$P\Project X_{3-i}$};
\node	[]		(t2b)	at	(\xdisposition + 1*\xstep, \ydisposition + -3*\ystep)	{};

\node[event] (w2) at (\xdisposition + 0*\xstep, \ydisposition + -2*\ystep) {};
\node[] (wt2) at (\xdisposition + 0*\xstep - \xbias, \ydisposition + -2*\ystep) {$\ov{\Write}$};

\node[event] (r) at (\xdisposition + 1*\xstep, \ydisposition + -2*\ystep) {};
\node[] (rt) at (\xdisposition + 1*\xstep + \xbias, \ydisposition + -2*\ystep) {$\Read$};

\draw[seqtrace] (t1a) to (t1b);
\draw[seqtrace] (t2a) to (t2b);

\draw[post, \darkred, dashed] (r) to (w2);

\end{tikzpicture}
\caption{$\RuleTwoAlgo(\Read)$}
\label{subfig:algorule2}
\end{subfigure}
\qquad
\begin{subfigure}[t]{0.45\textwidth}
\small
\centering
\begin{tikzpicture}[thick,
pre/.style={<-,shorten >= 1pt, shorten <=1pt, thick},
post/.style={->,shorten >= 2pt, shorten <=2pt,  very thick},
seqtrace/.style={->, line width=2},
und/.style={very thick, draw=gray},
event/.style={rectangle, minimum height=0.8mm, minimum width=3mm, fill=black!100,  line width=1pt, inner sep=0},
virt/.style={circle,draw=black!50,fill=black!20, opacity=0}]

\newcommand{\xdisposition}{0}
\newcommand{\ydisposition}{0}
\newcommand{\xstep}{1.3}
\newcommand{\ystep}{0.6}
\newcommand{\ybias}{0.4}
\newcommand{\xbias}{0.4}

\node	[]		(t1a)	at	(\xdisposition + 0*\xstep, \ydisposition + 0*\ystep)	{\normalsize$P\Project X_i$};
\node	[]		(t1b)	at	(\xdisposition + 0*\xstep, \ydisposition + -3*\ystep)	{};
\node	[]		(t2a)	at	(\xdisposition + 1*\xstep, \ydisposition + 0*\ystep)	{\normalsize$P\Project X_{3-i}$};
\node	[]		(t2b)	at	(\xdisposition + 1*\xstep, \ydisposition + -3*\ystep)	{};

\node[event] (w1) at (\xdisposition + 0*\xstep, \ydisposition + -1*\ystep) {};
\node[] (wt1) at (\xdisposition + 0*\xstep - \xbias, \ydisposition + -1*\ystep) {$\Write$};
\node[event] (w2) at (\xdisposition + 0*\xstep, \ydisposition + -2*\ystep) {};
\node[] (wt2) at (\xdisposition + 0*\xstep - \xbias, \ydisposition + -2*\ystep) {$\Write$};

\node[event] (wp) at (\xdisposition + 1*\xstep, \ydisposition + -1*\ystep) {};
\node[] (wpt) at (\xdisposition + 1*\xstep + \xbias, \ydisposition + -1*\ystep) {$\ov{\Write}$};
\node[event] (r) at (\xdisposition + 1*\xstep, \ydisposition + -2*\ystep) {};
\node[] (rt) at (\xdisposition + 1*\xstep + \xbias, \ydisposition + -2*\ystep) {$\Read$};

\draw[seqtrace] (t1a) to (t1b);
\draw[seqtrace] (t2a) to (t2b);

\draw[post, \darkred, dashed] (wp) to (w2);

\renewcommand{\xdisposition}{3}

\node	[]		(at1a)	at	(\xdisposition + 0*\xstep, \ydisposition + 0*\ystep)	{\normalsize$P\Project X_i$};
\node	[]		(at1b)	at	(\xdisposition + 0*\xstep, \ydisposition + -3*\ystep)	{};
\node	[]		(at2a)	at	(\xdisposition + 1*\xstep, \ydisposition + 0*\ystep)	{\normalsize$P\Project X_{3-i}$};
\node	[]		(at2b)	at	(\xdisposition + 1*\xstep, \ydisposition + -3*\ystep)	{};

\node[event] (aw1) at (\xdisposition + 0*\xstep, \ydisposition + -1*\ystep) {};
\node[] (awt1) at (\xdisposition + 0*\xstep - \xbias, \ydisposition + -1*\ystep) {$\ov{\Write}$};

\node[event] (awp) at (\xdisposition + 1*\xstep, \ydisposition + -1*\ystep) {};
\node[] (awpt) at (\xdisposition + 1*\xstep + \xbias, \ydisposition + -1*\ystep) {$\Write$};
\node[event] (ar) at (\xdisposition + 1*\xstep, \ydisposition + -2*\ystep) {};
\node[] (art) at (\xdisposition + 1*\xstep + \xbias, \ydisposition + -2*\ystep) {$\Read$};

\draw[seqtrace] (at1a) to (at1b);
\draw[seqtrace] (at2a) to (at2b);

\draw[post] (aw1) to (ar);
\draw[post, \darkred, dashed] (aw1) to (awp);

\end{tikzpicture}
\caption{$\RuleThreeAlgo(\Read)$}
\label{subfig:algorule3}
\end{subfigure}
\caption{
Illustration of the three closure operations $\RuleOneAlgo(\Read)$ (\protect\subref{subfig:algorule1}), $\RuleTwoAlgo(\Read)$ (\protect\subref{subfig:algorule2}) and $\RuleThreeAlgo(\Read)$ (\protect\subref{subfig:algorule3}).
We follow the convention that barred and unbarred write events ($\ov{\Write}$ and $\Write$) are bad writes and good writes for $\Read$, respectively.
In each case, the dashed edge shows the new order introduced by the algorithm in $Q$.
}
\label{fig:algo_closure}
\end{figure}
\begin{algorithm*}
\small
\SetInd{0.4em}{0.4em}
\DontPrintSemicolon
\caption{$\ClosureAlgo(\AnnotatedPO)$}\label{algo:closure}
\KwIn{An annotated partial order $\AnnotatedPO=(X_1, X_2, P, \Value, \SideAnnotation, \GoodWrites)$.}
\KwOut{The closure of $\AnnotatedPO$ if it exists, else $\bot$.}
\BlankLine
$Q\gets P$\tcp*[f]{We will strengthen $Q$ during the closure computation}\\
$\Flag\gets \True$\\
\While{$\Flag$}{
$\Flag\gets \False$\\
\ForEach(\tcp*[f]{Iterate over the reads}){$\Read\in \Reads{X_1\cup X_2}$}{
\uIf{$\Read$ violates \cref{item:closure1} of closure}{\label{line:closure_r1}
Call $\RuleOneAlgo(\Read)$\tcp*[f]{Strengthen $Q$ to remove violation}\\
$\Flag\gets \True$\tcp*[f]{Repeat as new violations might have appeared}\\
}
\uIf{$\Read$ violates \cref{item:closure2} of closure}{\label{line:closure_r2}
Call $\RuleTwoAlgo(\Read)$\tcp*[f]{Strengthen $Q$ to remove violation}\\
$\Flag\gets \True$\tcp*[f]{Repeat as new violations might have appeared}\\
}
\uIf{$\Read$ violates \cref{item:closure3} of closure}{\label{line:closure_r3}
Call $\RuleThreeAlgo(\Read)$\tcp*[f]{Strengthen $Q$ to remove violation}\\
$\Flag\gets \True$\tcp*[f]{Repeat as new violations might have appeared}\\
}
}
}
\Return{$(X_1, X_2, Q, \Value, \SideAnnotation, \GoodWrites)$} \tcp*[f]{The closure of $\AnnotatedPO$}
\end{algorithm*}
\begin{algorithm}
\small
\SetInd{0.4em}{0.4em}
\DontPrintSemicolon
\caption{$\RuleOneAlgo(\Read)$}\label{algo:rule1}
\BlankLine
$Y\gets \GoodWrites(\Read)\cap \VisibleWrites_{Q}(\Read)$\\
\lIf{$Y=\emptyset$}{
\Return{$\bot$}
}
$\Write\gets \min_{Q}(Y)$\tcp*{Since Rule~1 is violated, $\min_{Q}(Y)$ is unique}\label{line:ruleonealgo_write}
Insert $\Write\to \Read$ in $Q$
\end{algorithm}
\begin{algorithm}
\small
\SetInd{0.4em}{0.4em}
\DontPrintSemicolon
\caption{$\RuleTwoAlgo(\Read)$}\label{algo:rule2}
\BlankLine
$\Write\gets $ the unique event in $\TailWrites_{Q}(\Read)\cap X_{3-\SideIndicator_{\AnnotatedPO}(\Read)}$\tcp*{$\Write$ exists since \cref{item:closure1} of closure holds}\label{line:ruletwoalgo_write}
Insert $\Read\to \Write$ in $Q$
\end{algorithm}
\begin{algorithm}
\small
\SetInd{0.4em}{0.4em}
\DontPrintSemicolon
\caption{$\RuleThreeAlgo(\Read)$}\label{algo:rule3}
\BlankLine
$\ov{\Write}\gets$ the unique event in $\HeadWrites_{Q}(\Read)\cap \BadWrites(\Read)$\tcp*{$\ov{\Write}$ exists since \cref{item:closure1,item:closure2} of closure hold}\label{line:rulethreealgo_write}
$\Write\gets $ the unique event in $\TailWrites_{Q}(\Read)\cap X_{3-\SideIndicator_{\AnnotatedPO}(\ov{\Write})}$\tcp*{$\Write$ exists since \cref{item:closure1,item:closure2} of closure hold}\label{line:rulethreealgo_writep}
Insert $\ov{\Write}\to \Write$ in $Q$
\end{algorithm}

We have the following lemma regarding the correctness and complexity of $\ClosureAlgo$.

\smallskip
\begin{restatable}{lemma}{closure}\label{lem:closure}
$\ClosureAlgo$ correctly computes the closure of $\AnnotatedPO$ and requires $O(\Poly(n))$ time.
\end{restatable}

\subsection{Realizing Annotated Partial Orders}\label{subsec:realize}

Finally, we address the question of realizability of annotated partial orders.
\cref{lem:realizable_iff_feasible} implies that in order to decide whether an annotated partial order is realizable, it suffices to compute its closure,
and \cref{lem:closure} states that the closure can be computed efficiently.
Together, these two lemmas yield a simple algorithm for solving the realizability problem.

\smallskip\noindent{\bf Algorithm $\Realize$.}
We describe a simple algorithm $\Realize$ that decides whether an annotated partial order $\AnnotatedPO$ is realizable.
The algorithms runs in two steps.
\begin{compactenum}
\item Use \cref{lem:closure} to compute the closure of $\AnnotatedPO$. If the closure is $\bot$, report that $\AnnotatedPO$ is not realizable.
Otherwise, the closure is an annotated partial order $\AnnotatedPOQ$.
\item Use \cref{lem:closed_linearizable} to obtain a witness trace $\Trace$ that linearizes $\AnnotatedPOQ$.
Report that $\AnnotatedPO$ is linearizable, and $\Trace$ is the witness trace.
\end{compactenum}

We conclude the results of this section with the following theorem.

\smallskip
\begin{restatable}{theorem}{themclosure}\label{them:closure}
Let $\AnnotatedPO$ be an annotated partial order of $n$ events.
Deciding whether $\AnnotatedPO$ is realizable requires $O(\Poly(n))$ time.
If $\AnnotatedPO$ is realizable, a witness trace can be produced in $O(\Poly(n))$ time.
\end{restatable}

\smallskip\noindent{\bf Example on the realizability of annotated partial orders.}
We illustrate $\Realize$ on a simple example in \cref{fig:closure_example} with an annotated partial order $\AnnotatedPO=(X_1, X_2, P, \Value, \SideAnnotation, \GoodWrites)$, which we assume to be consistent.
We have a concurrent program $\System$ of two threads.
To represent $\AnnotatedPO$, we make the following conventions.
We have three global variables $x$, $y$, $z$, and a unique read event per variable.
Event subscripts denote the variable accessed by the corresponding event.
For each variable, we have a unique read event, and barred and unbarred events denote the good and bad write events, respectively, 
for that read event.
Since we have specified the good-writes for each read event, the value function $\Value$ is not important for this example.
Note also that $\SideAnnotation(\Read_x)=2$ (resp., $\SideAnnotation(\Read_z)=1$) since the good writes of $\Read_x$ (resp., $\Read_z$) are remote (resp., local) to  the read event.
The partial order $P$ of $\AnnotatedPO$ consists of the thread orders of each thread, shown in solid lines in \protect\cref{subfig:closure_example}.
The dashed edges of \protect\cref{subfig:closure_example} show the strengthening of $P$ performed by the algorithm $\ClosureAlgo$ (\cref{algo:closure}).
The numbers above the dashed edges denote both the order in which these orderings are added and the closure rule that is responsible for the corresponding ordering. In particular, algorithm $\ClosureAlgo$ performs the following steps.
\begin{compactenum}
\item Initially there are no dashed edges, and $\Read_x$ violates \cref{item:closure1} of closure, as there is no good write event for $\Read_x$ that is ordered before $\Read_x$.
$\RuleOneAlgo$ inserts an ordering $\Write_x\to \Read_x$ (dashed edge $1$).
\item After the previous step, $\Read_y$ violates \cref{item:closure2} of closure, as at this point, $\Read_y$ has only one maximal write event $\ov{\Write}_y$, which is bad for $\Read_y$.
$\RuleTwoAlgo$ inserts an ordering $\Read_y\to \ov{\Write}_y$ (dashed edge $2$).
\item After the previous step, $\Read_z$ violates \cref{item:closure3} of closure, as at this point, $\Read_z$ has a bad minimal write event $\ov{\Write}_z$ that is ordered before $\Read_z$ but not before any good write event.
$\RuleThreeAlgo$ inserts an ordering $\ov{\Write}_z\to \Write_z$ (dashed edge $3$).
\end{compactenum}
At this point no closure rule is violated, and $\ClosureAlgo$ returns the closure $\AnnotatedPOQ=(X_1, X_2, Q, \Value, \SideAnnotation, \GoodWrites)$ of $\AnnotatedPO$ where $P$ has been strengthened to $Q$ with the dashed edges.
Observe that $Q$ has Mazurkiewicz width $2$ (and not $1$), as there still exist pairs of conflicting events that are unordered,
both on variable $y$ and variable $z$.
For example, there exist two write events on variable $y$ that are unordered, and hence there exist some linearizations that are ``bad'' in the sense that the read event $\Read_y$ does not observe the good write event $\Write_y$.
Nevertheless, \cref{lem:closed_linearizable} guarantees that the corresponding annotated partial order is linearizable to a valid trace,
which is shown in \protect\cref{subfig:linearization_example}
We make two final remarks for this example.
\begin{compactenum}
\item Not every linearization of $Q$ produces a valid witness trace for the realizability of $\AnnotatedPOQ$, as some linearizations violate the additional constraints that every read event must observe a write event that is good for the read event.
Hence, the challenge is to find a correct witness.
\item $\AnnotatedPOQ$ has more than one witness of realizability. 
\protect\cref{subfig:linearization_example} shows one such witness $\Trace$, as constructed by \cref{lem:closed_linearizable}.
It is easy to verify that $\Trace$ is a valid witness.
Due to \cref{rem:realizable_valid}, the consistency of $\AnnotatedPO$ guarantees that $\Trace$ is a valid trace of the program $\System$.
\end{compactenum}
\begin{figure}[!h]
\small
\centering
\begin{subfigure}[b]{0.6\textwidth}
\centering
\begin{tikzpicture}[thick,
pre/.style={<-,shorten >= 1pt, shorten <=1pt, thick},
post/.style={->,shorten >= 2pt, shorten <=2pt,  very thick},
seqtrace/.style={->, line width=2},
und/.style={very thick, draw=gray},
event/.style={rectangle, minimum height=0.8mm, minimum width=3mm, fill=black!100,  line width=1pt, inner sep=0},
virt/.style={circle,draw=black!50,fill=black!20, opacity=0}]

\newcommand{\xdisposition}{0}
\newcommand{\ydisposition}{0}
\newcommand{\xstep}{1.8}
\newcommand{\ystep}{0.6}
\newcommand{\xbias}{0.4}

\node	[]		(t1a)	at	(\xdisposition + 0*\xstep, \ydisposition + 0*\ystep)	{\normalsize$P\Project X_1$};
\node	[]		(t1b)	at	(\xdisposition + 0*\xstep, \ydisposition + -6*\ystep)	{};
\node	[]		(t2a)	at	(\xdisposition + 1*\xstep, \ydisposition + 0*\ystep)	{\normalsize$P\Project X_{2}$};
\node	[]		(t2b)	at	(\xdisposition + 1*\xstep, \ydisposition + -6*\ystep)	{};
\draw[seqtrace] (t1a) to (t1b);
\draw[seqtrace] (t2a) to (t2b);

\node[event, draw=black, fill=black] (1) at (\xdisposition + 0*\xstep, \ydisposition + -1*\ystep) {};
\node[] (t1) at (\xdisposition + 0*\xstep - \xbias, \ydisposition + -1*\ystep) {\textcolor{black}{$\ov{\Write}_y$}};
\node[event, draw=black, fill=black] (2) at (\xdisposition + 0*\xstep, \ydisposition + -2*\ystep) {};
\node[] (t2) at (\xdisposition + 0*\xstep - \xbias, \ydisposition + -2*\ystep) {\textcolor{black}{$\Read_x$}};
\node[event, draw=black, fill=black] (3) at (\xdisposition + 0*\xstep, \ydisposition + -3*\ystep) {};
\node[] (t3) at (\xdisposition + 0*\xstep - \xbias, \ydisposition + -3*\ystep) {\textcolor{black}{$\Write_z$}};
\node[event, draw=black, fill=black] (4) at (\xdisposition + 0*\xstep, \ydisposition + -4*\ystep) {};
\node[] (rt4) at (\xdisposition + 0*\xstep - \xbias, \ydisposition + -4*\ystep) {\textcolor{black}{$\ov{\Write}_y$}};
\node[event, draw=black, fill=black] (5) at (\xdisposition + 0*\xstep, \ydisposition + -5*\ystep) {};
\node[] (t5) at (\xdisposition + 0*\xstep - \xbias, \ydisposition + -5*\ystep) {\textcolor{black}{$\Read_z$}};

\node[event, draw=black, fill=black] (6) at (\xdisposition + 1*\xstep, \ydisposition + -1*\ystep) {};
\node[] (t6) at (\xdisposition + 1*\xstep + \xbias, \ydisposition + -1*\ystep) {\textcolor{black}{$\Write_y$}};
\node[event, draw=black, fill=black] (7) at (\xdisposition + 1*\xstep, \ydisposition + -2*\ystep) {};
\node[] (t7) at (\xdisposition + 1*\xstep + \xbias, \ydisposition + -2*\ystep) {\textcolor{black}{$\Write_x$}};
\node[event, draw=black, fill=black] (8) at (\xdisposition + 1*\xstep, \ydisposition + -3*\ystep) {};
\node[] (t8) at (\xdisposition + 1*\xstep + \xbias, \ydisposition + -3*\ystep) {\textcolor{black}{$\ov{\Write}_z$}};
\node[event, draw=black, fill=black] (9) at (\xdisposition + 1*\xstep, \ydisposition + -4*\ystep) {};
\node[] (t9) at (\xdisposition + 1*\xstep + \xbias, \ydisposition + -4*\ystep) {\textcolor{black}{$\Read_y$}};
\node[event, draw=black, fill=black] (10) at (\xdisposition + 1*\xstep, \ydisposition + -5*\ystep) {};
\node[] (t10) at (\xdisposition + 1*\xstep + \xbias, \ydisposition + -5*\ystep) {\textcolor{black}{$\ov{\Write}_z$}};

\draw[post, \darkred, dashed] (7) to node[above]{$1$} (2);
\draw[post, \darkred, dashed] (9) to node[above]{$2$} (4);
\draw[post, \darkred, dashed] (8) to node[above]{$3$} (3);

\end{tikzpicture}
\caption{An annotated partial order $\AnnotatedPO$ and its closure (dashed edges).}
\label{subfig:closure_example}
\end{subfigure}
\quad
\begin{subfigure}[b]{0.35\textwidth}
\centering
\small
\def\rownumber{}
\begin{tabular}[b]{@{\makebox[1.2em][r]{\rownumber\space}} | l | l |}
\normalsize{$\mathbf{\SeqTrace_1}$} & \normalsize{$\mathbf{\SeqTrace_2}$}
\gdef\rownumber{\stepcounter{magicrownumbers}\arabic{magicrownumbers}} \\
\hline
$\ov{\Write}_y$ & \\
& $\Write_y$ \\
& $\Write_x$ \\
$\Read_x$ & \\
& $\ov{\Write}_z$\\
$\Write_z$ & \\
& $\Read_y$ \\
$\ov{\Write}_y$ & \\
$\Read_z$ &\\
& $\ov{\Write}_z$\\
\hline
\end{tabular}
\caption{A witness trace that linearizes $\AnnotatedPO$.}
\label{subfig:linearization_example}
\end{subfigure}
\caption{
\protect\cref{subfig:closure_example} shows an annotated partial order $\AnnotatedPO$ on a concurrent program of two threads. 
Subscripts denote the variable accessed by each event.
For each variable, we have a unique read event, and barred and unbarred events denote the good and bad write events, respectively, for that read event.
Dashed edges are added by $\ClosureAlgo$~(\cref{algo:closure}) during closure.
\protect\cref{subfig:linearization_example} shows a witness trace that linearizes $\AnnotatedPO$.
}
\label{fig:closure_example}
\end{figure}

\section{Value-centric Dynamic Partial Order Reduction}\label{sec:vcdpor}

We now present our algorithm $\VCDPOR$ for exploring the partitioning $\TraceSpaceMax/\VHB$.
Intuitively, the algorithm manipulates annotated partial orders of the form $\AnnotatedPO=(X_1, X_2, P, \Value, \SideAnnotation, \GoodWrites)$, where $X_1\subseteq \SysEvents_{\RootProcess}$ and $X_2\subseteq \SysEvents_{\neq \RootProcess}$, 
i.e., $X_1$ (resp., $X_2$) contains events of the root thread (resp., leaf threads).
We first introduce some useful concepts and then proceed with the main algorithm.

\smallskip\noindent{\bf Trace extensions and inevitable sets.}
Given a trace $\Trace$, an \emph{extension} of $\Trace$ is a trace $\Trace'$ such that $\Trace$ is a prefix of $\Trace'$.
We say that $\Trace'$ is a \emph{maximal extension} of $\Trace$ if $\Trace'$ is an extension of $\Trace$ and $\Trace'$ is maximal.
A set of events $X$ is \emph{inevitable} for $\Trace$ if for every maximal extension $\Trace'$ of $\Trace$ we have $X\in \Events{\Trace'}$.
A \emph{write extension} of $\Trace$, denoted by $\WriteExtend(\Trace)$, is any arbitrary largest extension $\Trace'$ of $\Trace$ such that
$\Events{\Trace'}\setminus\Events{\Trace}\subseteq \SysWrites$.
In words, we obtain each $\Trace'$ by extending $\Trace$ arbitrarily until (but not included) the next read event of each thread.
Note that for every such write extension $\Trace'$ of $\Trace$, for every thread $\Process$, the local trace
$\Trace'\Project\Events{\Process}$ is unique, and the set $\Events{\Trace'}$ is inevitable for $\Trace$.
Let $\AnnotatedPO$ be a closed annotated partial order over a set $X$.
A set of events $Y$ is \emph{inevitable} for $\AnnotatedPO$ if for every linearization $\Trace$ of $\AnnotatedPO$ and every maximal extension $\Trace'$ of $\Trace$, we have that $Y\subseteq \Events{\Trace'}$.

\smallskip\noindent{\bf Leaf refinement and minimal annotated partial orders.}
Consider  two partial orders $P$, $Q$ over a set $X$.
We say that $Q$ \emph{leaf-refines} $P$, denoted by $Q\LeafRefines P$ if for every pair of events $\Event_1, \Event_2\in X\cap \SysLeafEvents$,
if $\Confl{\Event_1}{\Event_2}$ and $\Event_1<_{P} \Event_2$ then $\Event_1<_{Q} \Event_2$.
In words, $Q$ leaf-refines $P$ if $Q$ agrees with $P$ on the order of every pair of conflicting events that belong to leaf threads.
Consider an annotated partial order $\AnnotatedPO=(X_1, X_2, P, \Value, \SideAnnotation, \GoodWrites)$.
We call $\AnnotatedPO$ \emph{minimal} if for every closed annotated partial order $\AnnotatedPOQ=(X_1, X_2, Q, \Value, \SideAnnotation, \GoodWrites)$, 
if $Q\LeafRefines P$ then $Q\Refines P$.
Intuitively, the minimality of $\AnnotatedPO$ guarantees that $P$ is the weakest partial order among all partial orders $Q$ that
\begin{compactenum}
\item agree with $P$ on the order of conflicting pairs of events that belong to leaf threads, and
\item make the resulting annotated partial order  $(X_1, X_2, Q, \Value, \SideAnnotation, \GoodWrites)$ closed.
\end{compactenum}
Hence $P$ does not contain any unnecessary orderings, given these two constraints.
Observe that if $\AnnotatedPO$ is minimal and $\AnnotatedPOK$ is the closure of $\AnnotatedPO$ then $\AnnotatedPOK$ is also minimal.
Afterwards, our algorithm $\VCDPOR$  will use minimal annotated partial orders to represent different classes of the $\VHB$ partitioning.

\smallskip\noindent{\bf Algorithm $\ExtendPO(\AnnotatedPO, X', \Value', \SideAnnotation', \GoodWrites')$.}
Let $\AnnotatedPO=(X_1, X_2, P, \Value, \SideAnnotation, \GoodWrites)$ be a minimal, closed annotated partial order, and $X=X_1\cup X_2$.
Consider 
\begin{compactenum}
\item a set $X'$ with (i)~$X'\setminus X\subseteq \SysWrites$ or $|X'\setminus X|=1$ and (ii)~$X'$ is inevitable for $\AnnotatedPO$, 
\item a value function $\Value'$ over $X'$ such that $\Value\subseteq \Value'$,
\item a side function $\SideAnnotation'$ over $X'$ such that $\SideAnnotation\subseteq \SideAnnotation'$, and
\item a good-writes set $\GoodWrites'$ over $X'$ such that $\GoodWrites\subseteq \GoodWrites'$.
\end{compactenum}
We rely on an algorithm called $\ExtendPO$ that constructs an \emph{extension} of $\AnnotatedPO=(X_1, X_2, P, \Value, \SideAnnotation, \GoodWrites)$ to $X'$, $\Value'$, $\SideAnnotation'$ and $\GoodWrites'$ as a set of minimal closed annotated partial orders $\{\AnnotatedPOK_i=(X'_1, X'_2, K_i, \Value', \SideAnnotation', \GoodWrites' ) \}_i$, where $X'_1\cup X'_2=X'$.
Intuitively, if $\Trace$ is a linearization of $\AnnotatedPO$, then for every extension $\Trace'$ of $\Trace$ such that $\Events{\Trace'}=X'$, $\Value_{\Trace'}=\Value'$ and $\SideAnnotation_{\Trace'}=\SideAnnotation'$, there exists some $\AnnotatedPOK_i$ that linearizes to $\Trace'$.
In $\VCDPOR$, we will use $\ExtendPO$ to extend annotated partial orders with new events.

We describe $\ExtendPO$ for the special case where $|X'\setminus X| = 1$. 
When $|X'\setminus X|=q>1$, $\ExtendPO$ calls itself recursively for every annotated partial order of its output set on a sequence of sets $Y_1,\dots, Y_q$ where $Y_q=X'$, $Y_0=X$ and $|Y_{i+1}\setminus Y_i|=1$.
Let $X'\setminus X=\{\Event \}$.
\begin{compactenum}
\item\label{item:extend_step1} If $\Proc{\Event}=\RootProcess$ (i.e., $\Event$ belongs to the root thread), 
the algorithm simply constructs a partial order $K$ over the set $X'$ such that $K\Project X= P$ and $\Event'<_{K}\Event$ for every event $\Event\in X'$ such that $\Event'<_{\TO} \Event$.
Afterwards, the algorithm constructs the annotated partial order $\AnnotatedPOK=(X'_1, X'_2, K, \Value', \SideAnnotation', \GoodWrites')$ and returns the singleton set $\PartialOrders_{\Write}=\{ \ClosureAlgo(\AnnotatedPOK) \}$.
\item\label{item:extend_step2} If $\Proc{\Event}\neq \RootProcess$ (i.e., $\Event$ belongs to the leaf threads),
the algorithm first constructs a partial order $K$ as in the previous item.
Afterwards, it creates a new partial order $K_i$ for every possible ordering of $\Event$ with all events $\Event'\in X_2$ such that $\Confl{\Event}{\Event'}$.
Finally, the algorithm constructs the annotated partial orders $A=\AnnotatedPOK_i=(X'_1, X'_2, K_i, \Value', \SideAnnotation', \GoodWrites')$, and returns the set
$\PartialOrders=\{ \ClosureAlgo(\AnnotatedPOK_i): \AnnotatedPOK_i\in A \text{ and } \ClosureAlgo(\AnnotatedPOK_i)\neq \bot \}$.
\end{compactenum}

\smallskip\noindent{\bf Causally-happens-before maps, guarding reads and candidate writes.}
A \emph{causally-happens-before (CHB) map} is a map $\NegativeAnnotation: \SysReads \to \System \to \SysReads\cup \{ \bot, \bbot \}$ such that for each read event $\Read\in \Domain(\NegativeAnnotation)$ and thread $\Process\in \System$ we have that $\NegativeAnnotation(\Read)(\Process)\in \SysReads_{\Process}\cup \{ \bot, \bbot \}$.
In words, $\NegativeAnnotation$ maps read events to functions that map every thread $\Process \in \System$ to a read event of $\Process$, or to some initial values $\{\bot, \bbot\}$.
Given a trace $\Trace$ and an event $\Event\in \Events{\Trace}$, we define the \emph{guarding read} $\GuardingRead_{\Trace}(\Event)$ of $\Event$ in $\Trace$ as the last read event of $\Proc{\Event}$ that happens before $\Event$ in $\Trace$, and $\GuardingRead_{\Trace}(\Event)=\bot$ if no such read event exists.
Formally,
\begin{align*}
 \GuardingRead_{\Trace}(\Event) = \max_{\Trace}( \{\Read\in \Reads{\Trace\Project \Proc{\Event}}: \Read <_{\TO} \Event \})
\end{align*}
where we take the maximum of the empty set to be $\bot$.
Given a trace $\Trace$, a CHB map $\NegativeAnnotation$ and a read event $\Read\in \Enabled(\Trace)$, we define the \emph{candidate write set} $\CandidateSet_{\Trace}^{\NegativeAnnotation}(\Read)$ of $\Read$ in $\Trace$ given $\NegativeAnnotation$ as follows:
\begin{align*}
\CandidateSet_{\Trace}^{\NegativeAnnotation}(\Read)=& \{ \Write\in \Writes{\Trace}:~ \Confl{\Read}{\Write} \quad \text{and}\\
\text{either}&\quad \GuardingRead_{\Trace}(\Write)=\bot \quad \text{and} \quad \NegativeAnnotation(\Read)(\Proc{\Write})=\bbot\\
\text{or}&\quad \GuardingRead_{\Trace}(\Write)\neq\bot \quad \text{and also} \quad \NegativeAnnotation(\Read)(\Proc{\Write})\in \{\bbot,\bot\} \quad \text{or} \quad \NegativeAnnotation(\Read)(\Proc{\Write})<_{\TO}\GuardingRead_{\Trace}(\Write)
\end{align*}
We refer to \cref{fig:candidate_writes} for an illustration of the above notation.
Intuitively, $\NegativeAnnotation(\Read)(\Process)$ encodes the prefix of the local trace of thread $\Process$ that contains write events which have already been considered by the algorithm as good writes for $\Read$. 
Instead of the whole prefix, we store the last read of that prefix. 
The two special values $\bot\bot$ and $\bot$ encode the empty prefix, and the prefix before the first read. 
The guarding read of a write $\Write$ is the last local read event the same thread that appears before $\Write$ in the execution so far. 
Hence, if the guarding read of $\Write$ appears before $C(\Read)(\Process)$, we know that $\Write$ has been considered as a good write for $\Read$. 
The candidate write set for $\Read$ contains writes that are considered as good writes for $\Read$ in the current recursive step. 
\begin{figure}[!h]
\small
\centering
\begin{subfigure}[b]{0.35\textwidth}
\centering
\small
\def\rownumber{}
\begin{tabular}[b]{@{\makebox[1.2em][r]{\rownumber\space}} | l | l | l |}
\normalsize{$\mathbf{\SeqTrace_1}$} & \normalsize{$\mathbf{\SeqTrace_2}$} & \normalsize{$\mathbf{\SeqTrace_3}$}
\gdef\rownumber{\stepcounter{magicrownumbers}\arabic{magicrownumbers}} \\
\hline
$\Write_x$ & & \\
& $\Write_y$ & \\
& $\Write_x$ & \\
& $\Read_y$ & \\
& $\Write_x$ & \\
& & $\Read_y$ \\
& & $\Write_x$ \\
& & $\Read_x$ \\
& & $\Write_x$ \\
\hline
\end{tabular}
\caption{A trace $\Trace$.
Threads $\Process_1$ and $\Process_3$ have enabled events $\Read_x^1$ and $\Read_x^3$ (not shown), which access the variable $x$.
}
\label{subfig:candidate_writes_trace}
\end{subfigure}
\quad
\begin{subfigure}[b]{0.6\textwidth}
\centering
\begin{tikzpicture}[thick,
pre/.style={<-,shorten >= 1pt, shorten <=1pt, thick},
post/.style={->,shorten >= 2pt, shorten <=2pt,  very thick},
seqtrace/.style={->, line width=2},
und/.style={very thick, draw=gray},
event/.style={rectangle, minimum height=0.8mm, minimum width=3mm, fill=black!100,  line width=1pt, inner sep=0},
virt/.style={circle,draw=black!50,fill=black!20, opacity=0}]

\def\ystep{1.5}
\node[] at (0,1*\ystep){
$\begin{aligned}
\Enabled(\Trace)\cap \SysEvents_{\Process_1}&=\Read_x^1\\
\Enabled(\Trace)\cap \SysEvents_{\Process_3}&=\Read_x^3\\
\end{aligned}$
};

\node[] at (0,0*\ystep) {
$\begin{aligned}
\NegativeAnnotation(\Read_x^1)&=\{ (\Process_1, \bot), (\Process_2, \Event_4), (\Process_3, \Event_6) \}\\
\NegativeAnnotation(\Read_x^3)&=\{ (\Process_1, \bbot), (\Process_2, \bbot), (\Process_3, \bbot) \}
\end{aligned}$
};

\node[] at (0,-1*\ystep) {
$\begin{aligned}
\CandidateSet_{\Trace}^{\NegativeAnnotation}(\Read_x^1)&=\{ \Event_9 \}\\
\CandidateSet_{\Trace}^{\NegativeAnnotation}(\Read_x^3)&=\{\Event_{1}, \Event_3, \Event_5, \Event_7, \Event_9 \}\\
\end{aligned}$
};

\end{tikzpicture}
\caption{
The candidate write sets of the read events $\Read_x^1$ and $\Read_x^3$ given the causally-happens-before map $\NegativeAnnotation$.
}
\label{subfig:candidate_writes}
\end{subfigure}
\caption{
Example of a trace (\protect\cref{subfig:candidate_writes_trace}) and candidate write sets of read events given their causally-happens-before maps (\protect
\cref{subfig:candidate_writes}). 
We denote by $\Event_i$ the $i$-th event of $\Trace$.
}
\label{fig:candidate_writes}
\end{figure}

\begin{algorithm}
\small
\SetInd{0.4em}{0.4em}
\DontPrintSemicolon
\caption{$\VCDPOR(\AnnotatedPO=(X_1, X_2, P, \Value, \SideAnnotation, \GoodWrites), \NegativeAnnotation)$}\label{algo:vcdpor}
\KwIn{A minimal closed annotated partial order $\AnnotatedPO$, a CHB map $\NegativeAnnotation$.
}
\BlankLine
$\Trace'\gets \Realize(\AnnotatedPO)$\tcp*[f]{$\AnnotatedPO$ is closed hence realizable}\label{line:vcdpor_realize}\\
$\Trace\gets \WriteExtend(\Trace')$\tcp*[f]{Extend $\Trace'$ until before the next read of each thread}\label{line:vcdpor_wextend} \\
\ForEach(\tcp*[f]{Extensions of $\AnnotatedPO$ to $\Events{\Trace}$}){$\AnnotatedPOQ\in \ExtendPO(\AnnotatedPO, \Events{\Trace}, \Value_{\Trace}, \SideAnnotation, \GoodWrites)$}{\label{line:vcdpor_extend}
$\NegativeAnnotation_{\AnnotatedPOQ} \gets \NegativeAnnotation$\tcp*[f]{Create a copy of the CHB $\NegativeAnnotation$}\label{line:vcdpor_neg_copy}\\
$\MutateRoot(\AnnotatedPOQ, \Trace, \NegativeAnnotation_{\AnnotatedPOQ})$\tcp*[f]{Process the root thread}\label{line:vcdpor_call_root}\\
\ForEach(\tcp*[f]{Process the leaf threads}){$\Process \in \Leaves$}{\label{line:vcdpor_for_leaves}
$\MutateLeaf(\AnnotatedPOQ, \Trace, \NegativeAnnotation_{\AnnotatedPOQ}, \Process)$\label{line:vcdpor_call_leaf}\\
}
}
\end{algorithm}
\smallskip\noindent{\bf Algorithm $\VCDPOR$.}
We are now ready to describe our main algorithm $\VCDPOR$ for the enumerative exploration of the partitioning $\TraceSpace/\VHB$.
The algorithm takes as input a minimal closed annotated partial order $\AnnotatedPO$ and a CHB map $\NegativeAnnotation$.
First, $\VCDPOR$ calls $\Realize$ to obtain a linearization $\Trace'$ of $\AnnotatedPO$ and constructs the write-extension $\Trace$ of $\Trace'$ which reveals new write events in $\Trace$.
Afterwards, the algorithm extends $\AnnotatedPO$ to the set $\Events{\Trace}$ by calling $\ExtendPO$.
Recall that $\ExtendPO$ returns a set of minimal closed annotated partial orders.
For every annotated partial order $\AnnotatedPOQ$ returned by $\ExtendPO$, the algorithm calls $\MutateRoot$ to process the read event of the root thread $\RootProcess$ that is enabled in $\Trace$.
Finally, the algorithm calls $\MutateLeaf$ for every leaf thread $\Process\neq \RootProcess$ to process the read event of $\Process$ that is enabled in $\Trace$.
For the initial call, we construct an empty annotated partial order $\AnnotatedPO$ and an initial CHB map $\NegativeAnnotation$ that for every read event $\Read\in \SysReads$ and thread $\Process\in \System$ maps $\NegativeAnnotation(\Read)(\Process)=\{ \bbot \}$.

\begin{algorithm}
\small
\SetInd{0.4em}{0.4em}
\DontPrintSemicolon
\caption{$\MutateRoot(\AnnotatedPOQ=(X_1, X_2, Q, \Value, \SideAnnotation, \GoodWrites), \Trace, \NegativeAnnotation_{\AnnotatedPOQ})$}\label{algo:mutate_root}
\KwIn{A minimal closed annotated partial order $\AnnotatedPOQ$, a trace $\Trace$, a CHB map $\NegativeAnnotation_{\AnnotatedPOQ}$.
}
\BlankLine
$\Read \gets\Enabled(\Trace, \RootProcess)$\tcp*[f]{The next enabled event in $\RootProcess$ is a read}\label{line:mutate_root_enabled}\\
$Y_1\gets \CandidateSet_{\Trace}^{\NegativeAnnotation_{\AnnotatedPOQ}}(\Read)\cap \SysWrites_{\RootProcess}$\tcp*[f]{The set of local candidate writes of $\Read$}\label{line:mutate_root_local_candidate}\\
$Y_2\gets \CandidateSet_{\Trace}^{\NegativeAnnotation_{\AnnotatedPOQ}}(\Read)\cap \SysWrites_{\neq \RootProcess}$\tcp*[f]{The set of remote candidate writes of $\Read$}\label{line:mutate_root_remote_candidate}\\
\ForEach(\tcp*[f]{$i=1$ ($i=2$) reads from local (remote) writes}){$i\in [2]$}{\label{line:mutate_root_for_local_remote}
$\SideAnnotation_{\Read}\gets \SideAnnotation\cup \{ (\Read, i) \}$\tcp*[f]{The new side function}\label{line:mutate_root_side_fn}\\
$\ValueDomain_{\Read}\gets \{ \Value_{\Trace}(\Write):~\Write\in Y_i \}$\tcp*[f]{The set of values of candidate writes of $\Read$}\label{line:mutate_root_values}\\
\ForEach(\tcp*[f]{Every value $v$ that $\Read$ may read}){$v \in \ValueDomain_{\Read}$}{
$\Value_{\Read}\gets \Value_{\Trace} \cup \{ (\Read, v) \}$\tcp*[f]{The new value function}\label{line:mutate_root_value_fn}\\
$\GoodWrites_{\Read}\gets \GoodWrites \cup \{ (\Read, \{\Write\in Y_i:~\Value_{\Trace}(\Write)=v \}) \}$\tcp*[f]{The new good-writes function}\label{line:mutate_root_good_fn}\\
$\AnnotatedPOK\gets \ExtendPO(\AnnotatedPOQ, X_1\cup X_2 \cup \{\Read\}, \Value_{\Read}, \SideAnnotation_{\Read}, \GoodWrites_{\Read})$\tcp*[f]{Returns  one element}\label{line:mutate_root_extend}\\
\uIf(\tcp*[f]{Extension is successful}){$\AnnotatedPOK\neq \bot$}{
Call $\VCDPOR(\AnnotatedPOK, \NegativeAnnotation_{\AnnotatedPOQ})$\tcp*[f]{Recurse}\\
}
}
}
$\NegativeAnnotation_{\AnnotatedPOQ}(\Read)\gets \{ (\Process, \max_{\Trace}(\{\Reads{\Trace}\Project\Process\}) ):~ \Process \in \System \} $\tcp*[f]{The last read of each thread in $\Trace$}\\
\end{algorithm}
\smallskip\noindent{\bf Algorithm $\MutateRoot$.}
The algorithm takes as input a minimal closed annotated partial order $\AnnotatedPOQ$, a trace $\Trace$ and a CHB map $\NegativeAnnotation_{\AnnotatedPOQ}$, and attempts all possible extensions of $\AnnotatedPOQ$ with the read event $\Read$ of $\RootProcess$ that is enabled in $\Trace$ to all possible values that are written in $\Trace$.
The algorithm first constructs two sets $Y_1$ and $Y_2$ which hold the local and remote, respectively, write events of $\Trace$ that are candidate writes for $\Read$ according to the CHB map $\NegativeAnnotation_{\AnnotatedPOQ}$.
Then, it iterates over the local ($i=1$) and remote ($i=2$) write choices for $\Read$ in $Y_i$.
Finally, the algorithm 
(i)~collects all possible values that $\Read$ may read from the set $Y_i$,
(ii)~constructs the appropriate new side function, value function and good-writes function, and 
(iii)~calls $\ExtendPO$ on these new parameters in order to establish the respective extension for $\Read$.
For every such case, $\ExtendPO$ returns a new minimal, closed annotated partial order $\AnnotatedPOK$ which is passed recursively to $\VCDPOR$.

\begin{algorithm}
\small
\SetInd{0.4em}{0.4em}
\DontPrintSemicolon
\caption{$\MutateLeaf(\AnnotatedPOQ=(X_1, X_2, Q, \Value, \SideAnnotation, \GoodWrites), \Trace, \NegativeAnnotation_{\AnnotatedPOQ}, \Process)$}\label{algo:mutate_leaf}
\KwIn{A minimal closed annotated partial order $\AnnotatedPOQ$, a trace $\Trace$, a CHB map $\NegativeAnnotation_{\AnnotatedPOQ}$, a thread $\Process$.
}
\BlankLine
$\Read \gets\Enabled(\Trace, \Process)$\tcp*[f]{The next enabled event in $\Process$ is a read}\\
$\ValueDomain_{\Read}\gets \{ \Value_{\Trace}(\Write):~\Write\in \CandidateSet_{\Trace}^{\NegativeAnnotation_{\AnnotatedPOQ}}(\Read) \}$\tcp*[f]{The set of values of candidate writes of $\Read$}\label{line:mutate_leaf_values}\\
\ForEach(\tcp*[f]{Every value $v$ that $\Read$ may read}){$v \in \ValueDomain_{\Read}$}{\label{line:mutate_leaf_for_value}
$\Value_{\Read}\gets \Value_{\Trace} \cup \{ (\Read, v) \}$\tcp*[f]{The new value function}\label{line:mutate_leaf_value_fn}\\
$\GoodWrites_{\Read}\gets \GoodWrites \cup \{ (\Read, \{ \Write \in \CandidateSet_{\Trace}^{\NegativeAnnotation_{\AnnotatedPOQ}}(\Read):~\Value_{\Trace}(\Write)=v \}) \}$\tcp*[f]{The new good-writes function}\label{line:mutate_leaf_good_fn}\\
\ForEach(\tcp*[f]{Returns many elements}){$\AnnotatedPOK \in \ExtendPO(\AnnotatedPOQ, X_1\cup X_2 \cup \{\Read\}, \Value_{\Read}, \SideAnnotation, \GoodWrites_{\Read})$}{\label{line:mutate_leaf_extend}
Call $\VCDPOR(\AnnotatedPOK, \NegativeAnnotation_{\AnnotatedPOQ})$\tcp*[f]{Recurse}\\
}
}
$\NegativeAnnotation_{\AnnotatedPOQ}(\Read)\gets \{ (\Process, \max_{\Trace}(\{\Reads{\Trace}\Project\Process\}) ):~ \Process \in \System \} $\tcp*[f]{The last read of each thread in $\Trace$}\\
\end{algorithm}

\smallskip\noindent{\bf Algorithm $\MutateLeaf$.}
The algorithm $\MutateLeaf$ takes as input a minimal closed partial order $\AnnotatedPOQ$, a trace $\Trace$, a CHB map $\NegativeAnnotation_{\AnnotatedPOQ}$, and a thread $\Process\in\Leaves$.
Similarly to $\MutateRoot$, $\MutateLeaf$ attempts all possible extensions of $\AnnotatedPOQ$ with the read event $\Read$ of $\Process$ that is enabled in $\Trace$ to all possible values that are written in $\Trace$.
The main difference compared to $\MutateRoot$ is that since $\Read$ belongs to a leaf thread, $\ExtendPO$ returns a set of minimal, closed annotated partial orders (as opposed to just one) which result from all possible orderings of $\Read$ with the write events of $X_2$ that are conflicting with $\Read$.
Then $\MutateLeaf$ makes a recursive call to $\VCDPOR$ for each such annotated partial order.

The following theorem states the main result of this paper.

\smallskip
\begin{theorem}\label{them:vcdpor}
Consider a concurrent program $\System$ over a constant number of threads, and let $\TraceSpaceMax$ be the maximal trace space of $\System$.
$\VCDPOR$ solves the local-state reachability problem on $\System$ and requires $O\left(|\TraceSpaceMax/\VHB|\cdot \Poly(n) \right )$ time, where $n$ is the length of the longest trace in $\TraceSpaceMax$.
\end{theorem}

We conclude with two remarks on  space usage  and the way lock events can be handled.

\smallskip
\begin{remark}[Space complexity]\label{rem:recursive}
To make our presentation simpler so far, $\VCDPOR$ and $\MutateLeaf$ iterate over the set of annotated partial orders returned by $\ExtendPO$, which can be exponentially large.
An efficient variant of $\VCDPOR$ shall explore these sets recursively, instead of computing all elements of each set imperatively.
This results in polynomial space complexity for $\VCDPOR$.
\end{remark}

\smallskip
\begin{remark}[Handling locks]\label{rem:locks}
For simplicity of presentation, so far we have neglected locks.
However, lock events can be incorporated naturally, as follows.
\begin{compactenum}
\item Each lock-release event is a write event, writing an arbitrary value.
\item Each lock-acquire event is a read event. Given two lock-acquire events $\Read_1, \Read_2$ the algorithm maintains that $\GoodWrites(\Read_1)\cap \GoodWrites(\Read_2)=\emptyset$
\end{compactenum}
\end{remark}

\newcommand{\myblue}{blue!85!black}
\newcommand{\myred}{red!85!black}

\smallskip\noindent{\bf $\VCDPOR$ running example.}
\begin{figure}
\centering
\small
\begin{subfigure}[b]{.3\textwidth}
\centering
\begin{subfigure}[t]{0.15\textwidth}
\begin{align*}
\text{Th}&\text{read}~\Process_{1}:\\
\hline\\[-1em]
1.~& \textcolor{\myblue}{\Write(y,1)}\\
2.~& \textcolor{\myblue}{\Write(y,2)}\\
3.~& \textcolor{\myred}{\Write(x,1)}\\
4.~& \textcolor{\myred}{\Read(x)}\\
\end{align*}
\end{subfigure}
\begin{subfigure}[t]{0.15\textwidth}
\begin{align*}
\text{Th}&\text{read}~\Process_{2}:\\
\hline\\[-1em]\
1.~& \textcolor{\myred}{\Write(x,1)}\\
2.~& \textcolor{\myred}{\Read(x)}\\
3.~& \textcolor{\myblue}{\Write(y,2)}\\
4.~& \textcolor{\myblue}{\Read(y)}\\
\end{align*}
\end{subfigure}
\caption{A concurrent program $\System$.}
\label{subfig:vcdpor_program}
\end{subfigure}
\begin{subfigure}[b]{.69\textwidth}
\centering
\begin{tikzpicture}[thick, >=latex,
pre/.style={<-,shorten >= 1pt, shorten <=1pt, thick},
post/.style={->,shorten >= 1pt, shorten <=1pt,  thick},
und/.style={very thick, draw=gray},
node1/.style={circle, minimum size=4mm, draw=black!100, line width=1pt, inner sep=0},
node2/.style={circle, minimum size=5mm, draw=black!100, fill=white!100, very thick, inner sep=0},
virt/.style={circle,draw=black!50,fill=black!20, opacity=0}]

\newcommand{\xdisposition}{0}
\newcommand{\ydisposition}{0}
\newcommand{\xstep}{1}
\newcommand{\ystep}{1.4}

\node[node1] (t) at (\xdisposition  + 0 * \xstep, \ydisposition + 0 * \ystep ){$a$};
\node[node1] (t1) at (\xdisposition  - 3.5 * \xstep, \ydisposition - 1 * \ystep ){$b$};
\node[node1] (t2) at (\xdisposition  + 0 * \xstep, \ydisposition - 1 * \ystep ){$f$};
\node[node1] (t3) at (\xdisposition  + 2 * \xstep, \ydisposition - 1 * \ystep ){$i$};

\node[node1] (t11) at (\xdisposition  - 3.5 * \xstep, \ydisposition - 2 * \ystep ){$c$};
\node[node1] (t21) at (\xdisposition  + 0 * \xstep, \ydisposition - 2 * \ystep ){$g$};
\node[node1] (t31) at (\xdisposition  + 1.7 * \xstep, \ydisposition - 2 * \ystep ){$j$};
\node[node1] (t32) at (\xdisposition  + 2.3 * \xstep, \ydisposition - 2 * \ystep ){$k$};

\node[node1] (t111) at (\xdisposition  - 3.8 * \xstep, \ydisposition - 3 * \ystep ){$d$};
\node[node1] (t112) at (\xdisposition  - 3.2 * \xstep, \ydisposition - 3 * \ystep ){$e$};
\node[node1] (t211) at (\xdisposition  + 0 * \xstep, \ydisposition - 3 * \ystep ){$h$};

\draw[->, very thick] (t) to node[left, label={[align=right, yshift=-6, xshift=-20]
$\textcolor{\myred}{\Read_{p_1}^4} \leftarrow \{ \textcolor{\myred}{\Write_{p_1}^3} \}$
}] {} (t1);
\draw[->, very thick] (t) to node[left, label={[align=right, yshift=-18, xshift=-23]
$\textcolor{\myred}{\Read_{p_1}^4} \leftarrow \{ \textcolor{\myred}{\Write_{p_2}^1} \}$
}] {} (t2);
\draw[->, very thick] (t) to node[left, label={[align=right, yshift=-12, xshift=20]
$\textcolor{\myred}{\Read_{p_2}^2} \leftarrow \{ \textcolor{\myred}{\Write_{p_1}^3},\:$ \\ $\textcolor{\myred}{\Write_{p_2}^1} \}$
}] {} (t3);

\draw[->, very thick] (t1) to node[left, label={[align=right, yshift=-22, xshift=-26]
$\textcolor{\myred}{\Read_{p_2}^2} \leftarrow \{ \textcolor{\myred}{\Write_{p_1}^3},\:$ \\ $\textcolor{\myred}{\Write_{p_2}^1} \}$
}] {} (t11);
\draw[->, very thick] (t2) to node[left, label={[align=right, yshift=-22, xshift=-26]
$\textcolor{\myred}{\Read_{p_2}^2} \leftarrow \{ \textcolor{\myred}{\Write_{p_1}^3},\:$ \\ $\textcolor{\myred}{\Write_{p_2}^1} \}$
}] {} (t21);
\draw[->, very thick] (t3) to node[left, label={[align=right, yshift=-12, xshift=-24]
$\textcolor{\myblue}{\Read_{p_2}^4} \leftarrow \{ \textcolor{\myblue}{\Write_{p_1}^1} \}$
}] {} (t31);
\draw[->, very thick] (t3) to node[left, label={[align=right, yshift=-22, xshift=32]
$\textcolor{\myblue}{\Read_{p_2}^4} \leftarrow \{ \textcolor{\myblue}{\Write_{p_1}^2},\:$ \\ $\textcolor{\myblue}{\Write_{p_2}^3} \}$
}] {} (t32);

\draw[->, very thick] (t11) to node[label={[align=right, yshift=-12, xshift=28]
$\textcolor{\myblue}{\Read_{p_2}^4} \leftarrow \{ \textcolor{\myblue}{\Write_{p_1}^1} \}$
}] {} (t112);
\draw[->, very thick] (t11) to node[left, label={[align=right, yshift=-12, xshift=-24]
$\textcolor{\myblue}{\Read_{p_2}^4} \leftarrow \{ \textcolor{\myblue}{\Write_{p_1}^2},\:$ \\ $\textcolor{\myblue}{\Write_{p_2}^3} \}$
}] {} (t111);
\draw[->, very thick] (t21) to node[right, label={[align=right, yshift=-22, xshift=26]
$\textcolor{\myblue}{\Read_{p_2}^4} \leftarrow \{ \textcolor{\myblue}{\Write_{p_1}^2},\:$ \\ $\textcolor{\myblue}{\Write_{p_2}^3} \}$
}] {} (t211);

\end{tikzpicture}
\caption{The $\VCDPOR$ exploration tree.}
\label{subfig:vcdpor_recursion}
\end{subfigure}
\caption{A program with two threads (\protect\cref{subfig:vcdpor_program}) and the corresponding $\VCDPOR$ exploration (\protect\cref{subfig:vcdpor_recursion}).}
\label{fig:vcdpor_example}
\end{figure}
\cref{fig:vcdpor_example} illustrates the main aspects of $\VCDPOR$ (Algorithms~\ref{algo:vcdpor},~\ref{algo:mutate_root},~and~\ref{algo:mutate_leaf})
on a small example.
We start with an empty annotated partial order $\AnnotatedPO$ and a CHB map $\NegativeAnnotation$ that is empty
(i.e., $\NegativeAnnotation(\Read)(\Process)=\{ \bbot \}$ for every read event $\Read\in \SysReads$ and thread $\Process\in \System$).
The initial trace obtained in~\cref{line:vcdpor_realize}~of~\cref{algo:vcdpor} is $\Trace' = \varepsilon$. Its write-extension
$\Trace$ in~\cref{line:vcdpor_wextend} contains the three writes of~$\Process_{1}$ and the first write of~$\Process_{2}$. Next,
\cref{line:vcdpor_extend} returns an annotated partial order $\AnnotatedPOQ_a$ that corresponds to the thread order $\TO\Project \Events{\Trace}$.
In $\Trace$, the root thread~$\Process_{1}$ has an enabled event (which is always a read), so $\MutateRoot$ (\cref{algo:mutate_root}) is called
on $\AnnotatedPOQ_a$ and the (empty) CHB map $\NegativeAnnotation_{\AnnotatedPOQ_a}$.
\hfill ($\dagger$) \\

The enabled read in~\cref{line:mutate_root_enabled} is $\textcolor{\myred}{\Read_{\Process_1}^4}$, its local candidate write
(computed in~\cref{line:mutate_root_local_candidate}) is $\textcolor{\myred}{\Write_{\Process_1}^3}$ and its
remote candidate write (computed in~\cref{line:mutate_root_remote_candidate}) is $\textcolor{\myred}{\Write_{p_2}^1}$.
This holds because $\NegativeAnnotation_{\AnnotatedPOQ_a}(\Read)(\Process_1)=\{ (\Process_1, \bbot), (\Process_2, \bbot) \}$,
which allows any write event to be observed.
For the local (\cref{line:mutate_root_for_local_remote}, $i=1$) candidate $\textcolor{\myred}{\Write_{\Process_1}^3}$, first the side
function is updated with $\{ ( \textcolor{\myred}{\Read_{\Process_1}^4}, 1) \}$ in~\cref{line:mutate_root_side_fn}.
Then in~\cref{line:mutate_root_values}, the only considered value is $1$. Thus, in~\cref{line:mutate_root_value_fn} the
value function is updated with $\{ ( \textcolor{\myred}{\Read_{p_1}^4}, 1) \}$, and in~\cref{line:mutate_root_good_fn} the
good-writes function is updated with $\{ ( \textcolor{\myred}{\Read_{\Process_1}^4}, \{  \textcolor{\myred}{\Write_{\Process_1}^3} \} ) \}$.
Then, such an update is successfully realized in~\cref{line:mutate_root_extend} by $\ExtendPO$, where the partial order is extended
with $\textcolor{\myred}{\Read_{\Process_1}^4}$ and afterwards it is closed using algorithm $\ClosureAlgo$~(\cref{algo:closure}).
Thus $\VCDPOR$~(\cref{algo:vcdpor}) is recursively called on the corresponding annotated partial order $\AnnotatedPOK_a$
(and the empty CHB map $\NegativeAnnotation_{\AnnotatedPOQ_a}$), and we proceed to the child $b$ of $a$.

In node $b$, no new event is added during the write-extension (\cref{line:vcdpor_wextend}), as
$\textcolor{\myred}{\Read_{p_1}^4}$ is the last event of~$\Process_{1}$, and in~\cref{line:vcdpor_extend} we obtain $\AnnotatedPOQ_b$.
The only thread with an enabled read event is~$\Process_{2}$, so $\MutateLeaf$ (\cref{algo:mutate_leaf}) is called on $\AnnotatedPOQ_b$
and $\Process_{2}$ (and empty CHB map $\NegativeAnnotation_{\AnnotatedPOQ_b}$).
The enabled read $\textcolor{\myred}{\Read_{\Process_2}^2}$ has candidate writes
$\textcolor{\myred}{\Write_{p_1}^3}$ and $\textcolor{\myred}{\Write_{p_2}^1}$, both of which write the same value
(c.f.~\cref{line:mutate_leaf_values}), and hence the algorithm will allow $\textcolor{\myred}{\Read_{\Process_2}^2}$ to observe either.
This is an example of the value-centric gains we obtain in this work. In~\cref{line:mutate_leaf_value_fn} the value function is
updated with $\{ ( \textcolor{\myred}{\Read_{\Process_2}^2}, 1) \}$, and in~\cref{line:mutate_leaf_good_fn} the good-writes function
is updated with $\{ ( \textcolor{\myred}{\Read_{\Process_2}^2}, \{  \textcolor{\myred}{\Write_{\Process_1}^3}, \textcolor{\myred}{\Write_{\Process_2}^1} \} ) \}$.
The realization of this update happens in~\cref{line:mutate_leaf_extend} by $\ExtendPO$, where the partial order is extended with
$\textcolor{\myred}{\Read_{\Process_2}^2}$ and then closed using algorithm $\ClosureAlgo$~(\cref{algo:closure}).
One annotated partial order $\AnnotatedPOK_b$ is returned and it is the argument of the further $\VCDPOR$ call
(with an empty CHB map $\NegativeAnnotation_{\AnnotatedPOQ_b}$), we proceed to the child $c$ of $b$.
In node $c$, the write-extension adds the event $\textcolor{\myblue}{\Write_{\Process_2}^3}$,
which, in similar steps as before, will lead to nodes $d$ and $e$.

Next, the recursion backtracks to the call of $\MutateRoot$ in the node $a$ ($\dagger$).
The second iteration ($i = 2$) of the loop in~\cref{line:mutate_root_for_local_remote} proceeds, where the remote candidate write
$\textcolor{\myred}{\Write_{\Process_2}^1}$ is considered for $\textcolor{\myred}{\Read_{\Process_1}^4}$.
In a similar fashion, the descendants $f$, $g$, and $h$ are created and $h$ concludes with a maximal trace.

Finally, the recursion backtracks to the node $a$ again, where $\MutateRoot$ ($\dagger$) concludes with updating the CHB
map as follows: $\NegativeAnnotation_{\AnnotatedPOQ_a}(\textcolor{\myred}{\Read_{\Process_1}^4})=\{ (\Process_{1}, \bot),\; (\Process_{2}, \bot) \}$.
The control-flow comes back to the initial $\VCDPOR$ call (from~\cref{line:vcdpor_call_root}), where the annotated partial
order $\AnnotatedPOQ_a$ with the (now updated) CHB map $\NegativeAnnotation_{\AnnotatedPOQ_a}$ is considered.
The thread $\Process_{2}$ has an enabled read ($\textcolor{\myred}{\Read_{\Process_2}^2}$) in $\Trace$, hence
$\MutateLeaf$ is called on $\AnnotatedPOQ_a$, $\NegativeAnnotation_{\AnnotatedPOQ_a}$, and $\Process_{2}$.
Eventually, the descendants $i$, $j$, and $k$ are created and the exploration concludes.
Note that in each of $i$, $j$, $k$, the thread $\Process_{1}$ has an enabled read $\textcolor{\myred}{\Read_{\Process_1}^4}$.
However, note that $\GuardingRead_{\Trace}(\textcolor{\myred}{\Write_{\Process_1}^3})=\GuardingRead_{\Trace}(\textcolor{\myred}{\Write_{\Process_2}^1})=\bot$
and in all those nodes we have 
$\NegativeAnnotation(\textcolor{\myred}{\Read_{\Process_1}^4})(\Process_1)=\{ (\Process_1, \bot), (\Process_2, \bot) \}$, and thus
$\textcolor{\myred}{\Write_{p\Process1}^3}$ and $\textcolor{\myred}{\Write_{\Process_2}^1}$ are never considered as candidate writes for $\textcolor{\myred}{\Read_{\Process_1}^4}$.
This illustrates how $\VCDPOR$ never explores the same class of $\VHB$ twice.

\section{Experiments}\label{sec:experiments}

We have seen in \cref{them:comparison} that $\VHB$ is a coarse partitioning that can be explored efficiently by $\VCDPOR$.
In this section we present an experimental evaluation of $\VCDPOR$ on various classes of concurrent benchmarks, to assess
\begin{compactenum}
\item the reduction of the trace-space partitioning achieved by $\VHB$, and
\item the efficiency with which this partitioning is explored by $\VCDPOR$.
\end{compactenum}

\smallskip\noindent{\bf Implementation and experiments.}
To address the above questions, we have made a prototype implementation of $\VCDPOR$ in the stateless model checker Nidhugg~\cite{Abdulla2015},
which works on LLVM IR\footnote{Code accessible at \url{https://github.com/ViToSVK/nidhugg/tree/valuecentric_stable}}.
We have tested $\VCDPOR$ on benchmarks coming in four classes:
\begin{compactenum}
\item The TACAS  Software Verification Competition (SV-COMP).
\item Mutual-exclusion algorithms from the literature.
\item Multi-threaded dynamic-programming algorithms that use memoization.
\item Individual benchmarks that exercise various concurrency patterns.
\end{compactenum}
Each benchmark comes with a scaling parameter, which is either the number of threads, or an unroll bound on all loops of the benchmark
(often the unroll bound also controls the number of threads that are spawned.)
We have compared our algorithm with three other state-of-the-art DPOR algorithms that are implemented in Nidhugg, namely
$\Source$~\cite{Abdulla14}, $\Optimal$~\cite{Abdulla14} and $\OptimalObs$ (``optimal with observers'')~\cite{Aronis18},
as well as our own implementation of $\DCDPOR$~\cite{Chalupa17}.
For our experiments, we have used a Linux machine with Intel(R) Xeon(R) CPU E5-1650 v3 @ 3.50GHz and 128GB of RAM.
We have run Nidhugg with Clang and LLVM version 3.8.
In all cases, we report the number of maximal traces and the total running time of each algorithm, subject to a timeout of 4~hours, indicated by ``-''.

\smallskip\noindent{\bf Implementation details.}
Here we clarify some details regarding our implementation.
\begin{compactenum}
\item In our theory so far, we have neglected dynamic thread creation for simplicity of presentation.
In practice, all our benchmarks spawn threads dynamically.
This situation is handled straightforwardly, by including in our partial orders the orderings that are naturally induced by spawn and join events.
\item The root thread is chosen as the first thread that is spawned from the main thread.
We make this choice instead of the main thread as in many benchmarks, the main thread mainly spawns worker threads and performs only a few concurrent operations.
\item In our presentation of $\ExtendPO(\AnnotatedPO, X', \Value', \SideAnnotation', \GoodWrites')$,
given $X'\setminus X=\{\Event \}$ such that $\Event$ belongs to a leaf thread, we consider all possible orderings of $\Event$ with conflicting
events from all leaf threads. 
In our implementation, we relax this in two ways. 
Given a write event $\Event_{\Write}$, we say it is
\emph{never-good} if it does not belong to $\GoodWrites'(\Read)$ for any read event $\Read$. 
Further, given $\Event_{\Write}$ and an annotated
partial order $\AnnotatedPOK$, we say that $\Event_{\Write}$ is \emph{unobservable} in $\AnnotatedPOK$, if for every linearization of $\AnnotatedPOK$ no read event can observe $\Event_{\Write}$. 
Given two unordered conflicting write events from leaf threads, we do not order them if
(i) both are never-good, or (ii) at least one is unobservable.
\end{compactenum}

\begin{figure}
\makebox[\linewidth][c]{
\begin{subfigure}[t]{0.5\textwidth}
\includegraphics[scale=0.33]{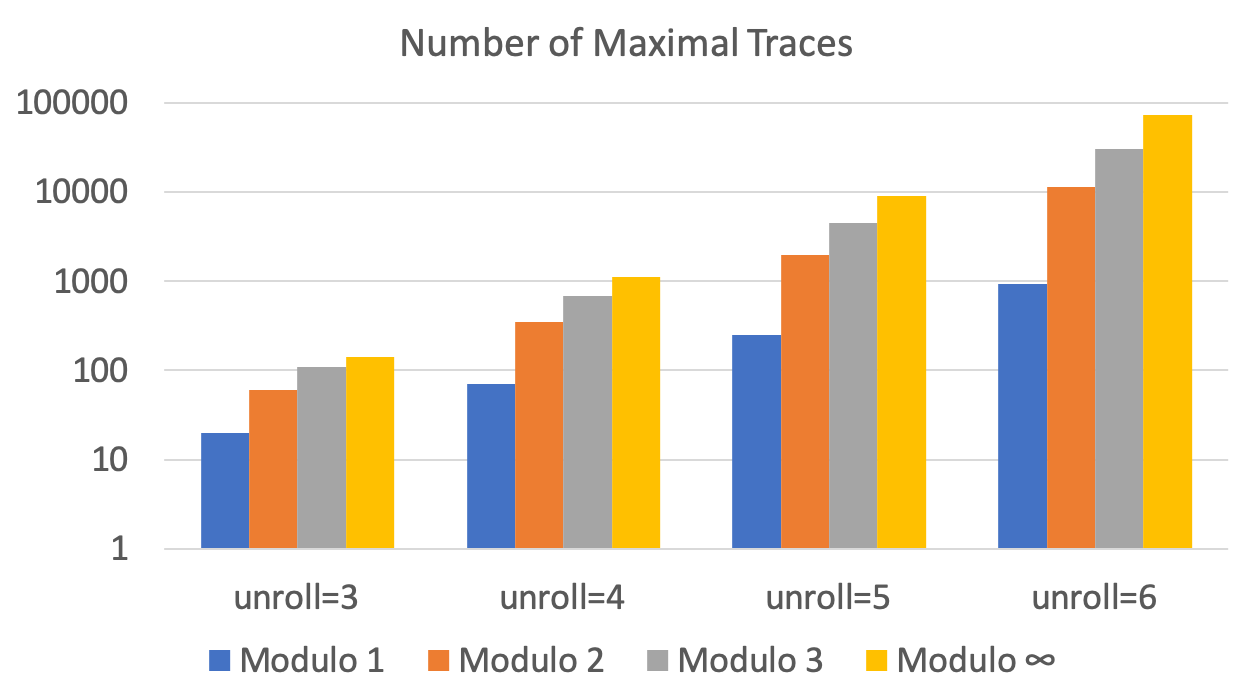}
\caption{Number of traces.}
\label{subfig:fibbenchtraces}
\end{subfigure}
\qquad
\begin{subfigure}[t]{0.5\textwidth}
\includegraphics[scale=0.33]{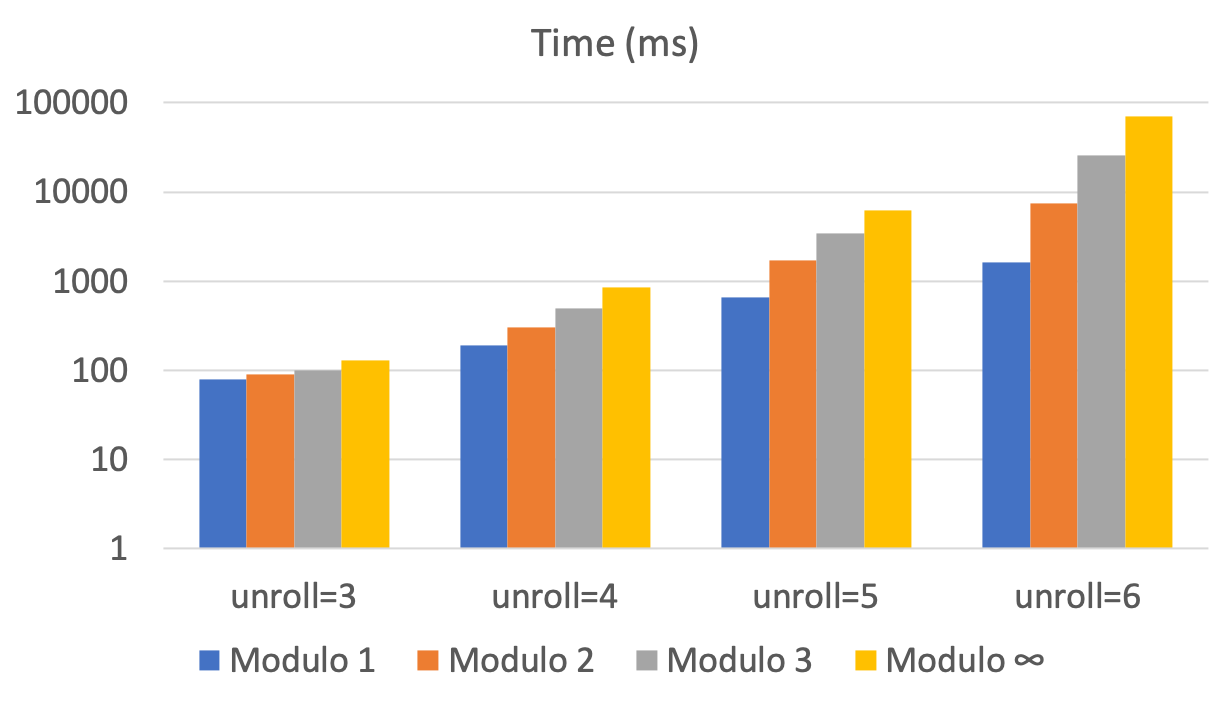}
\caption{Running time.}
\label{subfig:fibbenchtime}
\end{subfigure}
}
\caption{Number of traces (\protect\subref{subfig:fibbenchtraces}) and running time (\protect\subref{subfig:fibbenchtime}) on variants of the \texttt{fib\_bench} benchmark.}
\label{fig:fib_bench}
\end{figure}

\smallskip\noindent{\bf Value-centric gains.}
As a preliminary experimental step, we explore the gains of our value-centric technique on small variants of the simple benchmark \texttt{fib\_bench} from SV-COMP.
This benchmark consists of a main thread and two worker threads, and two global variables $x$ and $y$.
The first worker thread enters a loop in which it performs the update $x\gets x+y$.
Similarly, the second worker thread enters a loop in which it performs the update $y\gets y+x$.
To explore the sensitivity of our value-centric DPOR to values, we have created three variants \texttt{fib\_bench\_1}, \texttt{fib\_bench\_2}, \texttt{fib\_bench\_3} of the main benchmark.
In variant \texttt{fib\_bench\_i} each worker thread performs the addition modulo $i$.
Hence, the first and the second worker performs the update $x\gets (x+y)\mod i$ and $y\gets (y+x)\mod i$, respectively.
For smaller values of $i$, we expect more write events to write the same value, and thus $\VCDPOR$ to benefit both in terms of the traces explored and the running time.
Although simple, this experiment serves the purpose of quantifying the value-centric gains of $\VCDPOR$ in a controlled benchmark. 
\cref{fig:fib_bench} depicts the obtained results for the three variants of \texttt{fib\_bench}, where $\text{Modulo}=\infty$ represents the original benchmark (i.e., without the modulo operation).
We see that indeed, as $i$ gets smaller, $\VCDPOR$ benefits significantly in both number of traces and running time.
Moreover, this benefit gets amplified with higher unroll bounds.

\begin{table}
\makebox[\linewidth][c]{
\setlength\tabcolsep{2pt}
\scriptsize
\centering
\begin{tabular}{|l ||r r r r r| r r r r r|}
\hline
\textbf{Benchmark} & \multicolumn{5}{c|}{\textbf{Maximal Traces}} & \multicolumn{5}{c|}{\textbf{Time}}  \\
\hline
 & \boldmath{$\VCDPOR$} & \boldmath{$\mathbf{\Source}$} & \boldmath{$\Optimal$} & \boldmath{$\OptimalObs$} & \boldmath{$\DCDPOR$} & \boldmath{$\VCDPOR$} & \boldmath{$\mathbf{\Source}$} & \boldmath{$\Optimal$} & \boldmath{$\OptimalObs$} & \boldmath{$\DCDPOR$} \\
\hline
parker(6) & \textbf{38670} & 1100917 & 1100917 & 1023567 & 985807 & \textbf{1m29s} & 23m5s & 24m29s & 24m54s & 46m41s \\
parker(7) & \textbf{52465} & 1735432 & 1735432 & 1613807 & 1554237 & \textbf{2m23s} & 41m28s & 44m41s & 45m13s & 1h27m \\
parker(8) & \textbf{68360} & 2576147 & 2576147 & 2395947 & 2307467 & \textbf{3m35s} & 1h9m & 1h15m & 1h17m & 2h29m \\
\hline
27\_Boop(6) & \textbf{248212} & 35079696 & 35079696 & 4750426 & 1468774 & \textbf{3m26s} & 2h54m & 2h49m & 26m22s & 12m33s \\
27\_Boop(7) & \textbf{420033} & - & - & 10134616 & 2874202 & \textbf{6m33s} & - & - & 1h0m & 27m21s \\
27\_Boop(8) & \textbf{677870} & - & - & 20003512 & 5268064 & \textbf{11m54s} & - & - & 2h7m & 56m13s \\
\hline
30\_Fun\_Point(6) & \textbf{5040} & 665280 & 665280 & 665280 & 665280 & \textbf{5.52s} & 4m2s & 4m14s & 4m36s & 1m34s \\
30\_Fun\_Point(7) & \textbf{40320} & 17297280 & 17297280 & 17297280 & 17297280 & \textbf{57.50s} & 2h7m & 2h15m & 2h29m & 51m46s \\
30\_Fun\_Point(8) & \textbf{362880} & - & - & - & - & \textbf{10m51s} & - & - & - & - \\
\hline
45\_monabsex(5) & \textbf{600} & 14400 & 14400 & 9745 & 6197 & \textbf{0.44s} & 2.28s & 2.36s & 1.86s & 1.50s \\
45\_monabsex(6) & \textbf{13152} & 518400 & 518400 & 291546 & 180126 & \textbf{14.93s} & 1m41s & 1m41s & 1m5s & 1m0s \\
45\_monabsex(7) & \textbf{423360} & 25401600 & 25401600 & 11710405 & 7073803 & \textbf{13m30s} & 1h43m & 1h40m & 51m57s & 56m16s \\
\hline
46\_monabsex(5) & \textbf{1064} & 14400 & 14400 & 5566 & 2653 & \textbf{0.32s} & 1.98s & 2.02s & 0.87s & 0.51s \\
46\_monabsex(6) & \textbf{21371} & 518400 & 518400 & 157717 & 62864 & \textbf{6.26s} & 1m29s & 1m23s & 28.04s & 10.33s \\
46\_monabsex(7) & \textbf{621948} & 25401600 & 25401600 & 6053748 & 2057588 & \textbf{4m9s} & 1h38m & 1h23m & 21m3s & 7m24s \\
\hline
fk2012\_true(3) & \textbf{12400} & 42144 & 42144 & 42144 & 33886 & \textbf{5.55s} & 9.34s & 10.59s & 11.08s & 13.13s \\
fk2012\_true(4) & \textbf{252586} & 1217826 & 1217826 & 1217826 & 888404 & \textbf{2m3s} & 5m6s & 5m35s & 6m11s & 6m30s \\
fk2012\_true(5) & \textbf{3757292} & 24580886 & 24580886 & 24580886 & 16494444 & \textbf{37m3s} & 2h0m & 2h12m & 2h26m & 2h28m \\
\hline
fkp2013\_true(5) & \textbf{17751} & 86400 & 86400 & 48591 & 25626 & \textbf{3.75s} & 16.40s & 15.20s & 9.70s & 4.90s \\
fkp2013\_true(6) & \textbf{513977} & 3628800 & 3628800 & 1672915 & 786499 & \textbf{2m18s} & 14m27s & 12m55s & 6m34s & 3m18s \\
fkp2013\_true(7) & \textbf{20043857} & - & - & - & 32244120 & \textbf{2h16m} & - & - & - & 3h11m \\
\hline
nondet-array(4) & \textbf{404} & 2616 & 2616 & 688 & 592 & \textbf{0.13s} & 0.88s & 0.80s & 0.27s & 0.20s \\
nondet-array(5) & \textbf{10804} & 128760 & 128760 & 18665 & 15449 & \textbf{3.11s} & 46.23s & 46.99s & 8.66s & 4.26s \\
nondet-array(6) & \textbf{430004} & 9854640 & 9854640 & 711276 & 571476 & \textbf{2m36s} & 1h15m & 1h14m & 7m45s & 3m30s \\
\hline
pthread-de(7) & \textbf{327782} & 4027216 & 4027216 & 4027216 & 829168 & \textbf{1m10s} & 12m9s & 13m32s & 17m36s & 2m12s \\
pthread-de(8) & \textbf{2457752} & 43976774 & 43976774 & 43976774 & 6984234 & \textbf{10m29s} & 2h29m & 2h46m & 3h24m & 22m1s \\
pthread-de(9) & \textbf{18568126} & - & - & - & 59287740 & \textbf{1h33m} & - & - & - & 3h37m \\
\hline
reorder\_5(5) & \textbf{1016} & 1755360 & 1755360 & 68206 & 4978 & \textbf{0.21s} & 9m0s & 9m22s & 26.45s & 0.34s \\
reorder\_5(8) & \textbf{247684} & - & - & - & 437725 & 1m47s & - & - & - & \textbf{1m29s} \\
reorder\_5(9) & \textbf{1644716} & - & - & - & 1792290 & 22m53s & - & - & - & \textbf{12m38s} \\
\hline
scull\_true(3) & \textbf{3426} & 617706 & 617706 & 436413 & 172931 & \textbf{19.77s} & 9m46s & 10m22s & 9m7s & 4m46s \\
scull\_true(4) & \textbf{8990} & 2732933 & 2732933 & 1840022 & 656100 & \textbf{1m7s} & 51m37s & 54m33s & 46m12s & 25m56s \\
scull\_true(5) & \textbf{19881} & 9488043 & 9488043 & 6070688 & 1988798 & \textbf{3m8s} & 3h29m & 3h42m & 2h54m & 1h47m \\
\hline
sigma\_false(7) & \textbf{12509} & 135135 & 135135 & 30952 & 30952 & \textbf{10.52s} & 55.87s & 1m0s & 18.65s & 17.87s \\
sigma\_false(8) & \textbf{133736} & 2027025 & 2027025 & 325488 & 325488 & \textbf{2m4s} & 16m21s & 18m45s & 4m12s & 3m44s \\
sigma\_false(9) & \textbf{1625040} & - & - & 3845724 & 3845724 & \textbf{31m53s} & - & - & 1h6m & 53m28s \\
\hline
check\_bad\_arr(5) & \textbf{4046} & 12838 & 12838 & 10989 & 6689 & 2.74s & 6.98s & 6.83s & 6.49s & \textbf{2.72s} \\
check\_bad\_arr(6) & \textbf{87473} & 357368 & 357368 & 307097 & 187377 & 1m47s & 5m21s & 4m36s & 4m24s & \textbf{1m33s} \\
check\_bad\_arr(7) & \textbf{1856332} & 8245810 & 8245810 & 6943293 & 4069592 & 2h11m & 3h9m & 2h19m & 2h12m & \textbf{1h7m} \\
\hline
32\_pthread5(1) & \textbf{20} & 24 & 24 & 24 & \textbf{20} & \textbf{0.05s} & \textbf{0.04s} & \textbf{0.04s} & 0.06s & 0.06s \\
32\_pthread5(2) & \textbf{1470} & 1890 & 1890 & 1806 & \textbf{1470} & 0.67s & \textbf{0.38s} & 0.45s & 0.54s & 0.67s \\
32\_pthread5(3) & \textbf{226800} & 302400 & 302400 & 280800 & \textbf{226800} & 2m30s & \textbf{1m14s} & 1m17s & 1m17s & 2m21s \\
\hline
fkp2014\_true(2) & \textbf{16} & \textbf{16} & \textbf{16} & \textbf{16} & \textbf{16} & \textbf{0.05s} & \textbf{0.05s} & \textbf{0.04s} & \textbf{0.04s} & \textbf{0.05s} \\
fkp2014\_true(3) & \textbf{1098} & \textbf{1098} & \textbf{1098} & \textbf{1098} & \textbf{1098} & 0.86s & \textbf{0.19s} & \textbf{0.20s} & 0.21s & 0.72s \\
fkp2014\_true(4) & \textbf{207024} & \textbf{207024} & \textbf{207024} & \textbf{207024} & \textbf{207024} & 3m40s & \textbf{39.84s} & 41.70s & 44.67s & 3m15s \\
\hline
singleton(8) & \textbf{2} & 40320 & 40320 & 8 & 8 & 0.06s & 14.92s & 15.24s & \textbf{0.04s} & 0.09s \\
singleton(9) & \textbf{2} & 362880 & 362880 & 9 & 9 & 0.09s & 2m31s & 2m32s & \textbf{0.05s} & 0.15s \\
singleton(10) & \textbf{2} & 3628800 & 3628800 & 10 & 10 & 0.16s & 27m33s & 28m9s & \textbf{0.05s} & 0.19s \\
\hline
stack\_true(9) & \textbf{48620} & \textbf{48620} & \textbf{48620} & \textbf{48620} & \textbf{48620} & 2m24s & \textbf{37.55s} & 38.47s & 40.06s & 2m23s \\
stack\_true(10) & \textbf{184756} & \textbf{184756} & \textbf{184756} & \textbf{184756} & \textbf{184756} & 11m58s & \textbf{2m31s} & 2m40s & 2m50s & 11m1s \\
stack\_true(11) & \textbf{705432} & \textbf{705432} & \textbf{705432} & \textbf{705432} & \textbf{705432} & 58m34s & \textbf{10m32s} & 11m8s & 11m48s & 54m42s \\
\hline
48\_ticket\_lock(2) & \textbf{6} & \textbf{6} & \textbf{6} & \textbf{6} & \textbf{6} & 0.05s & \textbf{0.03s} & \textbf{0.04s} & \textbf{0.04s} & 0.05s \\
48\_ticket\_lock(3) & \textbf{204} & \textbf{204} & \textbf{204} & \textbf{204} & \textbf{204} & 0.25s & \textbf{0.08s} & 0.10s & \textbf{0.09s} & 0.34s \\
48\_ticket\_lock(4) & \textbf{41400} & \textbf{41400} & \textbf{41400} & \textbf{41400} & \textbf{41400} & 55.67s & \textbf{13.88s} & 15.27s & 16.56s & 52.57s \\
\hline
\end{tabular}
}
\caption{Experimental comparison on SV-COMP benchmarks.}
\label{tab:experiments_svcomp}
\end{table}

\smallskip\noindent{\bf Benchmarks from SV-COMP.}
Here we present experiments on benchmarks from SV-COMP (along the industrial benchmark \texttt{parker}) (\cref{tab:experiments_svcomp}).
We have replaced all assertions with simple read events.
This way we ensure a fair comparison among all algorithms in exploring the trace-space of each benchmark,
as an assertion violation would halt the search.
We have verified that all assertion violations present in these benchmarks are detected by all algorithms before this modification.
The scaling parameter in each case controls the size of the input benchmark in terms of loop unrolls.

\begin{table}
\makebox[\linewidth][c]{
\setlength\tabcolsep{2pt}
\scriptsize
\centering
\begin{tabular}{|l ||r r r r r| r r r r r|}
\hline
\textbf{Benchmark} & \multicolumn{5}{c|}{\textbf{Maximal Traces}} & \multicolumn{5}{c|}{\textbf{Time}}  \\
\hline
 & \boldmath{$\VCDPOR$} & \boldmath{$\mathbf{\Source}$} & \boldmath{$\Optimal$} & \boldmath{$\OptimalObs$} & \boldmath{$\DCDPOR$} & \boldmath{$\VCDPOR$} & \boldmath{$\mathbf{\Source}$} & \boldmath{$\Optimal$} & \boldmath{$\OptimalObs$} & \boldmath{$\DCDPOR$} \\
\hline
\hline
rod\_cut\_td3(7) & \textbf{4324} & 102128 & 102128 & 51974 & 23143 & \textbf{33.23s} & 4m14s & 7m43s & 3m47s & 1m28s \\
rod\_cut\_td3(8) & \textbf{14744} & 508646 & 508646 & 257707 & 114624 & \textbf{3m4s} & 27m32s & 57m42s & 28m2s & 12m9s \\
rod\_cut\_td3(9) & \textbf{50320} & 2574752 & - & 1300067 & 577682 & \textbf{17m24s} & 3h0m & - & 3h27m & 1h39m \\
\hline
rod\_cut\_td4(3) & \textbf{1478} & 91592 & 91592 & 17451 & 4810 & \textbf{0.97s} & 1m29s & 1m49s & 21.79s & 1.46s \\
rod\_cut\_td4(4) & \textbf{21358} & 2459640 & 2459640 & 359609 & 85203 & \textbf{28.55s} & 1h6m & 1h33m & 14m2s & 57.94s \\
rod\_cut\_td4(5) & \textbf{433371} & - & - & - & 2551714 & \textbf{20m57s} & - & - & - & 1h22m \\
\hline
rod\_cut\_bu3(6) & \textbf{19933} & 183516 & 183516 & 147746 & 71670 & \textbf{56.15s} & 2m23s & 3m59s & 3m26s & 2m3s \\
rod\_cut\_bu3(7) & \textbf{99622} & 1101084 & 1101084 & 886466 & 429494 & \textbf{8m6s} & 17m52s & 33m33s & 29m19s & 21m40s \\
rod\_cut\_bu3(8) & \textbf{498061} & 6606492 & - & - & 2574902 & \textbf{1h6m} & 2h12m & - & - & 3h30m \\
\hline
rod\_cut\_bu4(2) & \textbf{1901} & 33912 & 33912 & 14667 & 5377 & \textbf{0.70s} & 11.76s & 13.36s & 6.75s & 1.15s \\
rod\_cut\_bu4(3) & \textbf{74541} & 2246424 & 2246424 & 913299 & 292633 & \textbf{46.95s} & 18m50s & 24m12s & 11m37s & 1m52s \\
rod\_cut\_bu4(4) & \textbf{3007476} & - & - & - & - & \textbf{1h17m} & - & - & - & - \\
\hline
lis\_bu3(8) & \textbf{118812} & 1744064 & 1744064 & 475986 & 358347 & \textbf{4m24s} & 33m22s & 1h0m & 18m27s & 7m24s \\
lis\_bu3(9) & \textbf{368400} & 7001792 & - & 1439130 & 1092553 & \textbf{15m49s} & 2h38m & - & 1h10m & 27m6s \\
lis\_bu3(10) & \textbf{3133740} & - & - & - & - & \textbf{3h59m} & - & - & - & - \\
\hline
lis\_bu4(2) & \textbf{1137} & 18522 & 18522 & 7936 & 2828 & \textbf{0.45s} & 8.45s & 9.41s & 4.42s & 0.52s \\
lis\_bu4(3) & \textbf{29931} & 1024002 & 1024002 & 364560 & 101766 & \textbf{12.70s} & 10m36s & 12m49s & 5m0s & 19.41s \\
lis\_bu4(4) & \textbf{1222278} & - & - & - & 5679067 & \textbf{16m34s} & - & - & - & 37m20s \\
\hline
coin\_all\_td3(9) & \textbf{4015} & 566214 & 566214 & 23308 & 8071 & 22.23s & 34m25s & 1h20m & 2m36s & \textbf{21.13s} \\
coin\_all\_td3(10) & \textbf{9052} & 2444048 & - & 59168 & 19829 & \textbf{1m2s} & 2h56m & - & 8m20s & 1m3s \\
coin\_all\_td3(19) & \textbf{637859} & - & - & - & 1528102 & \textbf{2h43m} & - & - & - & 3h5m \\
\hline
coin\_all\_td4(2) & \textbf{5938} & 6406248 & - & 74153 & 20668 & \textbf{4.86s} & 3h46m & - & 3m27s & 6.47s \\
coin\_all\_td4(3) & \textbf{68966} & - & - & 1549115 & 319142 & \textbf{1m36s} & - & - & 2h15m & 2m34s \\
coin\_all\_td4(5) & \textbf{379086} & - & - & - & 2857926 & \textbf{16m12s} & - & - & - & 36m32s \\
\hline
coin\_min\_td3(8) & \textbf{46535} & 1902262 & 1902262 & 981936 & 382275 & \textbf{3m0s} & 1h13m & 2h12m & 1h12m & 14m0s \\
coin\_min\_td3(9) & \textbf{154663} & - & - & - & 1634899 & \textbf{11m36s} & - & - & - & 1h8m \\
coin\_min\_td3(11) & \textbf{1312252} & - & - & - & - & \textbf{2h4m} & - & - & - & - \\
\hline
coin\_min\_td4(4) & \textbf{9912} & 1470312 & 1470312 & 208367 & 46634 & \textbf{30.52s} & 36m17s & 51m36s & 7m6s & 47.93s \\
coin\_min\_td4(5) & \textbf{102154} & - & - & 3534815 & 718883 & \textbf{6m7s} & - & - & 2h59m & 14m59s \\
coin\_min\_td4(6) & \textbf{1490420} & - & - & - & - & \textbf{1h52m} & - & - & - & - \\
\hline
bin\_nocon\_td3(7) & \textbf{13202} & 1664672 & 1664672 & 471151 & 121350 & \textbf{29.57s} & 48m34s & 1h26m & 26m11s & 2m4s \\
bin\_nocon\_td3(8) & \textbf{44802} & - & - & 2825725 & 603668 & \textbf{1m54s} & - & - & 3h17m & 12m32s \\
bin\_nocon\_td3(11) & \textbf{922114} & - & - & - & - & \textbf{1h0m} & - & - & - & - \\
\hline
bin\_nocon\_bu3(6) & \textbf{52500} & 773122 & 773122 & 115625 & 75000 & 1m15s & 12m37s & 19m44s & 3m5s & \textbf{1m8s} \\
bin\_nocon\_bu3(7) & \textbf{262500} & 5411854 & 5411854 & 578125 & 375000 & 7m27s & 1h45m & 2h52m & 19m54s & \textbf{6m50s} \\
bin\_nocon\_bu3(8) & \textbf{1312500} & - & - & 2890625 & 1875000 & 45m2s & - & - & 2h1m & \textbf{41m9s} \\
\hline
\end{tabular}
}
\caption{Experimental comparison on dynamic-programming benchmarks.}
\label{tab:experiments_dp}
\end{table}

\smallskip\noindent{\bf Dynamic-programming benchmarks.}
Here we present experiments on various multi-threaded dynamic-programming algorithms (\cref{tab:experiments_dp}).
For efficiency, these algorithms use memoization to avoid recomputing instances that correspond to the same sub-problem.
The benchmarks consist of three or four threads.
In each case, all-but-one threads are performing the dynamic programming computation, and one thread reads a flag signaling that the computation is finished, as well as the result of the computation.
Each benchmark name contains either the substring ``td'' or the substring``bu'', denoting that the dynamic programming table is computed top-down or bottom-up, respectively.
The scaling parameter of each benchmark controls the different sizes of the input problem.
The dynamic programming problems we use as benchmarks are the following.
\begin{compactitem}
\item \texttt{rod\_cut} computes, given one rod of a given length and prices for rods of shorter lengths,
the maximum profit achievable by cutting the given rod.
\item \texttt{lis} computes, given an array of non-repeating integers, the length of the
longest increasing subsequence (not necessarily contiguous) in the array.
\item \texttt{coin\_all} computes, given an unlimited supply of coins of given denominations,
the total number of distinct ways to get a desired change.
\item \texttt{coin\_min} computes, given an unlimited supply of coins of given denominations,
the minimum number of coins required to get a desired change.
\item \texttt{bin\_nocon} computes the number of binary strings of a given length that
do not contain the substring '11'.
\end{compactitem}

\begin{table}
\makebox[\linewidth][c]{
\setlength\tabcolsep{2pt}
\scriptsize
\centering
\begin{tabular}{|l ||r r r r r| r r r r r|}
\hline
\textbf{Benchmark} & \multicolumn{5}{c|}{\textbf{Maximal Traces}} & \multicolumn{5}{c|}{\textbf{Time}}  \\
\hline
 & \boldmath{$\VCDPOR$} & \boldmath{$\mathbf{\Source}$} & \boldmath{$\Optimal$} & \boldmath{$\OptimalObs$} & \boldmath{$\DCDPOR$} & \boldmath{$\VCDPOR$} & \boldmath{$\mathbf{\Source}$} & \boldmath{$\Optimal$} & \boldmath{$\OptimalObs$} & \boldmath{$\DCDPOR$} \\
\hline
\hline
tsay(2) & \textbf{2488} & 7469 & 7469 & 7469 & 7469 & \textbf{0.81s} & 2.46s & 2.76s & 2.99s & 1.82s \\
tsay(3) & \textbf{241822} & 1414576 & 1414576 & 1414576 & 1414576 & \textbf{1m38s} & 10m2s & 10m54s & 12m1s & 7m42s \\
tsay(4) & \textbf{24609389} & - & - & - & - & \textbf{3h51m} & - & - & - & - \\
\hline
peter\_fisch(2) & \textbf{1371} & 4386 & 4386 & 4386 & 4386 & \textbf{0.69s} & 1.56s & 1.61s & 1.73s & 1.16s \\
peter\_fisch(3) & \textbf{70448} & 430004 & 430004 & 430004 & 430004 & \textbf{34.03s} & 2m54s & 3m10s & 3m31s & 2m20s \\
peter\_fisch(4) & \textbf{3747718} & - & - & - & - & \textbf{41m31s} & - & - & - & - \\
\hline
peterson(5) & \textbf{86929} & 268706 & 268706 & 268706 & 256457 & \textbf{32.42s} & 49.22s & 54.60s & 1m4s & 1m32s \\
peterson(6) & \textbf{880069} & 3462008 & 3462008 & 3462008 & 3303617 & \textbf{7m10s} & 11m50s & 13m18s & 15m51s & 25m29s \\
peterson(7) & \textbf{9013381} & 45046254 & 45046254 & - & - & \textbf{1h30m} & 2h56m & 3h21m & - & - \\
\hline
lamport(2) & \textbf{958} & 3940 & 3940 & 2454 & 1456 & \textbf{0.39s} & 0.75s & 0.77s & 0.59s & 0.45s \\
lamport(3) & \textbf{57436} & 741370 & 741370 & 328764 & 130024 & \textbf{28.14s} & 2m24s & 2m43s & 1m29s & 52.24s \\
lamport(4) & \textbf{3723024} & - & - & - & 13088038 & \textbf{49m40s} & - & - & - & 2h26m \\
\hline
dekker(5) & \textbf{89647} & 435245 & 435245 & 435245 & 435245 & \textbf{29.78s} & 1m14s & 1m23s & 1m37s & 2m14s \\
dekker(6) & \textbf{932559} & 6745775 & 6745775 & 6745775 & 6745775 & \textbf{6m44s} & 21m36s & 24m12s & 28m22s & 42m46s \\
dekker(7) & \textbf{9837974} & - & - & - & - & \textbf{1h28m} & - & - & - & - \\
\hline
X2Tv6(3) & \textbf{7859} & 20371 & 20371 & 20371 & 20371 & \textbf{3.89s} & 5.35s & 5.68s & 6.58s & 7.69s \\
X2Tv6(4) & \textbf{152999} & 596354 & 596354 & 596354 & 596354 & \textbf{1m38s} & 3m6s & 3m23s & 3m47s & 5m17s \\
X2Tv6(5) & \textbf{3058189} & 17836411 & 17836411 & 17836411 & 17836411 & \textbf{46m41s} & 1h51m & 2h3m & 2h21m & 3h36m \\
\hline
kessels(3) & \textbf{8900} & 13856 & 13856 & 13856 & 13856 & \textbf{2.80s} & 5.07s & 5.45s & 5.98s & 3.70s \\
kessels(4) & \textbf{194858} & 323400 & 323400 & 323400 & 323400 & \textbf{1m13s} & 2m19s & 2m30s & 2m48s & 1m41s \\
kessels(5) & \textbf{4379904} & 7763704 & 7763704 & 7763704 & 7763704 & \textbf{35m59s} & 1h8m & 1h13m & 1h22m & 53m50s \\
\hline
X2Tv7(9) & \textbf{452142} & 2004774 & 2004774 & 2004774 & 2004774 & \textbf{7m34s} & 24m59s & 27m10s & 29m54s & 13m36s \\
X2Tv7(10) & \textbf{1721564} & 7708671 & 7708671 & 7708671 & 7708671 & \textbf{35m19s} & 1h47m & 1h58m & 2h10m & 1h1m \\
X2Tv7(11) & \textbf{6584004} & - & - & - & - & \textbf{2h37m} & - & - & - & - \\
\hline
X2Tv2(2) & \textbf{894} & 1293 & 1293 & 1293 & 1293 & \textbf{0.32s} & 0.46s & 0.46s & 0.51s & 0.50s \\
X2Tv2(3) & \textbf{42141} & 69316 & 69316 & 69316 & 69316 & \textbf{17.73s} & 29.21s & 31.04s & 34.65s & 22.01s \\
X2Tv2(4) & \textbf{1827915} & 3552837 & 3552837 & 3552837 & 3552837 & \textbf{17m21s} & 31m13s & 33m46s & 37m35s & 25m52s \\
\hline
burns(4) & \textbf{381} & 140380 & 140380 & 140380 & 140380 & \textbf{0.31s} & 1m24s & 1m28s & 1m37s & 1m8s \\
burns(5) & \textbf{1415} & 2916980 & 2916980 & 2916980 & 2916980 & \textbf{0.98s} & 35m29s & 38m9s & 41m55s & 29m25s \\
burns(11) & \textbf{4114995} & - & - & - & - & \textbf{1h48m} & - & - & - & - \\
\hline
burns3(1) & \textbf{67} & 849 & 849 & 849 & 849 & \textbf{0.09s} & 0.45s & 0.40s & 0.44s & 0.49s \\
burns3(2) & \textbf{11297} & 1490331 & 1490331 & 1490331 & 1490331 & \textbf{16.27s} & 16m49s & 17m32s & 20m4s & 26m4s \\
burns3(3) & \textbf{1638338} & - & - & - & - & \textbf{1h0m} & - & - & - & - \\
\hline
X2Tv10(2) & \textbf{4130} & 5079 & 5079 & 5079 & 5079 & 1.81s & 1.94s & 1.95s & 2.18s & \textbf{1.71s} \\
X2Tv10(3) & \textbf{213381} & 308433 & 308433 & 308433 & 308433 & \textbf{1m47s} & 2m15s & 2m26s & 2m39s & 1m56s \\
X2Tv10(4) & \textbf{10274441} & 17910500 & 17910500 & 17910500 & 17910500 & \textbf{1h58m} & 2h48m & 3h2m & 3h29m & 2h35m\\
\hline
X2Tv5(4) & \textbf{38743} & 46161 & 46161 & 46161 & 46161 & \textbf{14.34s} & 21.05s & 22.57s & 24.92s & 15.35s \\
X2Tv5(5) & \textbf{595527} & 730647 & 730647 & 730647 & 730647 & \textbf{4m37s} & 6m28s & 6m57s & 7m50s & 5m2s \\
X2Tv5(6) & \textbf{9312813} & 11755440 & 11755440 & 11755440 & 11755440 & \textbf{1h26m} & 2h2m & 2h17m & 2h33m & 1h37m \\
\hline
X2Tv1(6) & \textbf{224803} & 253042 & 253042 & 253042 & 253042 & 1m45s & 2m19s & 2m27s & 2m46s & \textbf{1m42s} \\
X2Tv1(7) & \textbf{1880095} & 2115302 & 2115302 & 2115302 & 2115302 & 18m4s & 21m56s & 23m59s & 26m35s & \textbf{17m31s} \\
X2Tv1(8) & \textbf{15873308} & 17857733 & 17857733 & - & 17857733 & 2h59m & 3h29m & 3h49m & - & \textbf{2h51m} \\
\hline
X2Tv8(3) & \textbf{6168} & 9894 & 9894 & 8700 & 8434 & 2.79s & \textbf{2.56s} & 2.63s & 2.64s & 3.15s \\
X2Tv8(4) & \textbf{122932} & 228417 & 228417 & 194206 & 186040 & 1m8s & \textbf{1m7s} & 1m13s & 1m10s & 1m30s \\
X2Tv8(5) & \textbf{2503292} & 5391534 & 5391534 & 4428748 & 4192466 & 31m12s & \textbf{31m4s} & 34m43s & 32m37s & 44m43s \\
\hline
X2Tv9(3) & \textbf{7234} & 7304 & 7304 & 7304 & 7304 & 2.53s & \textbf{2.11s} & 2.23s & 2.49s & 2.41s \\
X2Tv9(4) & \textbf{150535} & 153725 & 153725 & 153725 & 153725 & 1m3s & \textbf{52.80s} & 56.85s & 1m3s & 56.86s \\
X2Tv9(5) & \textbf{3261067} & 3324991 & 3324991 & 3324991 & 3324991 & 29m53s & \textbf{22m17s} & 24m10s & 27m11s & 27m10s\\
\hline
szymanski(3) & \textbf{27892} & 27951 & 27951 & 27951 & 27951 & 12.06s & \textbf{5.06s} & 5.66s & 6.69s & 9.81s \\
szymanski(4) & \textbf{395743} & 396583 & 396583 & 396583 & 396583 & 4m0s & \textbf{1m26s} & 1m39s & 1m49s & 3m14s \\
szymanski(5) & \textbf{5734528} & 5746703 & 5746703 & 5746703 & 5746703 & 1h17m & \textbf{25m17s} & 28m59s & 32m36s & 1h1m\\
\hline
\end{tabular}
}
\caption{Experimental comparison on mutual-exclusion benchmarks.}
\label{tab:experiments_me}
\end{table}

\smallskip\noindent{\bf Mutual-exclusion benchmarks.}
Here we present experiments on various mutual-exclusion algorithms from the literature (\cref{tab:experiments_me}).
In particular, we use the two-thread solutions of Dijkstra~\cite{Dijkstra83}, Kessels~\cite{Kessels82}, Tsay~\cite{Tsay98}, Peterson~\cite{Peterson81}, Peterson-Fischer~\cite{Peterson77}, Szymanski~\cite{Szymanski88}, Dekker~\cite{Knuth66},
as well as various solutions of Correia-Ramalhete~\cite{Correia16}.
In addition, we use the two-thread and three-thread versions of Burns's algorithm~\cite{Burns80}.
These protocols exercise a wide range of communication patterns, based, e.g., on the number of shared variables and the number of sequentially consistent stores/loads required to enter/leave the critical section.
In all these benchmarks, each thread executes the corresponding protocol to enter a (empty) critical section a number of times, the latter controlled by the scaling parameter.

\begin{table}
\makebox[\linewidth][c]{
\setlength\tabcolsep{2pt}
\scriptsize
\centering
\begin{tabular}{|l ||r r r r r| r r r r r|}
\hline
\textbf{Benchmark} & \multicolumn{5}{c|}{\textbf{Maximal Traces}} & \multicolumn{5}{c|}{\textbf{Time}}  \\
\hline
 & \boldmath{$\VCDPOR$} & \boldmath{$\mathbf{\Source}$} & \boldmath{$\Optimal$} & \boldmath{$\OptimalObs$} & \boldmath{$\DCDPOR$} & \boldmath{$\VCDPOR$} & \boldmath{$\mathbf{\Source}$} & \boldmath{$\Optimal$} & \boldmath{$\OptimalObs$} & \boldmath{$\DCDPOR$} \\
\hline
\hline
eratosthenes(5) & \textbf{3500} & 1527736 & 1527736 & 27858 & 19991 & \textbf{16.92s} & 18m37s & 20m39s & 41.14s & 1m29s \\
eratosthenes(7) & \textbf{29320} & - & - & 253792 & 189653 & \textbf{3m37s} & - & - & 9m29s & 19m41s \\
eratosthenes(8) & \textbf{110380} & - & - & 938756 & 710551 & \textbf{11m29s} & - & - & 42m27s & 1h4m \\
\hline
redundant\_co(2) & \textbf{11} & 1969110 & 1969110 & 5401 & 729 & \textbf{0.06s} & 7m16s & 7m32s & 1.51s & 0.07s \\
redundant\_co(8) & \textbf{35} & - & - & 1118305 & 35937 & \textbf{0.09s} & - & - & 13m24s & 0.97s \\
redundant\_co(9) & \textbf{39} & - & - & 1778221 & 50653 & \textbf{0.07s} & - & - & 23m49s & 1.35s \\
\hline
float\_read(9) & \textbf{9} & 3628800 & 3628800 & 2305 & 10 & \textbf{0.05s} & 26m30s & 26m38s & 1.27s & \textbf{0.04s} \\
float\_read(15) & \textbf{15} & - & - & 245761 & 16 & \textbf{0.65s} & - & - & 3m52s & 0.74s \\
float\_read(16) & \textbf{16} & - & - & 524289 & 17 & \textbf{1.42s} & - & - & 9m25s & 1.44s \\
\hline
opt\_lock(2) & \textbf{2497} & 69252 & 69252 & 11982 & 6475 & \textbf{1.50s} & 15.10s & 15.53s & 3.25s & 2.50s \\
opt\_lock(3) & \textbf{80805} & 15036174 & 15036174 & 416850 & 212877 & \textbf{52.13s} & 1h5m & 1h9m & 2m9s & 1m29s \\
opt\_lock(4) & \textbf{2543298} & - & - & 14038926 & 6743831 & \textbf{37m41s} & - & - & 1h27m & 1h2m \\
\hline
\end{tabular}
}
\caption{Experimental comparison on individual benchmarks.}
\label{tab:experiments_ind}
\end{table}

\smallskip\noindent{\bf Individual benchmarks.}
Here we present experiments on individual benchmarks (\cref{tab:experiments_ind}):
\texttt{eratosthenes} consists of two threads computing the sieve of Eratosthenes in parallel;
\texttt{redundant\_co} consists of three threads, two of which repeatedly write to a variable and one reads from it;
\texttt{float\_read} consists of several threads, each writing once to a variable, and one reading from it (adapted from~\cite{Aronis18});
\texttt{opt\_lock} consists of three threads in an optimistic-lock scheme.
The scaling parameter controls the size in terms of loop unrolls.

\begin{table}
\makebox[\linewidth][c]{
\setlength\tabcolsep{2pt}
\scriptsize
\centering
\begin{tabular}{|l |r r r r ||l |r r r r||l |r r r r |}
\hline
\textbf{Benchmark} & $\mathsf{LOC}$ & $\mathsf{Var}$ & $\mathsf{Locks}$ & $\mathsf{Threads}$ & \textbf{Benchmark} & $\mathsf{LOC}$ & $\mathsf{Var}$ & $\mathsf{Locks}$ & $\mathsf{Threads}$ & \textbf{Benchmark} & $\mathsf{LOC}$ & $\mathsf{Var}$ & $\mathsf{Locks}$ & $\mathsf{Threads}$\\
\hline
parker & 134 & 4 & 0 & 2 &    48\_ticket\_lock & 52 & 3 & 1 & U&    dekker & 91 & 4 & 0 & 2 \\
27\_Boop & 74 & 4 & 0 & 4 &   rod\_cut\_td3 & 50 & 51 & 0 & 3 &   X2Tv6 & 75 & 4 & 0 & 2 \\
30\_Fun\_Point & 67 & 1 & 1 & U &    rod\_cut\_td4 & 62 & 51 & 0 & 4 & kessels & 44 & 3 & 0 & 2 \\
45\_monabsex & 24 & 1 & 0 & U &    rod\_cut\_bu3 & 36 & 51 & 0 & 3&    X2Tv7 & 83 & 3 & 0 & 2 \\
46\_monabsex & 22 & 2 & 0 & U &   rod\_cut\_bu4 & 37 & 51 & 0 & 4 &    X2Tv2 & 65 & 3 & 0 & 2\\
fk2012\_true & 100 & 1 & 2 & 3 &   lis\_bu3 & 47 & 51 & 0 & 3 &   burns & 70 & 3 & 0 & 2 \\
fkp2013\_true & 26 & 1 & 0 & U &    lis\_bu4 & 48 & 51 & 0 & 4 &  burns3 & 70 & 4 & 0 & 3 \\
nondet-array & 29 & 1 & 0 & U &   coin\_all\_td3 & 51 & 151 & 0 & 3 &   X2Tv10 & 91 & 3 & 0 & 2\\
pthread-de & 67 & 1 & 1 & U &     coin\_all\_td4 & 53 & 151 & 0 & 4 &  X2Tv5 & 55 & 4 & 0 & 2 \\
reorder\_5 & 1227 & 4 & 0 & U &    coin\_min\_td3 & 46 & 51 & 0 & 3&   X2Tv1 & 56 & 3 & 0 & 2 \\
scull\_true & 389 & 7 & 1 & 3 &   coin\_min\_td4 & 52 & 51 & 0 & 4 &  X2Tv8 & 64 & 4 & 0 & 2\\
sigma\_false & 36 & 1 & 0 & U &    bin\_nocon\_td3 & 43 & 101 & 0 & 3 &   X2Tv9 & 61 & 3 & 0 & 2 \\
check\_bad\_arr & 33 & 1 & 0 & U &    bin\_nocon\_bu3 & 53 & 101 & 0 & 3&    szymanski & 93 & 3 & 0 & 2 \\
32\_pthread5 & 87 & 4 & 1 & U &   tsay & 54 & 3 & 0 & 2 &  eratosthenes & 25 & U & 0 & 2\\
fkp2014\_true & 36 & 2 & 1 & U &    peter\_fisch & 59 & 3 & 0 & 2 & redundant\_co & 23 & 1 & 0 & 2 \\
singleton & 43 & 1 & 0 & U &    peterson & 68 & 4 & 0 & 2 &   float\_read & 25 & 1 & 0 & U \\
stack\_true & 104 & U & 1 & 2 &  lamport & 83 & 5 & 0 & 2 & opt\_lock & 31 & 2 & 0 & 3\\
\hline

\end{tabular}
}
\caption{Benchmark statistics.}
\label{tab:bench_stats}
\end{table}

\smallskip\noindent{\bf Summary.}
For the sake of completeness, we refer to \cref{tab:bench_stats} for some statistics on our benchmark set.
Entries marked with ``U'' denote that the corresponding parameter is controlled by the unroll bound of the respective benchmark.
In a variety of cases, the $\VHB$ partitioning is significantly coarser than each of the partitionings constructed by the other algorithms.
This coarseness makes $\VCDPOR$ more efficient in its exploration than the alternatives.
We note that in some cases, $\VHB$ offers little-to-no reduction, and then $\VCDPOR$ becomes slower than the alternatives,
due to the overhead incurred in constructing $\VHB$.
For example, for the benchmark \texttt{reorder\_5} of \cref{tab:experiments_svcomp}, the partitioning reduction achieved by $\VCDPOR$ is large enough compared to $\Source$, $\Optimal$ and $\OptimalObs$ that makes $\VCDPOR$ significantly faster than each of these techniques.
However, although the partitioning of $\VCDPOR$ is smaller than $\DCDPOR$, the corresponding reduction is not large enough to make $\VCDPOR$ faster than $\DCDPOR$ in this benchmark (in general, $\VCDPOR$ has a larger polynomial overhead than $\DCDPOR$.)
Similarly, for the benchmark \texttt{X2Tv9} of \cref{tab:experiments_me}, the reduction of the $\VHB$ partitioning is quite small, and although $\Source$ is the slowest algorithm in theory, its more lightweight nature makes it faster in practice for this benchmark.
Finally, we also identify benchmarks such as \texttt{stack\_true} and \texttt{48\_ticket\_lock} where there is no trace reduction at all, and are better handled by existing methods.
We note that our approach is fairly different from the literature, and our implementation of $\VCDPOR$ still largely unoptimized.
We  identify potential for improving the performance of $\VCDPOR$ by improving the closure computation, as well as reducing (or eliminating) the number of non-maximal traces explored by the algorithm.

\section{Related Work and Conclusions}\label{sec:related}

The formal analysis of concurrent programs is a major challenge in verification, and has been a subject of 
extensive research~\cite{Petri62,Cadiou73,Lipton75,Clarke86,Lal09,Farzan12,Farzan09}.
Since it is hard to reproduce bugs by testing due to scheduling nondeterminism,
systematic state space exploration by model checking is an important approach for the 
problem~\cite{Godefroid05,Musuvathi07,Andrews04,Clarke00,Alglave13}.
In this direction, stateless model checking has been employed to combat state-space explosion~\cite{G96,Godefroid97,Godefroid05,Musuvathi07b}.

To deal with the exponential number of interleavings faced by 
the early model checking~\cite{Godefroid97}, several reduction techniques have been proposed 
such as POR and context bounding~\cite{Peled93,Musuvathi07}.
Several POR methods, based on persistent set~\cite{Clarke99,G96,Valmari91}
and sleep set techniques~\cite{Godefroid97}, have been studied.
DPOR techniques were first proposed in~\cite{Flanagan05}, and several variants 
and improvements have been made since~\cite{Sen06,Lauterburg10,Tasharofi12,Saarikivi12,Sen06b}.
In~\cite{Abdulla14}, source sets and wakeup trees were developed to make DPOR optimal,
and the underlying computational problems were further studied in~\cite{Nguyen18}.
Besides the present work, further improvements over optimal DPOR have been made in~\cite{Aronis18,Chalupa17},
as well as with maximal causal models~\cite{HUANG15,Huang017}.
Other techniques such as unfoldings have also been explored~\cite{Kahkonen12,McMillan95,Sousa15}.
Techniques for POR have also been applied to relaxed memory models~\cite{Wang08,Abdulla2015,Kokologiannakis17,Demsky15,Huang16} and message passing programs~\cite{G96,Godefroid95,Katz92}.

In this work, we have introduced a new equivalence on traces, called the value-happens-before equivalence $\VHB$, which considers the values of trace events in order to determine whether two traces are equivalent.
We have shown that $\VHB$ is coarser than the standard happens-before equivalence, which is the theoretical foundation of the majority of DPOR algorithms.
In fact, this coarsening occurs even when there are no concurrent write events.
In addition, we have developed an algorithm $\VCDPOR$ that relies on $\VHB$ to partition the trace space into equivalence classes and explore each class efficiently.
Our experiments show that, in a variety of benchmarks, $\VHB$ indeed produces smaller partitionings than those explored by alternative, state-of-the-art methods, which often leads to a large reduction in running times.

\begin{acks}
The authors would also like to thank anonymous referees for their valuable comments and helpful suggestions.
This work is supported
by the \grantsponsor{FWF}{Austrian Science Fund (FWF) NFN}{} grants \grantnum{}{S11407-N23~(RiSE/SHiNE)} and \grantnum{}{S11402-N23~(RiSE/SHiNE)},
by the \grantsponsor{WWTF}{Vienna Science and Technology Fund (WWTF)}{} Project \grantnum{}{ICT15-003}, and
by the \grantsponsor{FWF}{Austrian Science Fund (FWF) Schrodinger}{} grant \grantnum{}{J-4220}.
\end{acks}

\bibliography{bibliography}

\newpage
\appendix
\section{Details of \cref{sec:ave}}\label{sec:app_equivalence_proofs}

In this section we prove \cref{them:comparison}.
We start with the following remark.

\smallskip
\begin{remark}\label{rem:value_hb}
For every pair of traces $\Trace_1, \Trace_2\in \TraceSpace$, if $\Value_{\Trace_1}\neq \Value_{\Trace_2}$ or $\SideAnnotation_{\Trace_1}\neq \SideAnnotation_{\Trace_2}$ then $\Observation_{\Trace_1}\neq \Observation_{\Trace_2}$ and $\HB{}{\Trace_1}{}\neq \HB{}{\Trace_2}{}$.
\end{remark}

We now prove \cref{them:comparison}.

\smallskip
\themcomparison*
\begin{proof}
The fact that $\VHB$ is sound follows from \cref{rem:soundness}.
Here we prove that that $\VHB$ is at least as coarse as $\Maz$.
Afterwards, we present two examples where $\VHB$ can, in fact, be exponentially coarser.

Consider two traces $\Trace_1, \Trace_2\in \TraceSpace$ such that $\Trace_1\not \VHBE \Trace_2$.
If $\Events{\Trace_1}\neq \Events{\Trace_2}$ then $\Trace_1\not \MazE\Trace_2$.
Else, if $\Value_{\Trace_1}\neq \Value_{\Trace_2}$ or $\SideAnnotation_{\Trace_1}\neq \SideAnnotation_{\Trace_2}$, by \cref{rem:value_hb} we have $\HB{}{\Trace_1}{}\neq \HB{}{\Trace_2}{}$.
Else, if $\CHB{}{\Trace_1}{}\neq \CHB{}{\Trace_2}{}$, there exists a read event such that $\Observation_{\Trace_1}(\Read)\neq \Observation_{\Trace_2}(\Read)$, which implies that $\HB{}{\Trace_1}{}\neq \HB{}{\Trace_2}{}$.
Finally, if $\HB{}{\Trace_1}{}\Project{\SysLeafEvents} \neq \HB{}{\Trace_2}{}\Project{\SysLeafEvents}$ then trivially $\HB{}{\Trace_1}{}\neq \HB{}{\Trace_2}{}$.
Hence, in all cases we obtain $\Trace_1\not \MazE \Trace_2$.
The desired result follows.
\end{proof}

\section{Details of \cref{sec:closure}}\label{sec:app_closure_proofs}

In this section we present details of \cref{sec:closure}.
We first give the formal proofs of \cref{lem:closed_linearizable}, \cref{lem:closure_defined} and \cref{lem:realizable_iff_feasible} that state properties of closed annotated partial orders.
Afterwards, we present our algorithm $\ClosureAlgo$ that computes the closure of an annotated partial order, and establish its correctness and complexity.

\smallskip\noindent{\bf Three lemmas on closed annotated partial orders.}
We start with \cref{lem:closed_linearizable}, which states that a closed annotated partial order is realizable.

\smallskip
\lemclosedlinearizable*
\begin{proof}
Let $\AnnotatedPO=(X_1, X_2, \Value, P, \SideAnnotation, \GoodWrites)$, and we construct a linearization $\Trace$ of $P$ as follows.
\begin{compactenum}
\item Create a partial order $Q$ as follows. 
\begin{compactenum}
\item For every pair of events $\Event_1, \Event_2$ with $\Event_1<_{P}\Event_2$, we have $\Event_1<_{Q}\Event_2$.
\item For every pair of events $\Event_1, \Event_2$ with $\Event_i\in X_i$ for each $i\in [2]$, if $\Event_2\not <_{P} \Event_1$ then $\Event_1<_{Q}\Event_2$.
\end{compactenum}
\item Create $\Trace$ by linearizing $Q$ arbitrarily.
\end{compactenum}
It is easy to see that since $\Width{P\Project X_1}=1$, $Q$ is indeed a partial order and thus $\Trace$ is well defined.
In addition, the above process takes $O(\Poly(n))$ time.
We now argue that $\Trace$ is indeed a witness trace.
It is clear that $Q\Refines P$ and thus $\Trace$ is a linearization of $P$.
It remains to argue that for every read event $\Read\in \Reads{X}$, we have that $\Observation_{\Trace}\in \GoodWrites(\Read)$.
We distinguish between the following cases.
\begin{compactenum}
\item $\Read\in X_1$. Let $\Write = \Observation_{\Trace}(\Read)$, and observe that $\Write<_{P}\Read$.
Assume towards contradiction that $\Write\in \BadWrites(\Read)$.
If $\SideAnnotation(\Read)=1$, by \cref{item:closure3} of closure we have that there exists a write event $\Write'\in X_1$ such that $\Write<_{P}\Write'$.
Since $\SideAnnotation(\Read)=1$, we have $\Write'<_{\TO} \Read$  thus $\Write'<_{P}\Read$ and $\Write\not \in \VisibleWrites_{P}(\Read)$, a contradiction.
Otherwise, $\SideAnnotation(\Read)=2$, and by \cref{item:closure1} of closure, there exists a write event $\Write' \in \GoodWrites(\Read)\cap \VisibleWrites_{P}(\Read)\cap X_2$ such that $\Write'<_{P} \Read$.
Observe that in this case $\Write=\Write'$, a contradiction.

\item $\Read\in X_2$.
Let $\Write=\Observation_{\Trace}(\Read)$, and observe that $\Write\in \TailWrites_{P}(\Read)$.
Assume towards contradiction that $\Write\in\BadWrites(\Read)$.
By \cref{item:closure3} of closure, we have that $\Write\not <_{P} \Read$.
In this case $\Write \in X_1$, and since $\Observation_{\Trace}(\Read)=\Write$, there exists no $\Write'\in X_2\cap \GoodWrites_{P}(\Read)\cap \VisibleWrites_{P}(\Read)$.
It followed that $|\TailWrites_{P}(\Read)|=1$, and by \cref{item:closure2} of closure we have that $\Write\in\GoodWrites(\Read)$, a contradiction.
\end{compactenum}
The desired result follows.
\end{proof}

We continue with \cref{lem:closure_defined} which states that the closure of an annotated partial order is unique.

\smallskip
\lemclosuredefined*
\begin{proof}
Assume towards contradiction otherwise, and let $Q_1, Q_2$ be two weakest partial orders (i.e., $Q_i\not \Refines Q_{3-i}$ for each $i\in[2]$) with the stated properties.
Let $Q=Q_1\cap Q_2$, thus $Q_1, Q_2\Refines Q$, and we argue that $(X_1, X_2, Q, \Value, \SideAnnotation, \GoodWrites)$ is closed.
Let $X=X_1\cup X_2$ and consider any read event $\Read\in \Reads{X}$, and we show that each of closure conditions holds for $\Read$.

\begin{compactenum}
\item First, assume that for some $i\in [2]$ there exists a write event $\Write_i\in \GoodWrites(\Read)\cap \HeadWrites_{Q_i}(\Read)\cap X_{\SideIndicator_{\AnnotatedPO}(\Read)}$.
Since $Q_i\Refines Q$, we have that $\Write_i\in \HeadWrites_{Q}(\Read)$ and thus \cref{item:closure1} of closure is satisfied.
Otherwise, for each $i\in [2]$ there exists a write event $\Write_i\in \GoodWrites(\Read)\cap \HeadWrites_{Q_i}(\Read)\cap X_{3-\SideIndicator_{\AnnotatedPO}(\Read)}$ such that $\Write_i<_{Q_i} \Read$.
Since $\MWidth{P\Project X_{3-\SideIndicator_{\AnnotatedPO}(\Read)}}=1$, we have that $\Write_i<_{Q} \Write_{3-i}$ for some $i\in[2]$,
and thus $\Write_i<_{Q} \Read$.
Finally, since $Q_i\Refines Q$ we have $\Write_i\in \VisibleWrites_{Q}(\Read)$ and thus $\Write_i\in \HeadWrites_{Q}(\Read)$.

\item First, assume that for some $i\in[2]$ there exists a write event $\Write_i\in \GoodWrites(\Read)\cap \TailWrites_{Q_i}(\Read)\cap X_{\SideIndicator_{\AnnotatedPO}}(\Read)$.
Since $Q_i\Refines Q$, we have that $\Write_i\in \TailWrites_{Q}(\Read)$ and thus \cref{item:closure2} of closure is satisfied.
Otherwise, for each $i\in [2]$ there exists a write event $\Write_i\in \GoodWrites(\Read)\cap \TailWrites_{Q_i}(\Read)\cap X_{3-\SideIndicator_{\AnnotatedPO}}(\Read)$ such that $\Write_i<_{Q_i} \Read$.
Since $\MWidth{P\Project X_{3-\SideIndicator_{\AnnotatedPO}(\Read)}}=1$, we have that $\Write_{3-i}<_{Q} \Write_{i}$ for some $i\in[2]$.
Since $Q_i\Refines Q$, we have $\Write_i \in \VisibleWrites_{Q}(\Read)$ and it remains to argue that $\Write_i\in \TailWrites_{Q}(\Read)$.
Indeed, if that is not the case then there exists a write event $\Write'\in \TailWrites_{Q}(\Read)$ such that $\Write_i<_{Q} \Write$.
But then $\Write_i<_{Q_j}\Write$ for each $j\in [2]$ and since $\Write\not \in \TailWrites_{Q_j}$, we have $\Read<_{Q_j} \Write$ for each $j\in [2]$. Hence $\Read<_{Q} \Write$, a contradiction.

\item Consider any write event $\Write'\in \BadWrites(\Read)\cap \HeadWrites_{Q}(\Read)$ such that $\Write'<_{Q} \Read$, and we have $\Write'<_{Q_i}\Read$ for each $i\in [2]$.

First assume that  $\SideIndicator_{\AnnotatedPO}(\Write)=\SideIndicator_{\AnnotatedPO}(\Read)$.
It follows that for each $i\in[2]$ there exists a write event $\Write_i \in \GoodWrites(\Read)\cap \TailWrites_{Q_i}(\Read)\cap X_{3-\SideIndicator_{\AnnotatedPO}}(\Read)$ such that $\Write'<_{Q} \Write_i$.
Since $\MWidth{P\Project X_{3-\SideIndicator_{\AnnotatedPO}}(\Read)}=1$, we have $\Write_{3-i}<_{P} \Write_i$ for some $i\in [2]$,
and thus $\Write'<_{Q_j} \Write_i$ for each $j\in [2]$. Hence $\Write'<_{Q} \Write_i$.
Since for each $j\in [2]$ we have $Q_j\Refines Q$, it is $\Write_i\in \VisibleWrites_{Q}(\Read)$, as desired.

Finally, assume that $\SideIndicator_{\AnnotatedPO}(\Write)=3-\SideIndicator_{\AnnotatedPO}(\Read)$.
If for some $i\in [2]$ there exists a write event $\Write_i\in \GoodWrites(\Read)\cap \VisibleWrites_{Q_i}\cap X_{\SideIndicator_{\AnnotatedPO}}(\Write')$ such that $\Write'<_{Q_i} \Write$, since $Q_i\Refines Q$ and $\MWidth{P\Project X_{\SideIndicator_{\AnnotatedPO}}(\Write')}=1$ we have $\Write_i\in \VisibleWrites_{Q}(\Read)$ and $\Write'<_{Q} \Write$ as desired.
Otherwise, due to \cref{item:closure2} of closure it follows that for the unique write event $\Write\in  X_{\SideIndicator_{\AnnotatedPO}}(\Read)\cap \VisibleWrites_{Q}(\Read)$ we have $\Write\in \GoodWrites(\Read)$.
\end{compactenum}
It follows that $Q$ is closed, a contradiction.
The desired result follows.
\end{proof}

Finally, we prove \cref{lem:realizable_iff_feasible} which states that an annotated partial order is realizable if and only if it is feasible (i.e., it has a closure).

\smallskip
\lemrealizableifffeasible*
\begin{proof}
Let $\AnnotatedPO=(X_1, X_2, P, \Value, \SideAnnotation, \GoodWrites)$. We prove each direction separately.

\noindent{$(\Rightarrow)$.}
If $\AnnotatedPO$ is feasible, let $\AnnotatedPOQ=(X_1, X_2, Q, \Value, \SideAnnotation, \GoodWrites)$ be the closure of $\AnnotatedPO$.
Since $\AnnotatedPOQ$ is closed, by \cref{lem:closed_linearizable} we have that $\AnnotatedPOQ$ is linearizable to a trace $\Trace$.
Since $Q\Refines P$, we have that $\Trace$ is also a linearization of $\AnnotatedPO$.

\noindent{$(\Leftarrow)$.}
If $\AnnotatedPO$ is realizable, there exists a trace $\Trace$ such that $\Trace \Refines P$ and for every read event $\Read\in \Reads{\Trace}$ we have $\Observation_{\Trace}(\Read)\in \GoodWrites(\Read)$. We can view $\Trace$ as a partial (total) order, and observe that the annotated partial order
$\AnnotatedPOF=(X_1, X_2, \Trace, \Value, \SideAnnotation, \GoodWrites)$ is closed. Hence $P$ is feasible.

The desired result follows.
\end{proof}

\smallskip{\em Correctness and complexity of $\ClosureAlgo$.}
Hence we argue about the correctness and complexity of $\ClosureAlgo$.
We start with the following straightforward lemma, which captures the complexity.

\smallskip
\begin{restatable}{lemma}{lemclosurecomplexity}\label{lem:closure_complexity}
$\ClosureAlgo$ requires $O(\Poly(n))$ time.
\end{restatable}
\begin{proof}
Let $n=|X_1 \cup X_2|$.
It is straightforward to see that testing whether $\AnnotatedPOQ$ violates any of the closure rules in \cref{line:closure_r1}, \cref{line:closure_r2} and \cref{line:closure_r3} requires polynomial time in $n$.
Every time one of these rules is violated, $\ClosureAlgo$ strengthens $\AnnotatedPOQ$ by inserting some new orderings in $\AnnotatedPOQ$.
Since $\ClosureAlgo$ can insert at most $n^2$ such new orderings, it follows that the running time of $\ClosureAlgo$ is $O(\Poly(n))$.
\end{proof}

We now turn our attention to the correctness of $\ClosureAlgo$. We establish the following lemma.

\smallskip
\begin{restatable}{lemma}{lemclosurecorrectness}\label{lem:closure_correctness}
The following assertions hold.
\begin{compactenum}
\item\label{item:closure_cor1} If $\ClosureAlgo(\AnnotatedPO)$ returns $\AnnotatedPOQ\neq \bot$ then $\AnnotatedPOQ$ is the closure of $\AnnotatedPO$.
\item\label{item:closure_cor2} If $\ClosureAlgo(\AnnotatedPO)$ returns $\bot$ then $\AnnotatedPO$ is not feasible.
\end{compactenum}
\end{restatable}
\begin{proof}

\noindent{\em Invariant.}
We first show that the following invariant holds at all times:
if $\AnnotatedPO$ has a closure $\AnnotatedPOF=(X_1, X_2, F, \Value, \SideAnnotation)$ then $F\Refines Q$.
The claim holds trivially in the beginning of $\ClosureAlgo$ since $Q=P$. 
Now assume that the algorithm inserts an ordering $\Event_1\to \Event_2$ in $Q$, let $Q'$ be the resulting partial order, and we will argue that $F\Refines Q'$.
By the induction hypothesis, we have that $F\Refines Q$.
We split cases based on which closure rule inserted the ordering $\Event_1\to \Event_2$.
\begin{compactenum}
\item $\RuleOneAlgo(\Read)$.
In this case $\Event_2=\Read$ and $\Event_1=\Write$ as instantiated in \cref{line:ruleonealgo_write} of \cref{algo:rule1}.
By \cref{item:closure1} of closure, there exists a write event $\Write'\in \GoodWrites(\Read)\cap \HeadWrites_{F}(\Read)$ such that $\Write'<_{F}\Read$.
By the induction hypothesis, we have that $F\Refines Q$, thus $\Write'\in \VisibleWrites_{Q}(\Read)$.
Observe that since $\MWidth{P\Project X_1}=\MWidth{P\Project X_2}=1$ and the rule is violated, we have that the set $Y=\GoodWrites(\Read)\cap \VisibleWrites_{Q}(\Read)$ is totally ordered in $Q$, thus $\Write\leq_{Q} \Write'$, and thus $\Write<_{F} \Read$, as desired.

\item $\RuleTwoAlgo(\Read)$.
In this case $\Event_1=\Read$ and $\Event_2=\Write$ as instantiated in \cref{line:ruleonealgo_write} of \cref{algo:rule2}.
By \cref{item:closure2} of closure, there exists a write event $\Write'\in\TailWrites_{F}\cap\GoodWrites(\Read)$.
By the induction hypothesis, we have that $F\Refines Q$, thus $\Write'\in \VisibleWrites_{Q}(\Read)$.
Observe that $\SideIndicator_{\AnnotatedPO}(\Write')=X_{3-\SideIndicator_{\AnnotatedPO}(\Read)}$, otherwise since $\MWidth{P\Project X_{\SideIndicator_{\AnnotatedPO}(\Read)}}=1$ we would have $\Write'\in \TailWrites_{Q}(\Read)$ and thus \cref{item:closure2} of closure would not be violated.
Since $\MWidth{P\Project X_{3-\SideIndicator_{\AnnotatedPO}}(\Read)}=1$, it follows that $\Write'<_{Q}\Write$ and thus $\Read<_{F} \Write$, as desired.

\item $\RuleThreeAlgo(\Read)$.
In this case $\Event_1=\ov{\Write}$ and $\Event_2=\Write$ as instantiated in \cref{line:rulethreealgo_write} and \cref{line:rulethreealgo_writep} of \cref{algo:rule3}, respectively.
Observe that at this point \cref{item:closure1} of closure is not violated for $\Read$, and thus $|\HeadWrites_{Q}(\Read)\cap \BadWrites_{Q}(\Read)|=1$, and $\ov{\Write}$ is the unique event in $\HeadWrites_{Q}(\Read)\cap \BadWrites(\Read)$.
By \cref{item:closure3} of closure, either $\ov{\Write}\not \in \HeadWrites_{F}(\Read)$, or there exists a write event $\Write'\in \GoodWrites(\Read)\cap \VisibleWrites_{F}(\Read)$ such that $\ov{\Write}<_{F} \Write'$.
Since $\MWidth{F\Project X_1}=\MWidth{F\Project X_2}=1$, it is easy to verify that in both cases there exists a write event $\Write''\in \TailWrites_{F}(\Read)\cap X_{3-\SideIndicator_{\AnnotatedPO}(\ov{\Write})}$ such that $\ov{\Write}<_{F} \Write''$ and $\Write''\leq_{F} \Write$, thus $\ov{\Write}<_{F}\Write$, as desired.
\end{compactenum}

\noindent{\em Main proof.}
We are now ready to prove the lemma.
We examine each item separately.
\begin{compactenum}
\item If $\ClosureAlgo(\AnnotatedPO)$ returns $\AnnotatedPOQ=(X_1, X_2, Q, \Value, \SideAnnotation, \GoodWrites)$ then we have that $\AnnotatedPOQ$ is closed and $Q\Refines P$. It follows that the closure of $\AnnotatedPO$ exists, and the above invariant establishes that $\AnnotatedPOQ$ is the closure of $\AnnotatedPO$.
\item If $\ClosureAlgo(\AnnotatedPO)$ returns $\bot$, then at some point the algorithm discovers a read event $\Read$ such that $\GoodWrites(\Read)\cap \VisibleWrites_{Q}(\Read)=\emptyset$. Assume towards contradiction that $\AnnotatedPO$ is feasible and $\AnnotatedPOF=(X_1, X_2, F, \Value, \SideAnnotation,\GoodWrites)$ is the closure of $\AnnotatedPO$. By our invariant above, it follows that $F\Refines Q$. But then $\GoodWrites(\Read)\cap \VisibleWrites_{F}(\Read)=\emptyset$, which contradicts \cref{item:closure1} of closure, a contradiction.
\end{compactenum}

The desired result follows.
\end{proof}

Finally, we prove \cref{them:closure} which concludes the results of \cref{sec:closure}.

\smallskip
\themclosure*
\begin{proof}
By \cref{lem:realizable_iff_feasible}, $\AnnotatedPO$ is realizable if and only if it is feasible.
By \cref{lem:closure_correctness} the algorithm $\ClosureAlgo(\AnnotatedPO)$ runs in $O(\Poly(n))$ time and returns the annotated partial order $\AnnotatedPOQ$ that is the closure of $\AnnotatedPO$ if and only if $\AnnotatedPO$ is feasible.
If $\AnnotatedPO$ is realizable, \cref{lem:closed_linearizable} provides a simple construction of a witness trace in $O(\Poly(n))$ time.
\end{proof}

\section{Details of \cref{{sec:vcdpor}}}\label{sec:app_vcdpor_proofs}

In this section we present details of \cref{sec:vcdpor}.
We first outline our algorithm $\ExtendPO$ that takes as input an annotated partial order $\AnnotatedPO$, 
and extends it to a new set of events.
This operation is central to our $\VCDPOR$ which performs a search of the trace space based on annotated partial orders, and $\ExtendPO$ is used to  add events to such partial orders.
Afterwards, we prove the correctness and complexity of $\VCDPOR$.

The following lemma states the key properties of $\ExtendPO$.

\smallskip
\begin{restatable}{lemma}{lemextendcorrectness}\label{lem:extend_correctness}
Let $\PartialOrders=\ExtendPO(\AnnotatedPO, X', \Value', \SideAnnotation', \GoodWrites')$.
Then $\ExtendPO$ runs in $O(m\cdot \Poly(n))$ time, where $n=|X'|$ and $m=|\PartialOrders|+1$, and the following assertions hold.
\begin{compactenum}
\item\label{item:extension1} Every annotated partial order $\AnnotatedPOK_i$ is closed and minimal.
\item\label{item:extension2} For every pair $K_i, K_j$, we have that $K_i\not \LeafRefines K_j$ and $K_j\not \LeafRefines K_i$.
\item\label{item:extension3} For every trace $\Trace$ such that 
(i)~$\Events{\Trace}=X'$, 
(ii)~for each read event $\Read\in \Reads{\Trace}$ we have $\Observation_{\Trace}(\Read)\in \GoodWrites'$ and
(iii)~$(\Trace\Project X) \LeafRefines P$,
there exists an annotated partial order $\AnnotatedPOK_i$ such that $\Trace$ is a linearization of $\AnnotatedPOK_i$.
\end{compactenum}
\end{restatable}
\begin{proof}
\noindent{\em Correctness.}
We first argue about the correctness of the algorithm, i.e., the assertions in \cref{item:extension1}-\cref{item:extension3} above.
\begin{compactenum}
\item This assertion is an immediate consequence of the facts that
(i)~$\AnnotatedPO$ is closed and minimal,
(ii)~$\ExtendPO$ constructs each annotated partial order simply by ordering conflicting events that belong to the leaf threads, and
(iii)~the closure of a minimal annotated partial order is also minimal.
\item This assertion holds trivially by construction.
\item Since $\AnnotatedPO$ is minimal, $\ExtendPO$ creates an annotated partial order $\AnnotatedPOQ=(X'_1, X'_2, Q, \Value', \SideAnnotation', \GoodWrites')$ such that $\Trace\LeafRefines Q$ and $\AnnotatedPOQ$ is also minimal.
Observe that $\AnnotatedPOQ$ is feasible, since $\Trace\Refines Q$ and $(X'_1, X'_2, \Trace, \Value', \SideAnnotation' \GoodWrites')$ is closed. Thus the algorithm will construct $\AnnotatedPOK_i=\ClosureAlgo(\AnnotatedPOQ)$ and include $\AnnotatedPOK_i$ in $\PartialOrders$.
\end{compactenum}
\noindent{\em Complexity.}
Since the number of threads is constant, for every recursive call of $\ExtendPO$, \cref{item:extend_step2} of the algorithm creates $O(\Poly(n))$ partial orders $K_i$, and since computing the closure of $K_i$ requires $O(\Poly(n))$ time, we have that $\ExtendPO$ spends $O(\Poly(n))$ in each recursive call.
It follows that constructing the whole set $\PartialOrders$ takes $O(m\cdot \Poly(n))$ time,
since $m$ is the size of the output (i.e., the number of leaves in the recursion) and every recursive step takes $O(\Poly(n))$ time.

\end{proof}


\smallskip\noindent{\bf Correctness and complexity of $\VCDPOR$.}
We now turn our attention to the correctness and complexity properties of $\VCDPOR$.
The proof concepts rely on the tree $T$ induced by the recursive calls of $\VCDPOR$.
We start with introducing the tree $T$ and proceed with the correctness and complexity statements of $\VCDPOR$.

\smallskip\noindent{\em The induced tree $T$.}
An execution of $\VCDPOR$ induces a tree $T$, where each node $u$ is labeled with an annotated partial order $\AnnotatedPO^u$ constructed at some recursive step by the algorithm.
We have two types of nodes.
\begin{compactenum}
\item A type~1 node $u$ is labeled with an annotated partial order $\AnnotatedPO^u$ such that $\AnnotatedPO^u$ was passed as an argument to a recursive call of $\VCDPOR$. These nodes correspond to all annotated partial orders returned by the algorithm $\ExtendPO$ when invoked from within $\MutateRoot$ or $\MutateLeaf$.
\item  A type~2 node $u$ is labeled with with an annotated partial order $\AnnotatedPO^u$ such that $\AnnotatedPO^u$ was not passed as an argument to $\VCDPOR$. These nodes correspond to all annotated partial orders returned by the algorithm $\ExtendPO$ when invoked from within $\VCDPOR$.
\end{compactenum}
We will use the induced tree $T$ to reason about the correctness and complexity of $\VCDPOR$.

\smallskip
\begin{remark}\label{rem:minimal}
For every node $u$, the annotated partial order $\AnnotatedPO^u$ is closed and minimal.
\end{remark}

\smallskip\noindent{\em Correctness.}
We first turn our attention to the correctness of $\VCDPOR$.
We will argue that for every target trace $\Trace^*$, the algorithm discovers the value function $\Value_{\Trace^*}$.
In particular, the induced tree $T$ has a node $u$ such that $\AnnotatedPO^u$ is of the form $\AnnotatedPO^u=(X_1, X_2, P, \Value_{\Trace^*}, \SideAnnotation, \GoodWrites )$, i.e., the value function of $\AnnotatedPO^u$ is the value function of the target trace $\Trace^*$.
In our discussion below, we fix such a target $\Trace^*$ and introduce some notation around it.

\smallskip\noindent{\em Compatible and witness nodes.}
An annotated partial order $\AnnotatedPO=(X_1, X_2, P, \Value, \SideAnnotation, \GoodWrites)$ is called \emph{compatible} with $\Trace^*$ if the following conditions hold. Let $X=X_1\cup X_2$.
\begin{compactenum}
\item $X\subseteq \Events{\Trace^*}$, $\Value\subseteq \Value_{\Trace^*}$, $\SideAnnotation\subseteq \SideAnnotation_{\Trace^*}$ and $(\Trace^*\Project X)\LeafRefines P$.
\item For every read event $\Read\in \Reads{X}$ we have that $\Observation_{\Trace^*}(\Read)\in \GoodWrites(\Read)$.
\end{compactenum}
A node $u$ of the induced tree $T$ is called \emph{compatible} with $\Trace^*$ if $\AnnotatedPO^u$ is compatible with $\Trace^*$.
We call $u$ a \emph{witness}  if $\Value^u=\Value_{\Trace}$, where $\Value^u$ is the value function of the annotated partial order $\AnnotatedPO^u$.
\smallskip
\begin{remark}\label{rem:compatible_ancestors}
If $u$ is compatible with $\Trace^*$ then every ancestor of $u$ is also compatible with $\Trace^*$.
\end{remark}

\smallskip\noindent{\em Left and leftmost movers.}
Consider a node $u$ of the induced tree $T$ such that $u$ is compatible with $\Trace^*$.
A child $z$ of $u$ in $T$ is called  a \emph{left mover} if 
\begin{compactenum}
\item $z$ is compatible with $\Trace^*$ and
\item $z$ is the first child of $u$ with this property, in the order the execution of $\VCDPOR$.
\end{compactenum}
We call $u$ a \emph{leftmost mover} if $u$ and every ancestor of $u$ (except for the root of $T$) is a left mover.
The correctness of $\VCDPOR$ is based on the following lemma.

\smallskip
\begin{restatable}{lemma}{lemleftmove}\label{lem:left_move}
If $u$ is a leftmost mover then either $u$ is a witness or $u$ has a child that is a leftmost mover.
\end{restatable}
\begin{proof}
Assume that $u$ is not a witness and we argue that $u$ has a child $z$ such that $z$ is compatible with $\Trace^*$.
Since $u$ is a leftmost mover, it will follow that $u$ has a child that is a leftmost mover.
We split cases based on whether $u$ is a type 1 or type 2 node.

\noindent{\em The node $u$ is a type 1 node.}
By \cref{rem:minimal}, $\AnnotatedPO^u=(X_1, X_2, P, \Value, \SideAnnotation, \GoodWrites)$ is a minimal, closed annotated partial order.
Consider the trace $\Trace$ constructed by $\VCDPOR$ in \cref{line:vcdpor_wextend}, and observe that
$\Events{\Trace}\subseteq \Events{\Trace^*}$, $\Value_{\Trace}\subseteq \Value_{\Trace^*}$ and $\SideAnnotation_{\Trace}\subseteq \SideAnnotation_{\Trace^*}$.
Consider the trace $\ov{\Trace} = \Trace^*\Project \Events{\Trace}$, and observe that
(i)~for every read event $\Read\in \Reads{\ov{\Trace}}$ we have $\Observation_{\ov{\Trace}}(\Read)\in \GoodWrites(\Read)$, and
(ii)~$\ov{\Trace}\LeafRefines P$.
By \cref{lem:extend_correctness}, $\ExtendPO$ in \cref{line:vcdpor_extend} returns an annotated partial order $\AnnotatedPOK_i$ such that $\ov{\Trace}$ is a linearization of $\AnnotatedPOK_i$.
We associate $z$ with $\AnnotatedPOK_i$.

\noindent{\em The node $u$ is a type 2 node.}
Consider any linearization $\Trace$ of $\AnnotatedPO^u=(X_1, X_2, P, \Value, \SideAnnotation, \GoodWrites)$, and since $u$ is compatible with $\Trace^*$, for every read event $\Read$ that is enabled in $\Trace$ we have that $\Read\in \Events{\Trace^*}$.
In addition, there exists a read event $\Read$ that is enabled in $\Trace$ and $\Observation_{\Trace^*}(\Read)\in \Events{\Trace}$.
Let $\Write=\Observation_{\Trace^*}(\Read)$, and we argue that $\Write\in\CandidateSet_{\Trace}^{\NegativeAnnotation^u}(\Read)$.
We distinguish between the following cases.
\begin{compactenum}
\item $\NegativeAnnotation^u=\bbot$. Then by definition, $\Write\in\CandidateSet_{\Trace}^{\NegativeAnnotation^u}(\Read)$.
\item $\NegativeAnnotation^u=\bot$ or $\NegativeAnnotation^u\in \SysReads$.
Then, there exists a type 2 ancestor $q$ of $u$ and a trace $\Trace_q$ that is a linearization of $\AnnotatedPO^q=(X'_1, X'_2, P', \Value', \SideAnnotation', \GoodWrites')$, $\Read$ is enabled in $\Trace^q$ and $\Write \in \CandidateSet_{\Trace^q}^{\NegativeAnnotation^q}(\Read)$.
It is straightforward to see that at that point the algorithm extended $\AnnotatedPO^q$ with $\Read$ and a good-writes function $\GoodWrites^q$ such that $\Write\in \GoodWrites^q(\Read)$.
A similar analysis as in the previous item shows that $\ExtendPO$ returned an annotated partial order $\AnnotatedPO'$ that is compatible.
In addition, $\AnnotatedPO'$ is associated with a node of $T$ that is a child of $q$ and that was visited before the ancestor of $u$ which is also a child of $q$ .
This contradicts the fact that $u$ is a leftmost mover.
It follows that $\Write\in\CandidateSet_{\Trace}^{\NegativeAnnotation^u}(\Read)$.
The rest follows by \cref{lem:extend_correctness}, similar to the previous case.
\end{compactenum}

The desired result follows.
\end{proof}

\smallskip
\begin{restatable}{lemma}{lemvcdporoptimal}\label{lem:vcdpor_optimal}
For every pair of traces $\Trace'_1, \Trace'_2$ constructed by $\VCDPOR$ in \cref{line:vcdpor_realize} in two recursive calls,
we have that $\Trace'_1\not \VHBE \Trace'_2$.
\end{restatable}
\begin{proof}
Consider the nodes $u_1, u_2$ of the induced tree $T$ that correspond to the recursive calls in which $\VCDPOR$ constructed the traces $\Trace'_1$ and $\Trace'_2$, respectively.
If $u_i$ is ancestor of $u_{3-i}$, for some $i\in[2]$, then clearly $\Events{\Trace'_i}\neq \Events{\Trace'_{3-i}}$.
Otherwise, let $u$ be the lowest common ancestor of $u_1$ and $u_2$ in $T$.
For each $i\in[2]$, let $z_i$ be the child of $u$ that is also an ancestor of $u_i$, and let $\AnnotatedPO^{z_i}=(X^i_1, X'_2, P^i, \Value^i, \SideAnnotation^i, \GoodWrites^i)$ be the annotated partial order that labels node $z_i$.
We distinguish between the following cases.
\begin{compactenum}
\item If $z$ is a type~2 node, then $\AnnotatedPO^{z_i}$ only differ on $P^i$.
By \cref{lem:extend_correctness},  there exists a pair of events $\Event_1, \Event_2\in X^1_2$ such that
(i)~$\Confl{\Event_1}{\Event_2}$ and
(ii)~$\HB{\Event_1}{P^1}{\Event_2}$ and $\HB{\Event_2}{P^2}{\Event_1}$.
It follows that $\HB{\Event_1}{\Trace'_1}{\Event_2}$ and $\HB{\Event_2}{\Trace'_2}{\Event_1}$, and since $\Event_1, \Event_2\in \SysLeafEvents$, we have that $\Trace'_1\not \VHBE\Trace'_2$.
\item If $z$ is a type~1 node, let $\Trace'$ be the trace constructed in \cref{line:vcdpor_realize}  by the recursive call to $\VCDPOR$ for node $z$.
We distinguish between the following cases.
\begin{compactenum}
\item If $\AnnotatedPO^{z_1}$ and $\AnnotatedPO^{z_2}$ occur from the same invocation to $\ExtendPO$, then the proof is similar to the previous item.
\item If $\AnnotatedPO^{z_1}$ and $\AnnotatedPO^{z_2}$ occur from different invocations to $\ExtendPO$, we examine whether both $\AnnotatedPO^{z_1}$ and $\AnnotatedPO^{z_2}$ were constructed by extending to the same read event $\Read$ or not.
In the former case, we examine the values $\Value^1(\Read)$ and $\Value^2(\Read)$ that $\Read$ was forced to read.
If $\Value^1(\Read)\neq \Value^2(\Read)$ then $\Value_{\Trace'_1}(\Read)\neq \Value_{\Trace'_2}$, 
whereas if $\Value^1(\Read)=\Value^2(\Read)$ then $\Proc{\Read}=\RootProcess$ and $\SideAnnotation^1(\Read)\neq \SideAnnotation^2(\Read)$ and thus $\SideAnnotation_{\Trace'_1}(\Read)\neq \SideAnnotation_{\Trace'_2}(\Read)$.
We are left with the case where $\AnnotatedPO^{z_1}$ and $\AnnotatedPO^{z_2}$ were constructed by extending to two different read events $\Read_1$ and $\Read_2$, respectively.
Assume wlog that $\AnnotatedPO^{z_2}$ was constructed after $\AnnotatedPO^{z_1}$.
If $\Value_{\Trace'_2}(\Read_1)\neq \Value_{\Trace'_1}(\Read_1)$ we are done.
Otherwise, let $\Read=\GuardingRead_{\Trace'_2}(\Observation_{\Trace}(\Read_1))$ and note that $\CHB{\Read}{\Trace'_2}{\Read_1}$.
Due to the causally-happens-before map $\NegativeAnnotation$ in that recursive call of $\VCDPOR$, we have that $\Read\not \in \Events{\Trace'}$ 
and thus $\NCHB{\Read}{\Trace'_1}{\Read_1}$.
\end{compactenum}
\end{compactenum}

In all cases, we have $\Trace'_1\not \VHBE \Trace'_2$, as desired.
\end{proof}

\smallskip
\begin{restatable}{lemma}{lemvcdporcomplexity}\label{lem:vcdpor_complexity}
$\VCDPOR$ runs in time $O\left(|\TraceSpaceMax/\VHB|\cdot \Poly(n)\right)$, where $n$ is the length of the longest trace in $\TraceSpaceMax$.
\end{restatable}
\begin{proof}
Consider two maximal traces $\Trace_1, \Trace_2\in \TraceSpaceMax$ such that $\Trace_1\VHBE\Trace_2$.
Let $\Trace'_1$, $\Trace'_2$  be prefixes of $\Trace_1$, $\Trace_2$, respectively, such that $\Events{\Trace'_1}=\Events{\Trace'_2}$,
and observe that $\Trace'_1\VHBE \Trace'_2$.
Since we have constantly many threads, it follows that given a maximal trace $\Trace$, there exist $O(\Poly(n))$ different sets $X\subseteq \Events{\Trace}$ for which there exists a trace $\Trace'$ such that (i)~$\Events{\Trace'}=X$ and (ii)~$\Trace$ is a maximal extension of $\Trace'$.
It follows that $|\TraceSpace/\VHB| = O\left(|\TraceSpaceMax/\VHB|\cdot \Poly(n)\right)$, and thus 
it suffices to argue that $\VCDPOR$ runs in time  $O\left(|\TraceSpace/\VHB|\cdot \Poly(n)\right)$.
By \cref{lem:vcdpor_optimal}, for every pair of traces $\Trace'_1$ and $\Trace'_2$ constructed by $\VCDPOR$ in \cref{line:vcdpor_realize},
we have that $\Trace'_1\not \VHBE \Trace'_2$, and hence each such trace falls into a different class of $\TraceSpace/\VHB$.
Thus the size of the induced tree $T$ is bounded by $|\TraceSpace/\VHB|$.
For every internal node $u$ of $T$, the children of $u$ in $T$ are produced by $O(\Poly(n))$ calls to $\ExtendPO$, which, by \cref{lem:extend_correctness}, requires $O(\Poly(n))$ time per child of $u$.
Hence the total time spent by $\VCDPOR$ is
\[
O(|T|\cdot \Poly(n))=O\left(|\TraceSpace/\VHB|\cdot \Poly(n)\right)= O\left(|\TraceSpaceMax/\VHB|\cdot \Poly(n)\right) .
\] 
The desired result follows.

\end{proof}

\section{Details of \cref{sec:experiments}}\label{sec:appendix_experiments}

\smallskip\noindent{\bf Identifying events.}
In the implementation we rely on the Nidhugg model-checker to identify events. 
An event $\Event$ is defined by a pair $(a_{\Event},b_{\Event})$, where $a_{\Event}$ is its thread-id and $b_{\Event}$ is the sequential number of the last LLVM instruction (of the corresponding thread) that is part of $\Event$. 
Note that this way, there can exist two traces $\Trace_1, \Trace_2$ and two different events $\Event_i\in\Trace_i$, for $i\in[2]$ such that $a_{\Event_1}=a_{\Event_2}$ and $b_{\Event_1}=b_{\Event_2}$, i.e., the two events have the same id in their respective traces.
However, this means that the control-flow leading to each event is different. 
In this case $\Trace_1$ and $\Trace_2$ differ in the value of a common event, and hence are treated as inequivalent.

\smallskip\noindent{\bf Root thread and order of threads for extension.}
In our presentation, given a concurrent program $\System=\{ \Process_i \}_{i=1}^k$ of $k$ threads, we always distinguish
$\Process_1$ as the root thread of $\System$. 
In our experiments, we choose $\Process_1$ to be the second thread of the program.
We note the choice of the root thread does not affect the soundness of our approach.



\smallskip\noindent{\bf Optimizations.}
Here we briefly report on three straightforward optimizations we have made in our implementation, namely
\begin{compactenum}
\item Choosing the order of reads to extend the annotated partial order.
\item Extensions yielding maximal traces.
\item Reduction in the number of annotated partial orders returned by $\ExtendPO$.
\end{compactenum}

\smallskip\noindent{\em 1. Choosing the order of reads to extend the annotated partial order.}
In $\VCDPOR$~\cref{algo:vcdpor}, given an extension $\AnnotatedPOQ$, we first call $\MutateRoot$,
and then for each leaf thread we call its corresponding $\MutateLeaf$. However, the order in which we call $\MutateRoot$ and
different $\MutateLeaf$ can have an effect on the shape of the induced recursion tree. 
In our experiments, we fix this order by giving higher priority to read events that are succeeded by some write event in their local thread.
Ties are broken arbitrarily.

\smallskip\noindent{\em 2. Extensions yielding maximal traces.}
Consider a call of $\VCDPOR$~\cref{algo:vcdpor} on a node $u$ of the induced recursion tree, 
In this call, an annotated partial order $\AnnotatedPOQ$ will be (attempted to) be extended with a read event $\Read$ to observe a value $v$.
If this extension is successful and results in a maximal trace, we do not attempt to extend annotated partial orders that correspond to siblings of $u$ with $\Read$ observing value $v$.

\smallskip\noindent{\em 3. Reduction in the number of annotated partial orders returned by $\ExtendPO$.}
In our presentation of $\ExtendPO(\AnnotatedPO, X', \Value', \SideAnnotation', \GoodWrites')$,
given $X'\setminus X=\{\Event \}$ such that $\Event$ belongs to a leaf thread, we consider all possible orderings of $\Event$ with conflicting
events from all leaf threads. However, in our implementation, we relax this in two ways. Given a write event $\Event_{\Write}$, we say it is
\emph{never-good} if it does not belong to $\GoodWrites'(\Read)$ for any read event $\Read$. Further, given $\Event_{\Write}$ and an annotated
partial order $\AnnotatedPOK$, we say that $\Event_{\Write}$ is \emph{unobservable} in $\AnnotatedPOK$, if for every linearization of $\AnnotatedPOK$ there is no
read event such that $\Event_{\Write}$ is its observation. 
Given two unordered conflicting write events from leaf threads, we do not order them if
(i) both of them are never-good, or (ii) at least one of them is unobservable. 

\smallskip\noindent{\bf Technical details.}
For our experiments we have used a Linux machine with Intel(R) Xeon(R) CPU E5-1650 v3 @ 3.50GHz (12 CPUs) and 128GB of RAM.
We have run Nidhugg with Clang and LLVM version 3.8.

\smallskip\noindent{\bf SV-COMP benchmark modifications.}
We have made small changes to some of the SV-COMP benchmarks so they can be processed by our prototype implementation in Nidhugg:
\begin{compactitem}
\item Verifier calls to perform acquire and release are handled by a \texttt{pthread\_mutex}.
\item Verifier calls to nondeterministically produce an arbitrary integer are replaced by a constant value.
\end{compactitem}
Further, we have made modifications so that the examples can be used for our experiments:
\begin{compactitem}
\item In critical section benchmarks, the thread routines are put in a loop with scalable size, so threads can reenter a critical section multiple times.
\item In benchmarks that contain an assertion violation, we replace the assertion with a read, so that we can measure the actual efficiency of
the trace space exploration. Note that for each benchmark with a violation, we first used the benchmark unchanged, and we
state that all algorithms considered in our experiments successfully caught the violation and reported a corresponding error trace.
\item We manually perform loop unrolling, i.e., we limit the amount of times each loop is executed by a scalable bound, instead of relying
on the loop bounding technique provided by Nidhugg.
\end{compactitem}

Finally, our $\VCDPOR$ implementation assumes that all global variables are initialized with value 0. Therefore, in benchmarks that contain
a different initialization, we put writes of the corresponding value to the corresponding variable, before any threads are spawned.
Note that this modification does not change the verification problem. We report this fact in case our implementation is used in different future works.

\end{document}